\documentclass[12pt, twoside]{book}
\usepackage[a4paper, width=150mm, top=25mm, bottom=25mm, bindingoffset=6mm]{geometry}
\usepackage{emptypage}
\usepackage{setspace}
\usepackage{framed}
\usepackage{fancyhdr}
\usepackage[T1]{fontenc}
\usepackage{verbatim}
\usepackage{xcolor}
\usepackage{graphicx}
\usepackage{tikz}
\tikzset{>=stealth}
\definecolor{darkred}{RGB}{200, 0, 0}
\definecolor{darkblue}{RGB}{0, 0, 180}
\definecolor{lightgray}{RGB}{230, 230, 230}
\definecolor{darkgray}{RGB}{210, 210, 210}
\definecolor{verydarkgray}{RGB}{150, 150, 150}
\definecolor{mycolour1}{HTML}{1717FF}
\definecolor{mycolour2}{HTML}{009400}
\definecolor{mycolour3}{HTML}{C52000}
\usepackage{hyperref}
\hypersetup{colorlinks, breaklinks, linkcolor=darkred, urlcolor=darkblue, citecolor=darkblue}
\usepackage{amssymb, amsmath, amsthm}
\usepackage{bbold}
\usepackage[mathscr]{euscript}
\makeatletter
\def\amsbb{\use@mathgroup \M@U \symAMSb}
\renewcommand\paragraph{
   \@startsection{paragraph}{4}{0mm}
      {-\baselineskip}
      {.1\baselineskip}
      {\normalfont\normalsize\bfseries}}
\makeatother
\newcommand{\paperA}
{
\begin{itemize}
\item \textbf{Secure bit commitment from relativistic constraints} \arxiv{1206.1740}\\
J.~Kaniewski, M.~Tomamichel, E.~H\"{a}nggi and S.~Wehner\\
\textsl{IEEE Transactions on Information Theory} \textbf{59}, 7 (2013).\\
(presented at \textsl{QCrypt '12})
\end{itemize}
}
\newcommand{\paperB}
{
\begin{itemize}
\item \textbf{Experimental bit commitment based on quantum communication and special relativity} \arxiv{1306.4801}\\
T. Lunghi, J. Kaniewski, F. Bussi\`{e}res, R. Houlmann, M. Tomamichel, A. Kent, N. Gisin, S. Wehner and H. Zbinden\\
\textsl{Physical Review Letters} \textbf{111}, 180504 (2013).\\
(presented at \textsl{QCrypt '13})
\end{itemize}
}
\newcommand{\paperC}
{
\begin{itemize}
\item \textbf{Practical relativistic bit commitment} \arxiv{1411.4917}\\
T. Lunghi, J. Kaniewski, F. Bussi\`{e}res, R. Houlmann, M. Tomamichel, S. Wehner and H. Zbinden\\
\textsl{Physical Review Letters} \textbf{115}, 030502 (2015).\\
(presented at \textsl{QCrypt '14})
\end{itemize}
}

\newcommand{\timeaxis}[1]
{
	\draw [->] (0, -0.5) to (0, #1);
	\draw [-, thick] (-0.1, 0) to (0.1, 0);
	\node[left, color=mycolour3] at (-0.1, 0) {$t = 0$};
	\node[left] at (-0.1, #1) {$t$};
}
\newcommand{\arxiv}[1]{[\texttt{arXiv:\href{http://arxiv.org/abs/#1}{#1}}]}
\newcommand{\horrule}{\rule{\linewidth}{0.5mm}}
\newcommand{\alice}[1]{Alice$_{#1}$}
\newcommand{\bob}[1]{Bob$_{#1}$}
\newcommand{\prover}[1]{Prover$_{#1}$}
\newcommand{\cval}{d}
\newcommand{\nbox}[2][8]{\hspace{#1pt} \mbox{#2} \hspace{#1pt}}

\newcommand{\chshn}{\textnormal{CHSH}_{n}}
\newcommand{\bs}[1]{\{0, 1\}^{#1}}
\newcommand{\bsr}[1]{\in \{0, 1\}^{#1}}
\renewcommand{\bsr}[1]{}

\newtheoremstyle{mytheoremstyle}	
		{\topsep}											
		{\topsep}											
		{\sl}													
		{}														
		{\scshape}										
		{.}														
		{.5em}												
		{}														
\theoremstyle{mytheoremstyle}
\newtheorem{thm}{Theorem}[chapter]
\newtheorem{df}{Definition}[chapter]
\newtheorem{lem}{Lemma}[chapter]
\newtheorem{prop}{Proposition}[chapter]
\newtheorem{obs}{Observation}[chapter]
\def \diracspacing {0.7pt}
\newcommand{\bra}[1]{\langle #1 \hspace{\diracspacing} |} 
\newcommand{\ket}[1]{| \hspace{\diracspacing} #1 \rangle} 
\newcommand{\braket}[2]{\langle #1 \hspace{\diracspacing} | \hspace{\diracspacing} #2 \rangle} 
\newcommand{\braketq}[1]{\braket{#1}{#1}} 
\newcommand{\ketbra}[2]{| \hspace{\diracspacing} #1 \rangle \langle #2 \hspace{\diracspacing} |} 
\newcommand{\ketbraq}[1]{\ketbra{#1}{#1}} 
\newcommand{\bramatket}[3]{\langle #1 \hspace{\diracspacing} | #2 | \hspace{\diracspacing} #3 \rangle} 
\newcommand{\bramatketq}[2]{\bramatket{#1}{#2}{#1}} 
\DeclareMathOperator{\tr}{tr}
\DeclareMathOperator{\id}{id}
\DeclareMathOperator{\wham}{w_{H}}
\DeclareMathOperator{\dham}{d_{H}}
\DeclareMathOperator{\err}{err}
\newcommand{\norm}[2][]{#1| \! #1| #2 #1| \! #1|}
\newcommand{\abs}[2][]{#1| #2 #1|}
\newcommand{\ave}[1]{\langle #1 \rangle}
\newcommand{\tran}[0]{^\textnormal{\tiny{T}}}
\newcommand{\pwin}{p_{\textnormal{win}}}

\newcommand{\cG}{\mathcal{G}}
\newcommand{\cH}{\mathcal{H}}

\newcommand{\cL}{\mathcal{L}}
\newcommand{\cM}{\mathcal{M}}

\newcommand{\cP}{\mathcal{P}}

\newcommand{\cS}{\mathcal{S}}
\newcommand{\cT}{\mathcal{T}}

\newcommand{\cX}{\mathcal{X}}
\newcommand{\cY}{\mathcal{Y}}
\newcommand{\cZ}{\mathcal{Z}}

\newcommand{\sH}{\mathscr{H}}

\newcommand{\hordots}[4]
{
	\draw [fill, color=#4] (#1 + 0.5, #2) circle [radius=#3];
	\draw [fill, color=#4] (#1, #2) circle [radius=#3];
	\draw [fill, color=#4] (#1 - 0.5, #2) circle [radius=#3];
}
\renewcommand{\th}
{
^{\textnormal{th}}
}
\newcommand{\phases}[1]
{
\draw[dashed, very thick, red] (3, 1.2) -- (3, #1);
\node[red] at (3, 1.4) {commitment point};
\draw[dashed, very thick, red] (6.2, 1.2) -- (6.2, #1);
\node[red] at (6.2, #1 - 0.2) {opening point};
\draw[fill=white] (0, 0) rectangle (2.8, 0.8);
\node (commit) at (1.4, 0.4) {1. commit};
\draw[fill=white] (3.2, 0) rectangle (6, 0.8);
\node (wait) at (4.6, 0.4) {2. sustain};
\draw[fill=white] (6.4, 0) rectangle (9.2, 0.8);
\node (open) at (7.8, 0.4) {3. open};
\draw[fill=white] (9.6, 0) rectangle (12.4, 0.8);
\node (verify) at (11, 0.4) {4. verify};
}
\newenvironment{prot}[2]
{
\begin{framed}
\noindent \textbf{Protocol~#1:}\ {\texttt{#2}}\\
}
{
\end{framed}
}
\newcommand{\intersectM}
{
\ifbackreporting
\cap \cM
\fi
}
\newcommand{\tmoprot}
{
\ifbackreporting
	\begin{prot}{8}{Bit commitment by transmitting measurement outcomes with backreporting}
	\label{prot:tmo-backreporting}
\else
	\begin{prot}{7}{Bit commitment by transmitting measurement outcomes}
	\label{prot:tmo}
\fi
\begin{enumerate}
\item (commit) At $t = 0$, \bob{0} chooses $x, \theta \in \{0, 1\}^{n}$ uniformly at random, creates $\ket{x^{\theta}}$ and sends it to \alice{0}. \alice{0} measures all the incoming qubits in the same basis (computational if $\cval = 0$ and Hadamard if $\cval = 1$)
\ifbackreporting
. The rounds in which a click was observed form $\cM$ and $y \in \{0, 1\}^{m}$ is the string of outcomes. \alice{0} announces $\cM$ to \bob{0}. \bob{0} continues with the protocol only if $m \geq \gamma n$.
\else
to produce $y \in \{0, 1\}^{n}$ (the string of measurement outcomes), which she then sends to \alice{1} and \alice{2}.
\fi
\item (open) At $t = 1$, \alice{1} and \alice{2} simultaneously send $d$ and $y$ to \bob{1} and \bob{2}, respectively.
\item (verify) \bob{1} and \bob{2} pass all the information to \bob{0}, who verifies that:
\begin{itemize}
\item \alice{1} and \alice{2} have attempted to unveil the same value
\item \alice{1} and \alice{2} have provided exactly the same string $y$
\item the string $y$ is consistent with the BB84 states initially prepared by \bob{0} up to the error threshold $\delta$
\begin{align*}
\dham( x_{S\intersectM}, y_{S\intersectM} ) \leq \delta &\nbox{for} d = 0, \hspace{1.5cm}\\
\dham( x_{T\intersectM}, y_{T\intersectM} ) \leq \delta &\nbox{for} d = 1.
\end{align*}
\end{itemize}
If all three conditions are satisfied, \bob{0} accepts the commitment.
\end{enumerate}
\end{prot}
}
\newcommand{\secretsharingnc}
{
\begin{prot}{2}{Bit commitment from secret sharing}
\label{prot:secret-sharing-nc}
\begin{enumerate}
\item (commit) Alice generates a random bit $a \in \{0, 1\}$, sends $x_{1} = d \oplus a$ to \bob{1} and $x_{2} = a$ to \bob{2}.
\item (open and verify) \bob{1} and \bob{2} get together and compute the commitment as $d = x_{1} \oplus x_{2}$.
\end{enumerate}
\end{prot}
}
\newcommand{\secretsharingrel}
{
\begin{prot}{5}{Bit commitment from secret sharing (relativistic)}
\label{prot:secret-sharing-rel}
\begin{enumerate}
\item (commit) At $t = 0$, \alice{1} sends $x_{1} = d \oplus a$ to \bob{1} and \alice{2} sends $x_{2} = a$ to \bob{2}. \bob{1} and \bob{2} immediately send $x_{1}$ and $x_{2}$ to \bob{0}.
\item (open and verify) At $t = 1$, \bob{0} receives $x_{1}$ and $x_{2}$ and computes the commitment as $d = x_{1} \oplus x_{2}$.
\end{enumerate}
\end{prot}
}
\newcommand{\sBGKWnc}
{
\begin{prot}{4}{Simplified-BGKW}
\begin{enumerate}
\item (commit) Bob sends $y_{1} = b$ to \alice{1} and she replies with $x_{1} = d \cdot y_{1} \oplus a$.
\item (open) \alice{2} reveals $x_{2} = a$ to Bob.
\item (verify) Bob verifies that $x_{1} \oplus x_{2} = d \cdot b$.
\end{enumerate}
\end{prot}
}
\newcommand{\sBGKWrel}[2]
{
\begin{prot}{6}{Simplified-BGKW (relativistic)}
\begin{enumerate}
\item (commit) At $t = 0$, \bob{1} sends $y_{1} = b$ to \alice{1} and she replies with $x_{1} = d \cdot y_{1} \oplus a$. \bob{1} immediately sends $x_{1}$ to \bob{2}.
\item (open) At $#1$, \alice{2} reveals $x_{2} = a$ to \bob{2}.
\item (verify) At $t #2$, \bob{2} receives $x_{1}$ and verifies that $x_{1} \oplus x_{2} = d \cdot b$.
\end{enumerate}
\end{prot}
}
\newcommand{\distributedOT}
{
\begin{prot}{1}{Distributed oblivious transfer}
\label{prot:distributed-OT}
\begin{enumerate}
\item (prepare) Bob generates an $n$-bit string $r \in \bs{n}$ uniformly at random and sends $(u_{0}, u_{1}) = (m_{0} \oplus r, m_{1} \oplus r )$ to \bob{1} and $(v_{0}, v_{1}) = (r, m_{0} \oplus m_{1} \oplus r )$ to \bob{2}.
\item (execute) Alice chooses a random bit $c \in \{0, 1\}$ and requests $u_{c}$ from \bob{1}. To retrieve $m_{d}$ she requests $v_{d \oplus c}$ from \bob{2} and computes the message as $m_{d} = u_{c} \oplus v_{d \oplus c}$.
\end{enumerate}
\end{prot}
}
\renewcommand{\arraystretch}{1.3}
\usepackage{makeidx}
\makeindex
\newif\ifbackreporting
\begin{document}
\begin{titlepage} 
\begin{center}
\horrule \\[0.8cm] 
{\huge \bfseries Relativistic quantum cryptography}\\[0.25cm] 
\horrule \\[1.5cm] 

\vskip 3cm
\textbf{J\k{E}DRZEJ KANIEWSKI}\\
\textit{(MMath, University of Cambridge)}
\vskip 3cm

\large \textit{A thesis submitted in fulfilment of the requirements\\ for the degree of Doctor of Philosophy}\\[0.3cm] 
\textit{in the}\\[0.4cm]
Centre for Quantum Technologies\\ National University of Singapore\\[1cm] 
\vskip 2cm
\begin{figure}[h!]
\centering
\includegraphics[width = 0.6 \columnwidth]{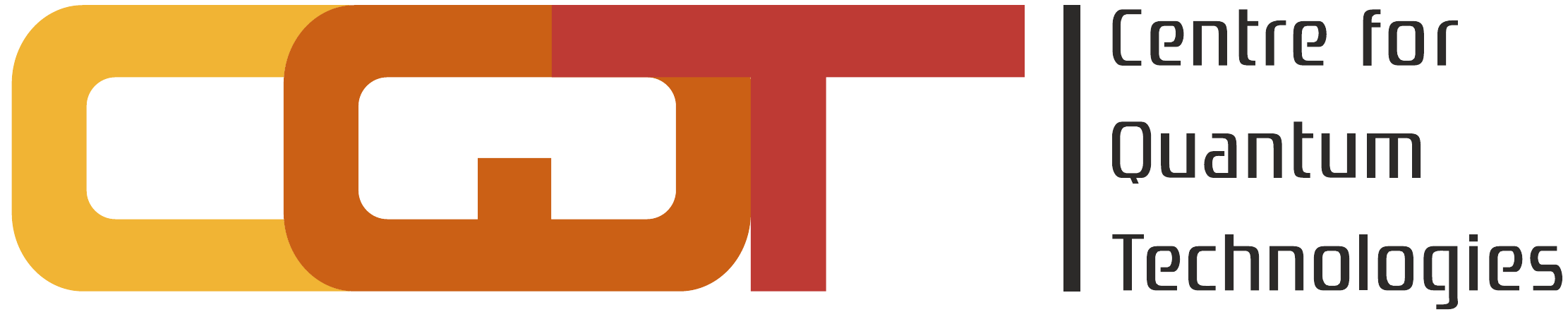}
\end{figure}
{\large 2015} 
\vfill
\end{center}
\end{titlepage}
\frontmatter
\chapter*{Declaration}

I hereby declare that this thesis is my original work and has been written by me in its entirety. I have duly acknowledged all the sources of information which have been used in the thesis.\\

\noindent This thesis has also not been submitted for any degree in any university previously.

\vspace{0.4cm}
\begin{figure}[!h]
\centering
\includegraphics{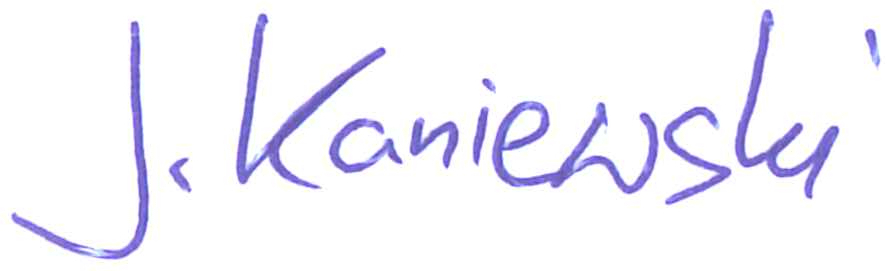}
\vspace{-1.75cm}
\end{figure}

\begin{center}
\rule{6cm}{0.5pt} \\
{ \bf J\k{e}drzej Kaniewski}
\vskip 2cm
9 September 2015
\end{center}
\chapter*{Acknowledgements}
I would like to thank my supervisor, Stephanie Wehner, for the opportunity to conduct a PhD in quantum information. I am grateful for her time, effort and resources invested in my education. Working with her and being part of her active and diverse research group made the last four years a great learning experience.

The fact that I was even able to apply for PhD positions is largely thanks to my brilliant and inspiring undergraduate supervisor, mentor and friend, Dr Peter Wothers MBE. I am particularly grateful for his supportive attitude when I decided to dedicate myself to quantum information. I am grateful to St.~Catharine's College for a wonderful university experience and several long-lasting friendships.

I would like to thank my collaborators, F{\'e}lix Bussi{\`e}res, Patrick J.~Coles, Serge Fehr, Nicolas Gisin, Esther H\"{a}nggi, Raphael Houlmann, Adrian Kent, Troy Lee, Tommaso Lunghi, Atul Mantri, Nathan McMahon, Gerard Milburn, Corsin Pfister, Robin Schmucker, Marco Tomamichel, Ronald de Wolf and Hugo Zbinden, who made research enjoyable and from whom I have learnt a lot.

I am also indebted to Corsin Pfister and Le Phuc Thinh, who have read a preliminary version of this thesis, and Tommaso Lunghi and Laura Man\v{c}inska, who have given comments on parts of it.

Special thanks go to Evon Tan for being the omnipresent good spirit of CQT. Her incredible problem-solving skills allowed me to focus on research and contributed greatly to the scientific output of this thesis.

I would like to thank Valerio Scarani for being approachable and always happy to talk about various aspects of quantum information and the scientific world in general.

I am grateful to my examiners: Anne Broadbent, Marcin Paw{\l}owski and Miklos Santha for the careful reading of this thesis and providing stimulating feedback. I would like to thank Alexandre Roulet, Jamie Sikora, Marco Tomamichel and Marek Wajs for useful comments on the defence presentation.

Dzi\k{e}kuj\k{e} Markowi Wajsowi za nieocenion\k{a} pomoc przy drukowaniu i sk{\l}adaniu doktoratu.

Arturowi Ekertowi chcia{\l}bym podzi\k{e}kowa\'{c} za czas, wsparcie i konkretne wskaz{\'o}wki w chwilach zw\k{a}tpienia.

Cho{\'c} to ju\.z par\k{e} lat chcia{\l}bym r{\'o}wnie{\.z} gor\k{a}co podzi\k{e}kowa{\'c} Krzysztofowi Ku{\'s}mierczykowi, Annie Mazurkiewicz i Bognie Luba{\'n}skiej za czas i wysi{\l}ek w{\l}o{\.z}ony w moj\k{a} edukacj\k{e} oraz za bycie {\'z}r{\'o}d{\l}em motywacji i inspiracji. Wszystko to, co uda{\l}o mi si\k{e} \mbox{osi\k{a}gn\k{a}{\'c}}, jest oparte na solidnych licealnych fundementach i bez ich wk{\l}adu nie by{\l}oby mo{\.z}liwe. Chc\k{e} tak{\.z}e podzi\k{e}kowa{\'c} Poniat{\'o}wce za niezapomniane trzy lata i wiele przyja{\'z}ni, kt{\'o}re trwaj\k{a} do dzisiaj.

Jackowi Jemielitemu chcia{\l}bym podzi\k{e}kowa{\'c} za pierwsze spotkanie z nauk\k{a} z prawdziwego zdarzenia, niespotykan\k{a} wytrwa{\l}o{\'s}{\'c} i cierpliwo{\'s}{\'c} a przede wszystkim za unikalne na skal\k{e} {\'s}wiatow\k{a} poczucie humoru, kt{\'o}rego cz\k{e}sto mi brakuje.

Doktorat dedykuj\k{e} w ca{\l}o{\'s}ci Mamie, Tacie, Siostrze i Bratu, bez wsparcia kt{\'o}rych to prze{\l}omowe dzie{\l}o nigdy by nie powsta{\l}o.
\chapter*{Abstract}
Special relativity states that information cannot travel faster than the speed of light, which means that communication between agents occupying distinct locations incurs some minimal delay. Alternatively, we can see it as temporary communication constraints between distinct agents and such constraints turn out to be useful for cryptographic purposes. In relativistic cryptography we consider protocols in which interactions occur at distinct locations at well-defined times and we investigate why such a setting allows to implement primitives which would not be possible otherwise.

Relativistic cryptography is closely related to non-communicating models, which have been extensively studied in theoretical computer science. Therefore, we start by discussing non-communicating models and its applications in the context of interactive proofs and cryptography. We find which non-communicating models might be useful for the purpose of bit commitment, propose suitable bit commitment protocols and investigate their limitations. We explain how some non-communicating models can be justified by special relativity and study what consequences such a translation brings about. In particular, we present a framework for analysing security of multiround relativistic protocols. We show that while the analysis of classical protocols against classical adversaries is tractable, the case of quantum protocols or quantum adversaries in a classical protocol constitutes a significantly harder task.

The second part of the thesis is dedicated to analysing specific protocols. We start by considering a recently proposed two-round quantum bit commitment protocol. We start by proving security under the assumption that idealised devices (single-photon source, perfect detectors) are available. Then, we propose a fault-tolerant variant of the protocol which can be implemented using realistic devices (weak-coherent source, noisy and inefficient detectors) and present a security analysis which takes into account losses, errors, multiphoton pulses, etc. We also report on an experimental implementation performed in collaboration with an experimental group at the University of Geneva.

In the last part we focus on classical schemes. We start by analysing a known two-round classical protocol and we show that successful cheating is equivalent to winning a certain non-local game. This is interesting as it demonstrates that even if the protocol is entirely classical, it might be advantageous for the adversary to use quantum systems. We also propose a new, multiround classical bit commitment protocol and prove its security against classical adversaries. The advantage of the multiround protocol is that it allows us to increase the commitment time without changing the locations of the agents. This demonstrates that in the classical world an arbitrary long commitment can be achieved even if the agents are restricted to occupy a finite region of space. Moreover, the protocol is easy to implement and we discuss an experiment performed in collaboration with the Geneva group.

We conclude with a brief summary of the current state of knowledge on relativistic cryptography and some interesting open questions that might lead to a better understanding of the exact power of relativistic models.
\chapter*{List of publications}
This thesis is based on three publications.

\noindent Chapters 3 and 4 are based on
\paperA
\noindent Chapter 5 is based on
\paperB
\noindent Chapter 6 is based on
\paperC
During his graduate studies the author has also contributed to the following publications.
\begin{enumerate}
\item \textbf{Query complexity in expectation} \arxiv{1411.7280}\\
J.~Kaniewski, T.~Lee and R.~de Wolf\\
\textsl{Automata, Languages, and Programming: Proceedings of ICALP '15},\\
\textsl{Lecture Notes in Computer Science} \textbf{9134} (2015).
\item \textbf{Equivalence of wave-particle duality to entropic uncertainty} \arxiv{1403.4687}\\
P.~J.~Coles,	J.~Kaniewski and S.~Wehner\\
\textsl{Nature Communications} \textbf{5}, 5814 (2014).\\
(presented at \textsl{AQIS '14})
\item \textbf{Entropic uncertainty from effective anticommutators} \arxiv{1402.5722}\\
J.~Kaniewski, M.~Tomamichel and S.~Wehner\\
\textsl{Physical Review A} \textbf{90}, 012332 (2014).\\
(presented at \textsl{AQIS '14} and \textsl{QCrypt '14})
\item \textbf{A monogamy-of-entanglement game with applications to device-independent quantum cryptography} \arxiv{1210.4359}\\
M.~Tomamichel, S.~Fehr, J.~Kaniewski and S.~Wehner\\
\textsl{New Journal of Physics} \textbf{15}, 103002 (2013).\\
(presented at \textsl{Eurocrypt '13} and \textsl{QCrypt '13})
\end{enumerate}
\tableofcontents
\mainmatter
\chapter*{Notation and list of symbols}
\addcontentsline{toc}{chapter}{Notation and list of symbols}
\begin{center}
\def\arraystretch{0.95}
\setlength{\tabcolsep}{0.7cm}
\begin{tabular}{c | l}
\textbf{Symbol} & \textbf{Meaning}\\
\hline
$[n]$ & set of integers from $1$ to $n$\\
$\abs{\cdot}$ & cardinality of a set or modulus of a number\\
$\sH$ & a Hilbert space\\
$\dim \sH$ & dimension of $\sH$\\
$\sH^{*}$ & dual space of $\sH$\\
$\cL(\sH)$ & linear operators acting on $\sH$\\
$\cH(\sH)$ & Hermitian operators acting on $\sH$\\
$\mathbb{1}$ & identity matrix\\
$L^{*}$ & complex conjugate of $L$\\
$L\tran$ & transpose of $L$ (with respect to the standard basis)\\
$L^{\dagger}$ & Hermitian conjugate of $L$\\
$\ket{\phi}, \ket{\psi}$ & pure quantum states\\
$\rho, \sigma$ & mixed quantum states\\
$\ket{\Psi_{d}}$ & maximally entangled state of dimension $d$\\
$H$ & Hadamard matrix\\
$\norm{\cdot}_{p}$ & Schatten $p$-norm\\
$\norm{\cdot}$ & Schatten $\infty$-norm\\
$\tr$ & trace\\
$\tr_{A}$ & partial trace over $A$\\
$\Phi$ & quantum channel\\
$\id$ & identity channel\\
$\wham(\cdot)$ & Hamming weight\\
$\dham(\cdot)$ & Hamming distance\\
$\oplus$ & exclusive-OR (\texttt{XOR})\\
$*$ & finite-field multiplication\\
$\Pr[\cdot]$ & probability\\
$\ave{\cdot, \cdot}$ & inner product\\
$\cX, \cY$ & finite alphabets\\
$\amsbb{F}_{q}$ & finite field of order $q$\\
$\cP_{k}$ & $k\th$ player (in a multiplayer game)\\
\end{tabular}
\end{center}
\chapter{Introduction}
Quantum cryptography lies at the intersection of physics and computer science. It brings together different communities and makes for a lively and exciting environment. It demonstrates that the fundamental principles of quantum physics can be cast and studied using the operational approach of cryptography. Besides, thanks to recent technological advances, practical applications are just round the corner.

Due to the interdisciplinary nature of quantum cryptography the relevant background knowledge spans multiple fields, which makes it particularly difficult to provide an introduction which would be both complete and concise. We have, therefore, chosen to focus on the topics which are directly related to quantum cryptography and skip over the less relevant areas.

This chapter starts with a short introduction to \emph{cryptography}, which is the study of exchanging and processing information in a secure fashion. We focus on \emph{two-party} (or \emph{mistrustful}) cryptography, whose goal is to protect the \emph{privacy} of an honest party interacting with potentially dishonest partners. Then, we introduce \emph{quantum information theory}, which studies how quantum systems can be used to store and process information. We discuss the main features that distinguish it from the classical information theory and briefly describe the early history of the field. The next part of this chapter brings the two topics together under the name of \emph{quantum cryptography}. We give a brief account of its early days, again, with a particular focus on two-party cryptography. We finish by giving a brief outline of this thesis.
\section{Cryptography}
\label{sec:cryptography}
Cryptography has been around ever since rulers of ancient tribes realised the need to send secret (or private) messages. Ideally, such messages should reveal no information if intercepted by an unauthorised party. The solution to this problem is known as a \index{cipher} \emph{cipher}, which is simply a procedure for converting a secret message (called the \emph{plaintext}) into another message (called the \emph{ciphertext}), which should be intelligible to a friend (who knows the particular cipher we are using) but should give no information to an enemy. The first confirmed accounts of simple ciphers come from ancient Greece and Rome, for example Julius Caesar used a simple shift cipher (now also known as a Caesar cipher) to ensure privacy of his correspondence. Until modern times designing practical (i.e.~easy to implement and difficult to break) ciphers was essentially the only branch of cryptography. One such cipher known as the \emph{one-time pad} \index{one-time pad} was invented by Gilbert S.~Vernam and Joseph O.~Mauborgne in 1917.\footnote{Note, however, that ideas that the one-time pad hinges on appeared as early as 1882 in a book by Frank Miller. For details consult an interesting survey on the state of cryptography at the turn of the century by Bellovin \cite{bellovin11}.} While the one-time pad guarantees (provably) secure communication it requires the two parties to share a random string of bits, known as a \emph{key}, which is as long as the message they want to send. This quickly becomes impractical if the parties want to exchange large amounts of data.

A report presented by Claude Shannon in 1945 marks the birth of modern cryptography \cite{shannon49}.\footnote{This work, presented in 1945 as a classified report at Bell Telephone Labs, was declassified and published in 1949.} Shannon proposed a formal definition of a (perfectly) secure cipher and proved that one-time pad satisfies such a stringent requirement. Moreover, he proved that any cipher that guarantees perfect security requires the key to be as long as the message (which essentially means that the one-time pad cannot be improved). But the contributions of this work go well beyond encryption and the analysis of one-time pad, as it was the first time that cryptography was phrased in the rigorous language of mathematics. This put cryptography on equal footing with other established sciences and set the stage for information theory (discovered by Shannon a couple of years later).

Nowadays cryptography is a mature field within which hundreds of \emph{cryptographic tasks} (or \emph{primitives}\index{cryptographic primitive}) have been defined and studied (and encryption, while obviously very important, is just one of them). Except for purely practical reasons for studying these tasks there is also a deeper motivation. Certain questions in cryptography (e.g.~finding sufficient assumptions to perform a given task or proving impossibility results) give us valuable and operational insight into the underlying information theory. While classical information theory is relatively well understood, its quantum counterpart is not. That is why studying quantum cryptography is an important pursuit and contributes towards our understanding of the quantum world we (probably) live in.

In this thesis we only consider a branch of cryptography known as \emph{two-party} \index{two-party cryptography} or \emph{mistrustful} cryptography, in which two parties, usually referred to as Alice and Bob, want to perform a certain task together but since they do not fully trust each other they want to minimise the amount of information revealed during the protocol. A simple example of such a scenario is the \emph{millionaires' problem} \index{millionaires' problem} introduced by Andrew Yao \cite{yao82}, in which two millionaires want to find out who is richer without revealing their actual wealth. This is certainly an interesting problem and, in fact, one that we often face in our everyday lives. Below we present and motivate some other natural two-party tasks.

\begin{itemize}
\item \textbf{Example 1:} Alice uses an online movie service called Bob, which charges separately for every downloaded movie. Alice has paid for one movie and wants to download it but being paranoid about privacy she is reluctant to reveal her choice to Bob. On the other hand, Bob wants to make sure that Alice only downloads one movie (and not more) so he is not keen on giving her access to the entire database. This problem, called \emph{oblivious transfer}\index{oblivious transfer}\footnote{Oblivious transfer comes in multiple flavours and the one described above is called $1$-out-of-$N$ oblivious transfer, where $N \geq 2$ is the total number of movies offered by Bob. Since we are only interested in fundamental possibility or impossibility results, studying the case of $N = 2$ is sufficient (it is known how to interconvert these primitives including even more exotic variants like Rabin oblivious transfer \cite{crepeau88}).}, turns out to be a convenient building block for two-party cryptography. In fact, it can be used to construct any other two-party primitive \cite{kilian88}.
\item \textbf{Example 2:} Alice has supernatural powers that allow her to predict the future, for example the results of tomorrow's draw of the national lottery. She wants to impress Bob (she likes to be admired) but she does not want him to get rich (she knows that money does not bring happiness). Hence, the goal is to \emph{commit} to a message without actually \emph{revealing} it until some later time. Such primitives are known as \emph{commitment schemes} \index{commitment scheme} \cite{blum81, brassard88}\footnote{Blum \cite{blum81} only implicitly mentions commitment schemes while Brassard, Chaum and Cr\'{e}peau \cite{brassard88} define them explicitly. See an encyclopedia entry on commitment schemes for more details \cite{crepeau11a}.} and the simplest one, in which the committed message is just one bit, is called \emph{bit commitment} \index{bit commitment} and constitutes one of the main topics of this thesis.
\item \textbf{Example 3:} Alice is a quantum hacker and throughout the years she has exposed dozens of improperly formulated security proofs and misguided calculations. Having realised the damage done to the quantum community she has contacted a law enforcement agency represented by Bob to negotiate turning herself in. Alice and Bob want to schedule a secret meeting but for obvious security reasons they want to make sure that the location is chosen in a truly random fashion. In other words, Alice and Bob want to agree on a random choice, which neither of them can bias (or predict it in advance). This primitive known as \emph{coin tossing} (or \emph{coin flipping}) \index{coin tossing} was introduced by Blum \cite{blum81}.
\end{itemize}

All these tasks produce conflicting interests between Alice and Bob. It is clear that security for either party can be ensured at the cost of leaving the other party completely unprotected. In case of oblivious transfer, for example, Alice could give up her privacy and simply announce which movie she wants to watch. Alternatively, Bob could provide Alice with the entire database, hoping that she will not abuse his trust.

The goal of two-party cryptography is to first come up with the right mathematical definition of these primitives and then find in what circumstances and under what assumptions they can (or cannot) be implemented. It is also interesting to study \emph{reductions} between different primitives, i.e.~how to use one primitive to implement another one, which leads to a resource theory for cryptography. For example, oblivious transfer can be used to implement commitment schemes because choosing a particular message can be interpreted as committing to its label. Commitment schemes, on the other hand, can be used to generate trusted randomness. For example, in order to generate one trusted bit we use a commitment scheme with two possible values (such a primitive is known as \index{bit commitment} \emph{bit commitment}). Alice commits to a bit $a$, then Bob announces bit $b$ and finally Alice reveals $a$ and the outcome of the coin toss is declared to be $a \oplus b$. As long as at least one of the parties is honest the resulting bit is uniform. The use of a commitment scheme ensures that $b$ does not depend on $a$ (which would allow Bob to cheat).

The holy grail of the field is the so-called \emph{information-theoretic security}\footnote{Some authors prefer to use the term \emph{unconditional security} instead. The name is motivated by the fact that the security proof assumes \emph{nothing} about the adversary. However, this has been criticised as every security model contains assumptions and no security statement can be proven without referring to them.}. There, the basic assumption is that the dishonest party is restricted by the underlying information theory, which is arguably the weakest assumption that one needs to perform security analysis. The term information-theoretic security goes back to Shannon (e.g.~see his definition of secure encryption \cite{shannon49}).

Unfortunately, it turns out that two-party primitives cannot be implemented with information-theoretic security (for both parties) unless we make some further assumptions.\footnote{While the impossibility is usually intuitive showing it formally requires some effort. As an example we present an informal argument why oblivious transfer is not possible with information-theoretic security. Consider the situation at the end of the protocol. If Bob is not able to deduce which movie Alice chose to download, it must be the case that the knowledge contained in the interaction is sufficient to reconstruct at least two different movies and nothing can stop Alice from doing that.} Below we give a brief overview of various (reasonable) assumptions that make information-theoretically secure two-party cryptography possible.
\begin{itemize}
\item \textbf{Trusted third-party :} The trivial solution is to introduce a trusted third-party, which implements the primitive for Alice and Bob. In the paranoid world, in which Alice and Bob trust nobody but themselves, this is not a satisfactory solution. Moreover, it makes all tasks trivially possible.
\item \textbf{Pre-shared resources :} Another solution that allows for two-party cryptography is to equip Alice and Bob with some shared correlations. This could be either shared randomness \cite{rivest99} or access to a source of inherent and unpredictable noise that allows to generate such correlations during the protocol \cite{crepeau97, winter03}.\footnote{Even for tasks whose only purpose is to generate trusted randomness like coin tossing this is still a non-trivial scenario because the correlations initially shared between Alice and Bob might be different from the ones we want to generate.}
\item \textbf{Technological limitations :} The standard real-world solution to the commitment task is for Alice to lock her message in a safe box, which she then hands over to Bob while keeping the key. Whenever Alice wants to reveal the message, she gives the key to Bob, who opens the safe box and reads the message. This is secure as long as Alice has no way of remotely modifying the message and Bob has no tools to open the safe box, i.e.~we must assume that they are subject to certain technological limitations. One can also assume that their ``digital technology'' is limited, e.g.~by restricting their computational power or storage capabilities, which again makes secure two-party cryptography possible. The former leads to the rich and practically important field of computational security\footnote{Computational security relies on the assumption that the adversary cannot solve a certain mathematical problem and let us mention two problematic aspects of this assumption. First of all, our belief that some mathematical problems are difficult is based mainly on the fact that many bright people have tried to solve them and failed (or maybe the successful ones prefer to keep a low profile). An efficient algorithm for solving such problems might be announced tomorrow and render all the currently used cryptographic protocols insecure. Secondly, most such schemes are vulnerable to \emph{retroactive} attacks. If a message sent today is required to remain secret for the next twenty years, the mathematical problem must resist new algorithms and improved computing power that might be developed in these twenty years. This is why we would like to ultimately drop such assumptions and find more solid foundations for our cryptographic systems.},  while the latter leads to the bounded storage model \cite{maurer91}.
\item \textbf{Communication constraints :} It is well-known that interrogating suspects one by one leads to better results than dealing with the whole group at the same time. In the cryptographic language this corresponds to forcing one (or more) parties to delegate agents, who perform certain parts of the protocol without communicating. Such setting was originally introduced in complexity theory under the name of \emph{multiprover models}\footnote{To avoid confusion we talk about \emph{multiprover} models in the context of complexity theory but use the term \emph{multiagent} in case of cryptographic protocols.} to evade certain impossibility results \cite{benor88}. These models are interesting from the cryptographic point of view but we must be explicit how they are adjusted to fit the framework of standard two-party cryptography (in which there are only two parties interacting and not more). On the bright side some types of non-communicating models can (with subtle adjustments) be implemented by requiring multiple agents to interact simultaneously at multiple locations (under the assumption that the speed of light is finite). The first explicit examples of such relativistic protocols came from Adrian Kent \cite{kent99, kent05}. This field, now known as \emph{relativistic cryptography}, constitutes the main topic of this thesis.
\end{itemize}
\section{Quantum information theory}
As mentioned before the report written by Shannon in 1945 marks the beginning of modern cryptography \cite{shannon49}. Thinking about encryption and the one-time pad led him to questions about the nature of information. Shannon's next paper investigating fundamental limits of compression and transmission \cite{shannon48} is considered the beginning of \emph{(classical) information theory}, which became an active field of research with a wide range of practical implications. While the basic framework of quantum mechanics already existed at the time (introduced in the 1920s and 30s by Bohr, Born, de Broglie, Dirac, Einstein, Heisenberg, Planck, Schr\"{o}dinger and others), rigorous connections between the two were not established until much later.

In 1935 Einstein, Podolsky and Rosen wrote a paper in which they argue that quantum mechanics cannot be considered a complete theory \cite{einstein35}. They postulate that for every measurement whose outcome is certain there exists an ``element of reality'' and deduce that due to the uncertainty principle incompatible observables cannot have simultaneous elements of reality. On the other hand, they note that in case of \index{entanglement} \emph{entangled}\footnote{The term \emph{Verschr\"{a}nkung} used ``to describe the correlations between two particles that interact and then separate, as in the Einstein-Podolsky-Rosen experiment'' first appeared in a letter written by Schr\"{o}dinger who also proposed the English translation: \emph{entanglement}.} particles the elements of reality of one system depend on the measurements performed on the other. Since they perceive the elements of reality as something objective, independent of any measurement process, they conclude that the quantum-mechanical description must be incomplete. This idea was further developed by John Bell \cite{bell64} who realised that the assumptions of Einstein, Podolsky and Rosen boil down to the existence of \index{local hidden variables} \emph{local hidden variables}, which completely determine the outcome of all possible measurements. Bell showed that any theory satisfying these requirements (like the classical theory) is subject to certain restrictions (now known as \index{Bell inequality} \emph{Bell inequalities}) and demonstrated that quantum mechanics violates such restrictions. The first explicit Bell inequality proposed by Clauser, Horne, Shimony and Holt \cite{clauser69} is a clear-cut evidence that the set of quantum correlations is strictly bigger than its classical counterpart.
\index{commitment scheme}
Realising that quantum mechanics gives rise to an information theory which is qualitatively different that the classical version, opened a new, fruitful research direction. Questions concerning storing or transmitting information using quantum systems have the appealing feature of being operational and fundamental at the same time. In the 1970s Holevo proved how many classical bits can be reliably stored in a quantum system \cite{holevo73} and Helstrom showed how to optimally distinguish two quantum states \cite{helstrom76}.

In 1980 Boris Tsirelson published a breakthrough paper, which exactly characterises the set of correlations achievable using quantum systems (in a restricted class of scenarios) \cite{tsirelson80}. Another important result concerning quantum correlations comes from Reinhard Werner, who showed that entanglement, while necessary, is not a sufficient condition for observing stronger-than-classical correlations \cite{werner89}. In 1982 Wootters and \.{Z}urek proved the celebrated \index{the no-cloning theorem} \emph{no-cloning} theorem, which states that given a single copy of an unknown quantum state, there does not exist a physical procedure that produces two (perfect) copies \cite{wootters82}. While the result itself is rather simple (including the proof), it has far-reaching consequences and shows that one should be rather careful when applying the classical intuition to quantum systems. Around the same time the first ideas to use quantum systems to perform computation came about. Richard Feynman proposed the concept of \index{quantum simulation} \emph{quantum simulation}, i.e.~using one quantum system to simulate another \cite{feynman82} while David Deutsch initiated the study of \index{quantum computation} \emph{quantum computation} by introducing the concept of a quantum Turing machine and presenting a simple problem which can be solved more efficiently using quantum systems \cite{deutsch85}. While the problem introduced by Deutsch is of little practical use, it is important as the first demonstration that quantum computing is strictly more powerful than its classical counterpart.

In 1994 Peter Shor published a paper that changed the status of quantum computation from an exercise in linear algebra to a field of potentially enormous practical impact \cite{shor94}. Shor proposed an algorithm that can efficiently factor large numbers and solve the discrete logarithm problem, which, as a consequence, allows to break all commonly used public cryptography systems. In 1996 Lov Grover published an algorithm that gives a quadratic speed-up while searching an unstructured database \cite{grover96}.\footnote{Note that the speed-up of Grover's algorithm is provable, i.e.~it is quadratically faster than \emph{any} classical algorithm. Shor's algorithm, on the other hand, is exponentially faster than \emph{the best known} classical algorithm.} These two results sparked enormous interest as they showed that quantum computation might be important from the practical point of view. Since then the task of finding new quantum algorithms and building an actual quantum computer has been a full-time job of hundreds of computer scientists, physicists and engineers.

It seems fair to say that it is the breakthroughs in quantum computation that gave the whole field a significant push and encouraged many brilliant researchers to work on quantum information. Since then the field has developed rapidly and this includes aspects closely related to quantum computation like quantum error-correction or quantum computer architecture but also areas which are not directly relevant like quantum correlations, quantum foundations, quantum Shannon theory or quantum cryptography. For more information we refer to a brief survey on early quantum information written by Bennett and Shor in 1998 \cite{bennett98} or to a book by Nielsen and Chuang \cite{nielsen00}, which became the primary textbook in the field (in particular for quantum computation). For a detailed introduction to the information-theoretic aspects (the quantum Shannon theory) see Chapter 1 of Mark M. Wilde's book \cite{wilde13}.

\section{Quantum cryptography}
\label{sec:quantum-cryptography}
In the late 1960s Stephen Wiesner wrote a paper on how to use quantum particles of spin-$\frac{1}{2}$ to produce money that is ``physically impossible to counterfeit''. The paper got rejected from a journal and ended up in Wiesner's drawer (the paper was eventually published in ACM SIGACT News \cite{wiesner83} about fifteen years later). These ideas were further pursued by Bennett, Brassard, Breidbart and Wiesner \cite{bennett83} and led to a groundbreaking paper proposing the first \index{quantum key distribution} quantum key distribution protocol, which allows two distant parties to communicate securely through an insecure quantum channel \cite{bennett84}. In 1991 Artur Ekert proposed a quantum key distribution protocol that relied on entanglement and Bell's theorem \cite{ekert91}. Another protocol (which relies on entanglement but not Bell's theorem) was presented in Ref.~\cite{bennett92a} and soon the first experimental demonstration of quantum key distribution was reported together with concrete solutions for the classical post-processing phase and explicit security estimates \cite{bennett92b}. Since then an enormous amount of progress has been made in both theoretical and practical aspects of quantum key distribution and it is well beyond the scope of this introduction to discuss it. A recent article by Ekert and Renner is an excellent account of the current state of quantum key distribution \cite{ekert14}.

Before we go into the details let us state very clearly that throughout this thesis we work under the (implicit) assumption that Alice and Bob trust their own devices. In other words, if the protocol requires Alice to generate a certain quantum state, she is capable of constructing a device that does just that and she may rest assured that the source does not accumulate information about the previous uses or leak secret data through extra degrees of freedom. While this assumption seems natural and easy to ensure in the classical world, it becomes more of a challenge in the quantum world simply because our understanding and expertise in quantum technologies are limited. These considerations gave rise to the field of \emph{device-independent} cryptography which aims to design protocols which remain secure even if executed using faulty or malicious devices. The fact that such strong security guarantees are even possible is clearly remarkable and this topic has received a lot of interest in the last couple of years. Due to a large volume of works on this topic we do not attempt to list the relevant references and point the interested reader at comprehensible and accessible lectures notes by Valerio Scarani~\cite{scarani12} as well as Sections IV.C and IV.D of a recent review on Bell nonlocality~\cite{brunner14}.

While quantum key distribution was and still remains the predominant area of research in quantum cryptography, other applications have been present from the very beginning as exemplified by Wiesner's unforgeable quantum money. The original paper of Bennett and Brassard contains a bit-commitment-based coin tossing protocol \cite{bennett84}. As pointed out by the authors the protocol is insecure if one of the parties leaves the quantum states untouched (instead of performing the prescribed measurements) but they consider it a ``merely theoretical threat'' due to the technological difficulty of implementing such a strategy. In 1991 Brassard and Cr\'{e}peau proposed a different quantum bit commitment protocol \cite{brassard91}, which does not suffer from the previous problem but is vulnerable against an adversary who can perform \emph{coherent measurements}, i.e.~joint measurements on multiple quantum particles, which, again, is considered difficult. By combining the two quantum bit commitment protocols they obtain a coin tossing protocol which can only be broken by an adversary who can \emph{both} keep entanglement \emph{and} perform coherent measurements. In the meantime a quantum protocol for oblivious transfer was proposed whose security, again, relies on the adversary being technologically limited \cite{bennett92}. In 1993 Brassard, Cr\'{e}peau, Jozsa and Langlois \cite{brassard93} proposed a new bit commitment protocol which comes with a complete security proof that does not rely on any technological assumptions. In other words, the protocol is claimed to be secure against all attacks compatible with quantum physics. In 1992 Bennett et al.~suggested how bit commitment and quantum communication can be used to construct oblivious transfer \cite{bennett92}. This construction was formalised and proven secure by Yao \cite{yao95}, who refers to it as the ``canonical construction'', which gave the optimistic impression that quantum mechanics allows for secure two-party cryptography without any extra assumptions.\footnote{This construction shows that in the quantum world bit commitment and oblivious transfer are equivalent, which is believed not to be true classically.} Unfortunately, it was later discovered that the protocol proposed in Ref.~\cite{brassard93} is insecure, which soon led to a complete impossibility result \cite{mayers97, lo97}. For a detailed account of quantum cryptography until that point please consult Refs.~\cite{brassard96, crepeau96, brassard97}.

The initial results of Mayers, Lo and Chau began a sequence of negative results. Impossibility of bit commitment immediately rules out oblivious transfer and, in fact, the same techniques can be used to rule out any one-sided two-party computation (i.e.~a primitive in which inputs from two parties produce output which is only given to one of them) \cite{lo97a}. The more complicated case of two-sided computation was first considered by Colbeck (for a restricted class of functions) \cite{colbeck07a} while the general impossibility result was proven by Buhrman, Christandl and Schaffner \cite{buhrman12}. In case of string commitment (i.e.~when we simultaneously commit to multiple bits) it is clear that the perfect primitive cannot be implemented but the situation becomes slightly more involved when it comes to imperfect primitives as the results depend on the exact security criteria used \cite{buhrman06a, buhrman08}. For more recent impossibility proofs of bit commitment see Refs.~\cite{dariano07, winkler11, chiribella13}.

While perfect quantum bit commitment is not possible, it is interesting to know what security trade-offs are permitted by quantum mechanics. In the classical case the trade-offs are trivial: in any classical protocol at least one of the parties can cheat with certainty. Preliminary results on the quantum security trade-offs were proven by Spekkens and Rudolph \cite{spekkens01}, while the optimal trade-off curve was found by Chailloux and Kerenidis \cite{chailloux11}. Interestingly enough, the achievability is argued through a construction that uses a complicated and rather poorly understood weak coin flipping protocol by Mochon \cite{mochon07}.

Another direction (similar to what was done previously in the classical world) is to identify the minimal assumptions that would make two-party cryptography possible in the quantum world.

One solution available in the classical world is to give Alice and Bob access to some trusted randomness. The quantum generalisation of this assumption would be to give Alice and Bob access to quantum systems or some other source of stronger-than-classical correlations \cite{buhrman06, winkler11a}. Such correlations indeed allow us to implement secure bit commitment. The advantage of this assumption over the classical counterpart is that in the classical case we had to trust whoever distributed the randomness (in the original paper referred to as the \emph{trusted initialiser} \cite{rivest99}). On the other hand, stronger-than-classical correlations guarantee security regardless of where they came from.

A natural quantum extension of the bounded storage model proposed by Maurer \cite{maurer91} is the quantum bounded storage model \cite{damgard05, damgard07, schaffner10} and its generalisation to the case of noisy quantum storage \cite{wehner08a, konig12, berta14}. While storing classical information seems easy and cheap (which makes the assumption of the adversary's bounded storage not particularly convincing), reliable storage of quantum information continues to pose a significant challenge and, hence, makes for a reasonable assumption. Another technological limitation that leads to secure bit commitment is the restriction on the class of quantum measurements that the dishonest party can perform \cite{salvail98}.

The proposal to combine quantum mechanics with relativistic\footnote{Throughout this thesis the term relativity always refers to special relativity.} communication constraints (attributed to Louis Salvail) was already mentioned in 1996 \cite{brassard96, crepeau96}. The early papers of Kent \cite{kent99, kent05} consider security against quantum adversaries but the actual protocols are completely classical. To the best of our knowledge, the first quantum relativistic protocol was proposed by Colbeck and Kent for a certain variant of coin tossing \cite{colbeck06a}. This marks the beginning of \emph{quantum relativistic cryptography}.
\section{Outline}
The main theme of this thesis is relativistic quantum cryptography with a particular focus on commitment schemes. Chapter~\ref{chap:preliminaries} contains the necessary background in quantum information theory and cryptography.

In Chapter~\ref{chap:nc-models} we introduce non-communicating models as they originally appeared in the context of interactive proofs. We show why they are useful in cryptography and determine the exact communication constraints that might allow for secure commitment schemes. For each of these models we present a provably secure bit commitment protocol.

Chapter~\ref{chap:relativistic-protocols} introduces the framework for relativistic protocols. We start with a couple of simple examples and then present a procedure which maps a relativistic protocol onto a model with partial communication constraints. We show that at least in some scenarios the analysis of such models is tractable.

In Chapter~\ref{chap:transmitting} we focus on a particular quantum bit commitment protocol. We analyse its security by mapping it onto a simple quantum guessing game. Moreover, we adapt the original protocol to make it robust against experimental errors and we extend the security analysis appropriately. We briefly report on an implementation of the protocol done in collaboration with an experimental group at the University of Geneva.

In Chapter~\ref{chap:multiround} we propose a new, classical multiround bit commitment protocol and analyse its security against classical adversaries. The multiround protocol allows to achieve arbitrarily long commitments (at the cost of growing resources) with explicit and easily-computable security guarantees. Again, we briefly discuss an experiment performed in collaboration with the Geneva group.

Chapter~\ref{chap:conclusions} summarises the content of this thesis and outlines a couple of interesting direction for future research in quantum relativistic cryptography.
\chapter{Preliminaries}
\label{chap:preliminaries}
In this chapter we establish the notation, nomenclature and some basic concepts used throughout this thesis.
\section{Notation and miscellaneous lemmas}
\subsection{Strings of bits}
\label{sec:strings-of-bits}
Given two bits $a, b \in \{0, 1\}$ we use ``$\oplus$'' to denote their exclusive-OR (\texttt{XOR})
\begin{equation*}
a \oplus b := a + b \mod 2.	
\end{equation*}
For an $n$-bit string $x \in \bs{n}$, let $x_{k}$ be the $k\th$ bit of $x$ and the \texttt{XOR} of two strings (of equal length) is defined bitwise. The fractional Hamming weight \index{Hamming weight} of $x$ is the fraction of ones in the string
\begin{equation*}
\wham(x) = \frac{1}{n} \abs[\big]{ \{ k \in [n] : x_{k} = 1 \} },
\end{equation*}
where $\abs{\cdot}$ denotes the cardinality of the set. The fractional Hamming distance \index{Hamming distance} between $x$ and $y$ is the fraction of positions at which the two strings differ
\begin{equation*}
\dham(x, y) = \frac{1}{n} \abs[\big]{ \{ k \in [n] : x_{k} \neq y_{k} \} }.
\end{equation*}
Note that the Hamming weight can be interpreted as the distance from the string of all zeroes $0^{n}$: $\wham(x) = \dham(x, 0^{n})$. For $S \subseteq [n]$, we use $x_{S}$ to denote the substring of $x$ specified by the indices in $S$. If $d \in \{0, 1\}$ is a bit, we define
\begin{equation*}
d \cdot x =
\begin{cases}
0^{n} &\nbox{if} d = 0,\\
x &\nbox{if} d = 1.
\end{cases}
\end{equation*}
\subsection{Cauchy-Schwarz inequality for probabilities}
When dealing with probabilities we use uppercase letters to denote random variables and lowercase letters to denote values they might take, e.g.~$\Pr[X = x]$. For $j, k \in \cS$ we use $\sum_{j \neq k}$ as a shorthand notation for $\sum_{j \in \cS} \sum_{k \in \cS \setminus \{j\}}$.

\begin{lem}
\label{lem:cauchy-schwarz}
Let $X$ be a uniform random variable over $[n]$, i.e.~$\Pr[X = x] = \frac{1}{n}$ for all $x \in [n]$, and let $\{E_{j}\}_{j = 1}^{m}$ be a family of events defined on $[n]$. Let $p$ be the average probability (of these events)
\begin{equation*}
p := \frac{1}{m} \sum_{j = 1}^{m} \Pr[ E_{j} ]
\end{equation*}
and $c$ be the cumulative size of the pairwise intersections
\begin{equation*}
c := \sum_{j \neq k} \Pr[ E_{j} \wedge E_{k} ].
\end{equation*}
Then the following inequality holds
\begin{equation*}
p \leq \frac{ 1 + \sqrt{1 + 4c} }{2m}.
\end{equation*}
\end{lem}
\begin{proof}
Each event can be represented by a vector in $\amsbb{R}^{n}$ whose entries are labelled by integers from $[n]$. If a particular outcome belongs to the event, we set the corresponding component to $1/\sqrt{n}$ and if it does not we set it to $0$
\begin{equation*}
[ s_{j} ]_{x} =
\begin{cases}
\frac{1}{\sqrt{n}} &\nbox[6]{if} x \in E_{j},\\
0 &\nbox[6]{otherwise.}
\end{cases}
\end{equation*}
Moreover, let $n$ be the normalised, uniform vector: $[n]_{x} = 1/\sqrt{n}$ for all $x \in [n]$. It is straightforward to check that with these definitions we have
\begin{equation*}
\Pr[ E_{j} ] = \ave{s_{j}, n} = \ave{s_{j}, s_{j}} \nbox[8]{and} \Pr[ E_{j} \wedge E_{k} ] = \ave{s_{j}, s_{k}},
\end{equation*}
where $\ave{\cdot, \cdot}$ denotes the standard inner product on $\amsbb{R}^{n}$ and since the vectors are non-negative we have $\ave{s_{j}, s_{k}} \geq 0$. Since the inner product is linear we have
\begin{equation*}
p = \frac{1}{m} \sum_{j = 1}^{m} \Pr[ E_{j} ] = \frac{1}{m} \sum_{j} \ave{s_{j}, n} = \frac{1}{m} \ave{ \sum_{j} s_{j}, n},
\end{equation*}
which can be upper bounded using the Cauchy-Schwarz inequality. Since $\ave{n, n} = 1$ we have
\begin{equation*}
\ave{ \sum_{j} s_{j}, n}^{2} \leq \sum_{j k} \ave{s_{j}, s_{k}} = \sum_{j} \ave{s_{j}, s_{j}} + \sum_{j \neq k} \ave{s_{j}, s_{k}} = m p + c,
\end{equation*}
which gives the following quadratic constraint
\begin{equation*}
p^{2} \leq \frac{p}{m} + \frac{c}{m^{2}}.
\end{equation*}
Solving for $p$ gives the desired bound.
\subsection{Chernoff bound for the binomial distribution}
\begin{lem}[\cite{chernoff52}]
\label{lem:chernoff}
Let $X_{1}, X_{2}, \ldots, X_{n}$ be independent random variables taking on values 0 or 1. Let $X = \sum_{i = 1}^{n} X_{i}$ and $\mu$ be the expectation value of $X$. Then for any $\delta > 0$ the following inequality holds
\begin{equation*}
\Pr[X < (1 - \delta) \mu] < \bigg(\frac{e^{- \delta}}{(1 - \delta)^{1 - \delta}} \bigg)^{\mu} \leq \exp \bigg(- \frac{\mu \delta^{2}}{2} \bigg).
\end{equation*}
Alternatively, setting $s = (1 - \delta) \mu$ gives
\begin{equation*}
\Pr[X < s] < \exp \bigg(- \frac{1}{2} \Big( \sqrt{\mu} - \frac{s}{\sqrt{\mu}} \Big)^{2} \bigg).
\end{equation*}
\end{lem}
\section{Quantum mechanics}
Quantum mechanics despite its mysterious nature admits a relatively simple mathematical description. While it is an interesting question to ask \emph{why} quantum mechanics is as it is, instead of being more (or less) powerful (and indeed such questions constitute the main topic of quantum foundations), we take a more hands-on approach. Namely, we accept the standard textbook formulation of quantum mechanics as it is and investigate its consequences. Section \ref{sec:linear-algebra} defines the basic notions of linear algebra (and, hence, can be skipped by most readers), which will be necessary to describe the quantum formalism in Section \ref{sec:quantum-formalism}.
\subsection{Linear algebra}
\label{sec:linear-algebra}
The following section contains the bare minimum of linear algebra necessary to understand this thesis and serves primarily the purpose of establishing consistent notation and nomenclature. For a complete and detailed introduction to linear algebra we refer to the excellent textbooks by Rajendra Bhatia \cite{bhatia97, bhatia09}.

In this thesis we restrict our attention to finite-dimensional systems. Let $\sH$ be a Hilbert space of finite dimension $d = \dim \sH < \infty$ over complex numbers.
Let $\sH^{*}$ denote the dual space of $\sH$, i.e.~the space of linear functionals on $\sH$. We employ the \index{bra-ket notation} \emph{bra-ket} notation proposed by Paul Dirac \cite{dirac39}, in which elements of $\sH$ are written as \emph{kets} $\ket{\psi} \in \sH$ and each ket has an associated \emph{bra}, denoted by $\bra{\psi} \in \sH^{*}$, such that applying the linear functional $\bra{\psi}$ to an arbitrary vector $\ket{\phi}$, written as a \emph{bra-ket} $\braket{\psi}{\phi}$, corresponds exactly to evaluating the inner product between $\ket{\phi}$ and $\ket{\psi}$. A set of $d$ vectors $\{ \ket{e_{j}} \}_{j = 1}^{d}$ constitutes an orthonormal basis if the vectors are orthogonal and normalised, i.e.~$\braket{e_{j}}{e_{k}} = \delta_{jk}$, where $\delta_{jk}$ is the Kronecker delta.

Let $\cL(\sH)$ be the set of linear operators acting on $\sH$. The \emph{identity operator}, denoted by $\mathbb{1}$, is the unique operator that satisfies
\begin{equation*}
\mathbb{1} \ket{\psi} = \ket{\psi}
\end{equation*}
for all $\ket{\psi} \in \sH$. Writing a linear operator $L \in \cL(\sH)$ in a particular basis $\{ \ket{e_{j}} \}_{j = 1}^{d}$ leads to a $d \times d$ (complex) matrix whose entries equal
\begin{equation*}
L_{jk} = \bramatket{e_{j}}{L}{e_{k}},
\end{equation*}
where the expression $\bramatket{e_{j}}{L}{e_{k}}$ should be understood as $\braket{e_{j}}{ \big(L | e_{k}} \big)$. Note that while the operator and its matrix representation are not the same object (the former is basis-independent, while the latter corresponds to a particular basis) for the purpose of this thesis this distinction may be ignored and we will use the two terms interchangeably.
The \index{trace} \emph{trace} of a square matrix $L$ is the sum of its diagonal entries
\begin{equation*}
\tr L = \sum_{j} L_{jj} = \sum_{j} \bramatketq{e_{j}}{L}.
\end{equation*}
The \index{Hermitian conjugate} \emph{Hermitian conjugate} of an operator $L$, denoted by $L^{\dagger}$, is defined to satisfy
\begin{equation*}
[L^{\dagger}]_{jk} = [L_{kj}]^{*},
\end{equation*}
where $^{*}$ denotes the \index{complex conjugate} complex conjugate. An operator satisfying $L^{\dagger} = L$ is called \emph{Hermitian} (or \emph{self-adjoint}) and we denote the set of Hermitian operators acting on $\sH$ by $\cH(\sH)$. Operators satisfying $L L^{\dagger} = L^{\dagger} L = \mathbb{1}$ are called \index{unitary operator} unitary operators or unitaries.

It is easy to verify that for a Hermitian operator $H = H^{\dagger}$ we have $\bramatketq{\psi}{H} \in \amsbb{R}$ for all vectors $\ket{\psi} \in \sH$. A Hermitian operator is called \index{positive semidefinite operator} \emph{positive semidefinite} if
\begin{equation*}
\bramatketq{\psi}{H} \geq 0
\end{equation*}
for all vectors $\ket{\psi} \in \sH$, which is often written as $H \geq 0$.

Every linear operator $L \in \cL(\sH)$ admits a \index{singular value decomposition} \emph{singular value decomposition}, i.e.~it can be written in the form $L = USV$, where $U$ and $V$ are unitary operators and $S$ is a diagonal matrix of real, non-negative entries known as the \emph{singular values} of $L$. Let $s = ( s_{1}, s_{2}, \ldots, s_{d} )$ be the vector of singular values. For $p \in [1, \infty)$ the \index{Schatten norm} Schatten $p$-norm of $L$, denoted by $\norm{L}_{p}$, is defined as the vector $p$\hspace{0.75pt}-norm of $s$
\begin{equation*}
\norm{L}_{p} := \Big( \sum_{j = 1}^{d} s_{j}^{p} \Big)^{1/p}.
\end{equation*}
For the purpose of this thesis we will only need the limit $p \to \infty$ so let us define
\begin{equation*}
\norm{L} := \lim_{p \to \infty} \norm{L}_{p} = \max_{j} s_{j}.
\end{equation*}
\subsection{Quantum formalism}
\label{sec:quantum-formalism}
A pure state of a quantum system is described by a normalised vector, i.e.~$\ket{\psi} \in \sH$ such that $\braketq{\psi} = 1$. We adopt the convention that every $d$-dimensional Hilbert space $\sH$ is equipped with an orthonormal basis $\{ \ket{k} \}_{k = 0}^{d - 1}$, which we call the \index{computational basis} \emph{computational} (or \emph{standard}) basis. Writing $\ket{\psi}$ in this basis
\begin{equation*}
\ket{\psi} = \sum_{j = 0}^{d - 1} c_{j} \ket{j}
\end{equation*}
allows us to interpret it as a $d$-dimensional complex unit vector. The global phase of a state is inconsequential, i.e.~quantum mechanics tells us that vectors $\ket{\psi}$ and $e^{i \alpha} \ket{\psi}$ (for $\alpha \in \amsbb{R}$) correspond to the same physical state.
The smallest non-trivial quantum system corresponds to $d = 2$ and is called a \index{qubit} \emph{qubit} (a term coined by Schumacher and Wootters \cite{schumacher95}). The Hadamard operator is defined as
\begin{equation*}
H = \frac{1}{\sqrt{2}} \sum_{j, k \in \{0, 1\}} (-1)^{jk} \ketbra{j}{k}
\end{equation*}
or
\begin{equation*}
H = \frac{1}{\sqrt{2}} \left( \begin{array}{c c}
	1 & 1 \\
	1 & -1
\end{array} \right).
\end{equation*}
It is easy to verify that $H$ is simultaneously Hermitian ($H = H^{\dagger}$) and unitary ($H H^{\dagger} = H^{\dagger} H = H^{2} = \mathbb{1}$). Define $\ket{+} := H \ket{0}$, $\ket{-} := H \ket{1}$ and let us call $\{ \ket{+}, \ket{-} \}$ the \index{Hadamard basis} \emph{Hadamard} (or \emph{diagonal}) basis. The computational and Hadamard bases are widely used in cryptography because they are an example of \emph{mutually unbiased bases} (for $d = 2$), i.e.~they satisfy
\begin{equation*}
\abs{ \braket{0}{+} } = \abs{ \braket{0}{-} } = \abs{ \braket{1}{+} } = \abs{ \braket{1}{-} } = \frac{1}{\sqrt{2}} \bigg( = \frac{1}{\sqrt{d}} \bigg),
\end{equation*}
which captures the notion of being maximally incompatible.

A mixed quantum state on $\sH$ is a Hermitian operator, which is positive semidefinite and of unit trace. We define the set of (mixed) quantum states on $\sH$
\begin{equation*}
\cS(\sH) := \{ \rho \in \cH(\sH) : \rho \geq 0 \nbox{and} \tr \rho = 1 \}.
\end{equation*}
The operator $\rho$ describing a mixed state is called the \index{density matrix} \emph{density matrix}. Mixed states are a generalisation of pure states: an arbitrary pure state $\ket{\psi}$ can be represented as a density matrix $\rho = \ketbraq{\psi}$. Mixed states arise naturally when dealing with composite systems.

Suppose we have two systems (or registers) $A$ and $B$ described by Hilbert spaces $\sH_{A}$ and $\sH_{B}$, respectively, and we want to describe the global state of the system. What are the allowed states on $A$ and $B$ taken together? In case of pure states, quantum mechanics tells us to take the tensor product of the two Hilbert spaces, i.e.~$\ket{\psi}_{AB} \in \sH_{A} \otimes \sH_{B}$. Therefore, an arbitrary pure bipartite state can be written as
\begin{equation*}
\ket{\psi}_{AB} = \sum_{jk} c_{jk} \ket{j}_{A} \ket{k}_{B},
\end{equation*}
where $\ket{j}_{A} \ket{k}_{B}$ should be understood as $\ket{j}_{A} \otimes \ket{k}_{B}$ (the tensor product symbol is commonly omitted to avoid notational clutter). Given a bipartite system one might wonder what can be said about the \emph{marginal states} on $A$ and $B$ (similar to the concept of the marginals of a probability distribution). In particular, one would expect that if we restrict ourselves to measurements on $A$ alone then it should be possible to ``truncate'' $\ket{\psi}_{AB}$ to $A$ by disregarding any information about $B$. This intuition leads the concept of \index{reduced state} \emph{reduced states}. Let us first write the density matrix corresponding to $\ket{\psi}_{AB}$
\begin{equation*}
\rho_{AB} = \ketbraq{\psi}_{AB} = \sum_{jj' kk'} c_{jk}^{\phantom{*}} c_{j'k'}^{*} \ketbra{j}{j'}_{A} \otimes \ketbra{k}{k'}_{B}.
\end{equation*}
Given an operator acting on multiple registers we define the operation of \index{partial trace} \emph{partial trace} which ``traces out'' a particular register, e.g.~
\begin{equation*}
\tr_{B} \big( \ketbra{j}{j'}_{A} \otimes \ketbra{k}{k'}_{B} \big) = \ketbra{j}{j'}_{A} \cdot \tr_{B} \big( \ketbra{k}{k'}_{B} \big) = \ketbra{j}{j'}_{A} \cdot \delta_{kk'}.
\end{equation*}
Note that the standard trace operation corresponds to tracing all the registers. It is easy to verify that partial traces commute so we can without ambiguity write $\tr_{AB}(\cdot) = \tr_{A} \tr_{B} (\cdot) = \tr_{B} \tr_{A} (\cdot)$. Tracing out the $B$ register from the density matrix $\rho_{AB}$ gives
\begin{equation*}
\rho_{A} = \tr_{B} \rho_{AB} = \sum_{jj' k} c_{jk}^{\phantom{*}} c_{j'k}^{*} \ketbra{j}{j'}_{A},
\end{equation*}
which is easily verified to be a valid quantum state $\rho_{A} \in \cS(\sH_{A})$ and which we call \emph{the reduced state on $A$}. It is easy to verify that the knowledge of $\rho_{A}$ suffices to make all possible predictions about operations or measurements that act solely upon subsystem $A$. In cryptography reduced states are important because they allow us to quantify the amount of knowledge that a particular subsystem provides to its holder.

Once we know how to describe the state of a quantum system we would like to know how we can interact with it. To extract any information from a quantum state one needs to \index{measurement} \emph{measure} it. Note that this is one of the aspects in which quantum theory differs significantly from its classical counterpart. In the classical world the object and its (complete) description are \emph{operationally equivalent}: given the description one can construct the object and given the object one can determine (to an arbitrary precision) its description. In the quantum world a single copy of an object gives us significantly less information than its complete description as demonstrated by the no-cloning theorem \cite{wootters82}. In contrast to the classical world, every quantum system can be measured in multiple ways, which means that the measurement process must be described explicitly. A measurement\footnote{We implicitly assume that we are only interested in the classical outcome of the measurement and ignore the post-measurement state.} on a $d$-dimensional quantum state which yields outcomes from a finite alphabet $\cX$ is a collection of positive semidefinite operators $\{ F_{x} \}_{x \in \cX}$\footnote{To avoid confusion whenever describing the set of measurement operators we explicitly state the index that must be summed over to obtain identity.} that add up to ($d$-dimensional) identity
\begin{equation}
\label{eq:measurement-operators-conditions}
F_{x} \geq 0 \nbox{and} \sum_{x \in \cX} F_{x} = \mathbb{1}.
\end{equation}
Quantum mechanics is a probabilistic theory, i.e.~it only allows us to calculate \emph{probabilities} of observing different outcomes. According to \index{Born's rule} \emph{Born's rule} \cite{born26} measuring the state $\rho \in \cS( \sH )$ yields outcome $x$ with probability
\begin{equation*}
p(x) = \tr (F_{x} \rho).
\end{equation*}
It is easy to see that the condition \eqref{eq:measurement-operators-conditions} is imposed to ensure that the resulting probability distribution is non-negative and normalised \emph{for every state}. Note that such an information-theoretic formulation of the measurement process does not necessarily coincide with the notion of measuring a physical quantity, e.g.~the outcome might not be a number so one cannot talk about the expectation value or the standard deviation of the measurement.

The process of measuring a quantum state can be seen as a map that takes a quantum state and outputs a probability distribution. This naturally generalises to maps in which the output remains quantum and such maps are known as \index{quantum channel} \emph{quantum channels}. The identity channel (i.e.~the unique channel that leaves every state unaffected) is denoted by $\id$\footnote{Note that it is common in quantum information to use the same symbol for the identity channel and the identity operator since it is usually clear from the context which one is meant. To avoid confusion we prefer to use different symbols.}. Generally, a map $\Phi: \cL(\sH_{A}) \to \cL(\sH_{B})$ is a quantum channel iff:
\begin{enumerate}
\item $\Phi$ is linear, i.e.~for any $\alpha, \beta \in \amsbb{C}$ and $X, Y \in \cL(\sH_{A})$
\begin{equation*}
\Phi( \alpha X + \beta Y ) = \alpha \Phi( X ) + \beta \Phi( Y ).
\end{equation*}
\item $\Phi$ is completely positive, i.e.~for any $X_{AR} \in \cH( \sH_{A} \otimes \sH_{R} )$, where $\sH_{R}$ is an auxiliary Hilbert space of arbitrary dimension,
\begin{equation*}
X \geq 0 \implies ( \Phi_{A} \otimes \id_{R} ) (X_{AR}) \geq 0.
\end{equation*}
\item $\Phi$ is trace-preserving, i.e.~for any $X \in \cL( \sH_{A} )$
\begin{equation*}
\tr \Phi(X) = \tr X.
\end{equation*}
\end{enumerate}
These properties can be rigorously derived from the assumption that a channel is a result of a unitary evolution acting on a larger Hilbert space. On a more pragmatic level, these rules ensure that the channel is a linear map that takes quantum states on $A$ into valid quantum states on $B$. When dealing with multipartite states it might be useful to explicitly write out the input and output registers, e.g.~$\Phi_{A \to B}$.
\subsection{Remote state preparation}
\label{sec:equivalence}
A state of the form
\begin{equation}
\label{eq:cq-state}
\rho_{X B} = \sum_{x} p_{x} \ketbraq{x}_{X} \otimes \rho^{x}_{B}
\end{equation}
is called \index{classical-quantum state} \emph{classical-quantum} (cq) since the first register represents a classical random variable $X$ while the second is a general quantum system. Such states describe how a quantum system can be correlated with some classical data. One way of obtaining such a state is to sample the classical random variable $X$ and prepare subsystem $B$ in a particular state conditional on the outcome. Here, we show how to use entanglement to \emph{remotely} prepare a certain class of such states, a phenomenon also known as \emph{steering}.

Define the \index{maximally entangled state} maximally entangled state of dimension $d$ as
\begin{equation*}
\ket{\Psi_{d}}_{AB} := \frac{1}{\sqrt{d}} \sum_{j = 1}^{d} \ket{j}_{A} \ket{j}_{B}.
\end{equation*}
It is easy to verify that in this case both marginals are \emph{maximally mixed}, i.e.~proportional to the identity matrix
\begin{equation*}
\tr_{A} \ketbraq{\Psi_{d}}_{AB} = \tr_{B} \ketbraq{\Psi_{d}}_{AB} = \frac{ \mathbb{1} }{d}.
\end{equation*}
Moreover, for an arbitrary linear operator $L = \sum_{jk} L_{jk} \ketbra{j}{k}$, we have
\begin{equation*}
\tr_{A} \big[ (L \otimes \mathbb{1}) \ketbraq{\Psi_{d}} \big] = L\tran,
\end{equation*}
where $L\tran$ denotes the \index{matrix transpose} transpose with respect to the computational basis
\begin{equation*}
L\tran = \sum_{jk} L_{jk} \ketbra{k}{j}.\footnote{The transpose operation is basis-dependent just like the definition of the maximally entangled state.}
\end{equation*}
If we replace $L$ with a measurement operator this implies that observing a particular outcome on $A$ results in a particular subnormalised quantum state on $B$. Hence, we have \emph{remotely prepared} a state on $B$ by performing a measurement on $A$. It is easy to see that any cq-state of the form \eqref{eq:cq-state} which satisfies
\begin{equation}
\label{eq:correct-marginal}
\sum_{x} p_{x} \rho^{x} = \frac{ \mathbb{1} }{d},
\end{equation}
can be obtained by performing the right measurement on one half of the $d$-dimensional maximally entangled state. More specifically, the appropriate measurement $\{F_{x}\}_{x}$ is described by measurement operators $F_{x} = d p_{x} (\rho^{x}) \tran$. The restriction \eqref{eq:correct-marginal} expresses the rule that the reduced state on $B$ must remain unchanged, i.e.~it must remain maximally mixed. This phenomenon turns out to be important in quantum cryptography.

An essential feature of quantum information is the ability to encode information in two (or more) incompatible bases. The most common example was originally introduced by Wiesner \cite{wiesner83} but goes under the name of \index{BB84 states} \emph{BB84 states} (after Bennett and Brassard who popularised the term \cite{bennett84}). In this case Alice uses either computational or Hadamard basis to encode a logical bit $x \in \{0, 1\}$ in a qubit which she later sends to Bob. If the logical bit is uniform the two encodings lead to
\begin{equation*}
\rho_{XB} = \frac{1}{2} \sum_{x} \ketbraq{x}_{X} \otimes \ketbraq{x}_{B}.
\end{equation*}
and
\begin{equation*}
\rho_{XB} = \frac{1}{2} \sum_{x} \ketbraq{x}_{X} \otimes H \ketbraq{x}_{B} H,
\end{equation*}
respectively. It is easy to verify that both of these satisfy relation \eqref{eq:correct-marginal} with $d = 2$. This leads to an important observation (in this particular cryptographic context due to Bennett, Brassard and Mermin \cite{bennett92a}) that such states can be prepared by first generating the maximally entangled state of two qubits $\ket{\Psi_{2}}$ and then measuring subsystem $A$ in the right basis. In fact, Alice simply makes a measurement in either computational or Hadamard basis.

Since measurements on Alice's side commute with any operations on Bob's side, they can be delayed until some later point in the protocol, which means that now Alice and Bob share entanglement during the protocol. In other words, we have turned a prepare-and-measure scheme (Alice prepares a state and sends it to Bob, who performs a measurement), in which there is no entanglement between Alice and Bob, into an equivalent (from the security point of view) entanglement-based scheme (Alice and Bob simultaneously perform measurements on a shared entangled state). 
Often the entanglement-based schemes are easier to analyse, which we we will take advantage of to prove security of a quantum relativistic bit commitment protocol in Chapter~\ref{chap:transmitting}.
\section{Multiplayer games}
\label{sec:multiplayer-games}
For the purpose of this thesis, a game is an interaction between a \emph{referee} and one or more \emph{players}. The referee asks each player a question and the player must give an answer. In most cases the players are not allowed to communicate during the game. A \emph{strategy} is a procedure that the players follow to generate their answers. At the end of the game, the referee decides whether the game is won or lost.
\subsection{Classical and quantum strategies}
For concreteness, we consider a game of $m$ \emph{non-communicating} players. Each player receives an \emph{input} from $\cX$ and is required to \emph{output} a symbol from $\cY$ ($\cX$ and $\cY$ are arbitrary finite alphabets). A game is defined by the input distribution
\begin{equation*}
p : \underbrace{\cX \times \cX \times \ldots \times \cX}_{m \text{ times}} \mapsto [0, 1]
\end{equation*}
and a \emph{predicate function}
\begin{equation*}
V :\underbrace{(\cX \times \cY) \times \ldots \times (\cX \times \cY)}_{m \text{ times}} \mapsto \{0, 1\}
\end{equation*}
which specifies whether the players \emph{win} or \emph{lose} for a particular combination of inputs and outputs.\footnote{Clearly, this generalises in a straightforward manner to the case where the range of $V$ is $\amsbb{R}$. Then $V$ assigns a particular \emph{score} to every combination of inputs and outputs. However, in this thesis we only consider games in which the players either win or lose.}

Every strategy available to classical players can be written as a convex combination of deterministic strategies. Hence, the maximum winning probability, denoted by $\omega$ and referred to as \emph{the classical value} of the game, can be achieved by a deterministic strategy. A deterministic strategy is a collection of $m$ functions $(f_{j})_{j = 1}^{m}$, $f_{j} : \cX \mapsto \cY$, which determine each player's response. Therefore,
\begin{equation*}
\omega := \max_{f_{1}, f_{2}, \ldots, f_{m}} \sum_{x_{1} \in \cX} \ldots \sum_{x_{m} \in \cX} p (x_{1}, \ldots, x_{m}) V \big(x_{1}, f_{1}(x_{1}), \ldots, x_{m}, f_{m}(x_{m}) \big),
\end{equation*}
where the maximum is taken over all combinations of functions.

Quantum players, in turn, are allowed to share a quantum state and perform measurements that depend on the inputs. For simplicity in the quantum setting we only describe two-player games ($m = 2$) but these concepts extend in a straightforward way to an arbitrary number of players (see for example Ref.~\cite{vidick13}). A quantum strategy consists of a bipartite pure quantum state (of finite dimension) $\ket{ \psi }_{AB}$\footnote{It is sufficient to consider pure states since a mixed state can be written as a convex combination of pure states.} and measurements that each player will perform for every possible input $x \in \cX$, denoted by $\{ F_{y}^{x} \}_{y \in \cY}, \{ G_{y}^{x} \}_{y \in \cY}$. The maximum winning probability achievable by quantum players denoted by $\omega^{*}$ is called \emph{the quantum value}
\begin{equation*}
\omega^{*} := \sup \sum_{y_{1}, y_{2} \in \cY} \sum_{x_{1}, x_{2} \in \cX} p (x_{1}, x_{2}) V \big(x_{1}, y_{2}, x_{2}, y_{2} \big) \bramatketq{\psi}{ F_{y_{1}}^{x_{1}} \otimes G_{y_{2}}^{x_{2}} },
\end{equation*}
where the optimisation is taken over all quantum strategies.

Calculating the classical value of a game can be done by iterating over all possible strategies. While this is clearly not efficient (the number of strategies to check is exponential in the number of inputs), at least in principle it can be done.\footnote{In fact, finding the classical value of a general game is NP-hard, i.e.~we believe it cannot be done efficiently.} On the other hand, computing the quantum value is a more difficult problem and no generic procedure is known.\footnote{The quantum value of an \texttt{XOR} game can be calculated using semidefinite programming techniques \cite{wehner06}. For general games there exist hierarchies by Navascu\'{e}s, Pironio, Ac\'{i}n \cite{navascues07} and Doherty, Liang, Toner, Wehner \cite{doherty08}, which give increasingly tighter approximations on the correct value. While these hierarchies ultimately converge to the correct value, the rate of convergence is not well-understood. Moreover, calculating the higher level approximations becomes a difficult task from the computational point of view.} The problem stems from the fact that we do not have a convenient description of the quantum set of correlations (i.e.~there is no efficient procedure to decide whether a given point belongs to the set or not). To establish an upper bound on the quantum value of a game it is common to consider a larger set of correlations known as the \emph{no-signalling} correlations, which does admit a simple description. Intuitively, this is the largest set of correlations that does not allow to send messages between different parties and the simplest example is the so-called Popescu-Rohrlich box \cite{popescu94}. Because the no-signalling set is a polytope (i.e.~the convex hull of a finite set of extreme points) we know how to optimise over it (at least in principle, efficiency considerations apply as before). For a detailed characterisation of different sets of correlations refer to a recent review paper on Bell nonlocality \cite{brunner14}.
\subsection{Finite fields}
A field is a set with two operations: addition and multiplication, which satisfy the usual properties as listed below.
\begin{itemize}
\item The field is closed under multiplication and addition.
\item Both operations are associative.
\item Both operations are commutative.
\item There exist additive and multiplicative identity elements.
\item There exist additive and multiplicative inverses (except for the additive identity which does not have a multiplicative inverse).
\item Multiplication is distributive over addition.
\end{itemize}
It is easy to see that real or complex numbers form with the standard addition and multiplication are fields. We call a field \emph{finite} (the name \emph{Galois field} is also used after \'{E}variste Galois) if the set of elements is finite. The order of a finite field is the number of elements in the set and a finite field of order $q$ exists iff $q$ is a prime power, i.e.~$q = p^{k}$ for some prime $p$ and integer $k$. Since all finite fields of a given order are isomorphic (i.e.~they are identical up to relabelling of the elements), we speak of \emph{the} finite field of order $q$ denoted by $\amsbb{F}_{q}$. For a thorough introduction to finite fields please consult an excellent book by Mullen and Mummert \cite{mullen07}.

Finite fields appear often in coding theory and cryptography since they are finite sets closed under (appropriately defined) addition, multiplication and their inverses. Moreover, all these operations can be implemented efficiently on a computer. Fields corresponding to $p = 2$ are a common choice since their elements have a natural representation as strings of bits. The protocol proposed in Chapter~\ref{chap:multiround} uses finite-field arithmetic and its security hinges on the difficulty of a certain family of multiplayer games. In this section we prove upper bounds on the classical value of such games and discuss the connection to a natural algebraic problem concerning multivariate polynomials over finite fields.
\subsection{Definition of the game}
\label{sec:game-definition}
\begin{figure}
\centering
\begin{tikzpicture}[scale=1, line width=0.5]
\draw (-0.8, -0.4) rectangle (0, 0.4);
\node at (-0.4, 0) {$\cP_{1}$};
\draw[-] (1.3, 1.5) -- (1.3, -1.5);
\draw (2.6, -0.4) rectangle (3.4, 0.4);
\node at (3, 0) {$\cP_{2}$};
\draw[-] (4.7, 1.5) -- (4.7, -1.5);
\hordots{6.2}{0}{0.08}{black}
\draw[-] (7.7, 1.5) -- (7.7, -1.5);
\draw (9.1, -0.4) rectangle (9.9, 0.4);
\node at (9.5, 0) {$\cP_{m}$};
\node at (-0.4, 1.4) {$X_{2}, X_{3}, \ldots, X_{m}$};
\draw [->] (-0.4, 1.1) to (-0.4, 0.6);
\node at (3, 1.4) {$X_{1}, X_{3}, \ldots, X_{m}$};
\draw [->] (3, 1.1) to (3, 0.6);
\node at (9.5, 1.4) {$X_{1}, X_{2}, \ldots, X_{m - 1}$};
\draw [->] (9.5, 1.1) to (9.5, 0.6);
\node at (-0.4, -1.4) {$Y_{1}$};
\draw [->] (-0.4, -0.6) to (-0.4, -1.1);
\node at (3, -1.4) {$Y_{2}$};
\draw [->] (3, -0.6) to (3, -1.1);
\node at (9.5, -1.4) {$Y_{m}$};
\draw [->] (9.5, -0.6) to (9.5, -1.1);
\end{tikzpicture}
\caption{The ``Number on the Forehead'' model. Vertical lines remind us that the players are not allowed to communicate.}
\label{fig:number-on-the-forehead}
\end{figure}
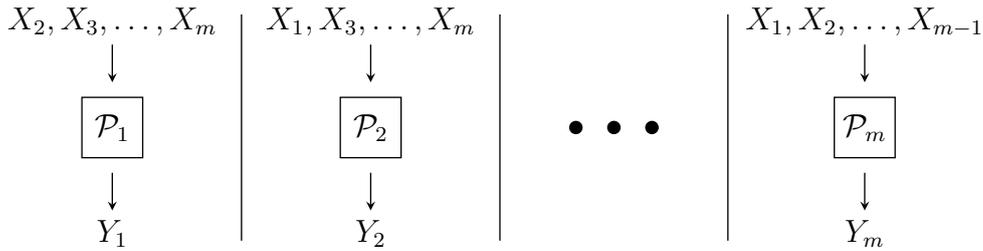
Buhrman and Massar \cite{buhrman05} proposed a generalisation of the CHSH game \cite{clauser69}, which was further studied by Bavarian and Shor \cite{bavarian15}. A natural multiplayer generalisation of this game arises in the security analysis of the multiround bit commitment protocol in Chapter~\ref{chap:multiround}. Since the analysis does not require familiarity with the actual actual bit commitment protocol and might be of independent interest, we have decided to make it a stand-alone component of the Preliminaries (rather than incorporating it in Chapter~\ref{chap:multiround}).

Consider a game with $m$ players, denoted by $\cP_{1}, \cP_{2}, \ldots, \cP_{m}$, and let $X_{1}, X_{2}, \ldots, X_{m}$ be random variables drawn independently, uniformly at random from $\amsbb{F}_{q}$
\begin{equation*}
p(x_{1}, x_{2}, \ldots, x_{m}) = ( q^{-1} )^{m} = q^{-m}.
\end{equation*}
We use $[n]$ to denote the set of integers between $1$ and $n$, $[n] := \{1, 2, \ldots, n\}$. In the \index{Number on the Forehead model} ``Number on the Forehead'' model \cite{chandra83} $\cP_{k}$ receives \emph{all the random variables except for the $k\th$ one}, which we denote by $X_{[m] \setminus \{k\}}$, and is required to output an element of $\amsbb{F}_{q}$, which we denote by $Y_{k}$ (see Fig.~\ref{fig:number-on-the-forehead}). The game is won if the \emph{sum of the outputs equals the product of the inputs} (all the operations are performed in the finite field), i.e.~the predicate function is
\begin{equation*}
V(x_{1}, y_{2}, \ldots, x_{m}, y_{m}) =
\begin{cases}
1 &\nbox{if} \prod_{k = 1}^{m} x_{k} = \sum_{k = 1}^{m} y_{k},\\
0 &\nbox{otherwise.}
\end{cases}
\end{equation*}
If the player $\cP_{k}$ employs a deterministic strategy described by $f_{k} : \amsbb{F}_{q}^{(m - 1)} \mapsto \amsbb{F}_{q}$, i.e.~he outputs $Y_{k} = f_{k} ( X_{ [m] \setminus \{k\} } )$, then the winning probability equals
\begin{equation*}
\omega_{m}(f_{1}, f_{2}, \ldots, f_{m}) := \Pr \Big[ \prod_{k = 1}^{m} X_{k} = \sum_{k = 1}^{m} f_{k}( X_{ [m] \setminus \{k\} } ) \Big].
\end{equation*}
As described in Section \ref{sec:multiplayer-games} the classical value of the game equals
\begin{equation}
\label{eq:omega-definition}
\omega_{m} := \max_{ f_{1}, f_{2}, \ldots, f_{m} } \omega_{m}(f_{1}, f_{2}, \ldots, f_{m}),
\end{equation}
where the maximisation is taken over all combinations of functions from $\amsbb{F}_{q}^{(m - 1)}$ to $\amsbb{F}_{q}$.\footnote{Clearly, this is a function of both $q$ and $m$ but we have decided not to mention the dependence on $q$ explicitly (to avoid overcrowding the symbol with sub- or superscripts). This is justified because in our application $q$ is a parameter that has to be chosen before the protocol begins, whereas $m$ can be decided upon at some later point.} Our goal is to find an upper bound on $\omega_{m}$ as a function of $q$ and $m$.
\subsection{Relation to multivariate polynomials over finite fields}
As the probability distribution of inputs is uniform the winning probability of a particular deterministic strategy (defined by a collection of functions $f_{1}, f_{2}, \ldots, f_{m}	$) is proportional to the number of inputs $(x_{1}, x_{2}, \ldots, x_{m})$ on which the following equality holds
\begin{equation}
\label{eq:winning-condition}
\prod_{k = 1}^{m} x_{k} = \sum_{k = 1}^{m} f_{k}( x_{ [m] \setminus \{k\} } ).
\end{equation}
Alternatively, we can count the zeroes of the following function
\begin{equation*}
P(x_{1}, x_{2}, \ldots, x_{m}) = \prod_{k = 1}^{m} x_{k} - \sum_{k = 1}^{m} f_{k}( x_{ [m] \setminus \{k\} } ).
\end{equation*}
By the Lagrange interpolation method every function from $\amsbb{F}_{q}^{n}$ to $\amsbb{F}_{q}$ (for arbitrary $n \in \amsbb{N}$) can be written as a polynomial. Therefore, the question concerns the number of zeroes of the polynomial $P$. Different strategies employed by the players give rise to different polynomials and we need to characterise what polynomials are ``reachable'' in this scenario. The output of $\cP_{k}$ is an arbitrary polynomial of $x_{ [m] \setminus \{k\}}$, hence, it only contains terms that depend on \emph{at most $m - 1$ variables}. This means that the part of $P$ that depends on \emph{all $m$ variables} comes solely from the first term and equals $\prod_{k = 1}^{m} x_{k}$. Therefore, finding the classical value of the game is equivalent to finding the polynomial with the largest number of zeroes, whose only term that depends on all $m$ variables equals $\prod_{k = 1}^{m} x_{k}$. This shows that the optimal strategy for our game is closely related to purely algebraic properties of polynomials over finite fields.
\subsection{A recursive upper bound on the classical value}
\label{sec:recursive-upper-bound}
Here, we find explicit upper bounds on $\omega_{m}$ through an induction argument. First, note that for $m = 1$ there is only one term on the right-hand side of Eq.~\eqref{eq:winning-condition} and since this term takes no arguments it is actually a constant. Since $X_{1}$ is uniform we have
\begin{equation*}
\omega_{1} := \max_{ c \in \amsbb{F}_{q} } \Pr [ X_{1} = c ] = \frac{1}{q}.
\end{equation*}
Now, we derive an upper bound on $\omega_{m}$ in terms of $\omega_{m - 1}$. For a fixed strategy $(f_{1}, f_{2}, \ldots, f_{m})$ the winning probability can be written as
\begin{align*}
\omega_{m}(f_{1}, f_{2}, \ldots, f_{m}) &= \Pr \big[ X_{1} X_{2} \ldots X_{m} = \sum_{k = 1}^{m} f_{k}( X_{ [m] \setminus \{k\} } ) \big]\\
&= \sum_{y \in \amsbb{F}_{q}} \Pr[X_{m} = y] \cdot \Pr \big[ X_{1} X_{2} \ldots X_{m} = \sum_{k = 1}^{m} f_{k}( X_{ [m] \setminus \{k\} } ) \big| X_{m} = y \big]\\
&= q^{-1} \sum_{y} \Pr \big[ X_{1} X_{2} \ldots X_{m} = \sum_{k = 1}^{m} f_{k}( X_{ [m] \setminus \{k\} } ) \big| X_{m} = y \big].
\end{align*}
Conditioning on a particular value of $X_{m}$ leads to events that only depend on $X_{1}, X_{2}, \ldots, X_{m - 1}$. In particular, setting $X_{m} = y$ defines the event $F_{y}$
\begin{equation*}
F_{y} \iff X_{1} X_{2} \ldots X_{m - 1} y = \sum_{k = 1}^{m - 1} f_{k}( X_{ [m - 1] \setminus \{k\} }, y ) + f_{m}(X_{[m - 1]}),
\end{equation*}
which satisfies
\begin{equation}
\label{eq:fy-definition}
\Pr[ F_{y} ] = \Pr \big[ X_{1} X_{2} \ldots X_{m} = \sum_{k = 1}^{m} f_{k}( X_{ [m] \setminus \{k\} } ) | X_{m} = y \big].
\end{equation}
We can use Lemma~\ref{lem:cauchy-schwarz} to find a bound on $\omega_{m}(f_{1}, f_{2}, \ldots, f_{m}) = q^{-1} \sum_{y} \Pr[ F_{y} ]$ as long as we are given bounds on $\Pr[ F_{y} \wedge F_{z}]$ for $y \neq z$.
\begin{prop}
\label{prop:intersections}
For $y \neq z$ we have $\Pr[ F_{y} \wedge F_{z}] \leq \omega_{m - 1}$.
\end{prop}
\begin{proof}
Eq.~\eqref{eq:fy-definition} defines $F_{y}$ through a certain equation in the finite field. If the equations corresponding to $F_{y}$ and $F_{z}$ are satisfied simultaneously then any linear combination of these equations is also satisfied. More specifically, we define a new event
\begin{equation}
\label{eq:gyz-definition}
G_{yz} \iff X_{1} X_{2} \ldots X_{m - 1} (y - z) = \sum_{k = 1}^{m - 1} f_{k}( X_{ [m - 1] \setminus \{k\} }, y ) - f_{k}( X_{ [m - 1] \setminus \{k\} }, z )
\end{equation}
and since $F_{y} \wedge F_{z} \implies G_{yz}$ we are guaranteed that $\Pr[F_{y} \wedge F_{z}] \leq \Pr[G_{yz}]$. To find an upper bound on $\Pr[ G_{yz} ]$ we give the players more power by allowing a more general expression on the right-hand side. In Eq.~\eqref{eq:gyz-definition} the $k\th$ term is a particular function of $X_{ [m - 1] \setminus \{k\} }, y$ and $z$, so let us replace it by an arbitrary function of these variables
\begin{equation*}
f_{k}( X_{ [m - 1] \setminus \{k\} }, y ) - f_{k}( X_{ [m - 1] \setminus \{k\} }, z ) \quad \to \quad g_{k}( X_{ [m - 1] \setminus \{k\} }, y, z ).
\end{equation*}
Under this relaxation, we arrive at the following equality
\begin{equation*}
X_{1} X_{2} \ldots X_{m - 1} (y - z) = \sum_{k = 1}^{m - 1} g_{k}( X_{ [m - 1] \setminus \{k\} }, y, z ).
\end{equation*}
Clearly, $(y - z)$ is a constant, non-zero multiplicative factor known to each player. Dividing the equation through by $(y - z)$ leads to the same game as considered before but one player has been eliminated (there are only $m - 1$ players now). Therefore,
\begin{equation*}
\Pr[ F_{y} \wedge F_{z} ] \leq \Pr[ G_{yz} ] \leq \omega_{m - 1}.
\end{equation*}
\end{proof}
\noindent This allows us to prove the main technical result.
\begin{prop}
The classical value of the game defined in Section \ref{sec:game-definition} satisfies the following recursive relation
\begin{equation}
\label{eq:recursive-bound}
\omega_{m} \leq \frac{1 + \sqrt{1 + 4 q ( q - 1 ) \omega_{m - 1} } }{ 2 q }.
\end{equation}
\end{prop}
\begin{proof}
The statement follows directly from combining Lemma~\ref{lem:cauchy-schwarz} with Proposition~\ref{prop:intersections}.
\end{proof}
\noindent Since we know that $\omega_{1} = q^{-1}$, we can obtain a bound on $\omega_{m}$ by recursive evaluation of Eq.~\eqref{eq:recursive-bound}. More precisely, we get $\omega_{m} \leq c_{m}$ for
\begin{equation*}
c_{m} =
\begin{cases}
q^{-1} &\nbox{for} m = 1,\\
\frac{1 + \sqrt{1 + 4 q ( q - 1 ) c_{m - 1} } }{ 2 q } &\nbox{for} m \geq 2.
\end{cases}
\end{equation*}
Note that this bound is always non trivial, i.e.~$c_{m} < 1$ for all values of $q$ and $m$. To obtain a slightly weaker but simpler form presented in Eq.~\eqref{eq:bound-multi} in Chapter~\ref{chap:multiround} we note that $1 - 4 q c_{m - 1} \leq 0$ and set $q = 2^{n}$.
\section{Cryptographic protocols and implementations}
\label{sec:cryptographic-protocols}
Cryptography is a field is driven by applications, i.e.~the starting point is a particular task that two (or more) parties want to perform. Formulating a task in a rigorous, mathematical language gives rise to a \index{cryptographic primitive} \emph{cryptographic primitive}. In case of two-party cryptography two aspects must be specified.
\begin{itemize}
\item \textbf{Correctness:} The expected behaviour when executed by honest parties.
\item \textbf{Security:} A list of behaviours that are forbidden \emph{regardless} of the strategies that the dishonest parties might employ.
\end{itemize}
Defining correctness is straightforward because what we want to achieve is clear from the beginning. Finding the right definition of security, on the other hand, might be a challenging task. Converting our intuition about what the primitive should not allow for into a mathematical statement is not always straightforward and often multiple security definitions are simultaneously in use depending on the exact context. Sometimes security is perfect (cf.~the hiding property of bit commitment in Definition \ref{df:hiding}), but more often it is quantified by a (small) number usually denoted by $\varepsilon$ (cf.~the binding property in Definition \ref{df:binding}), which can be (usually) understood as an upper bound on the probability that a cheating attempt is successful.

It is worth emphasising that no meaningful statements can be made if \emph{all involved parties} decide to cheat simply because if they collectively deviate in the ``right'' way they can produce any imaginable output. If all the dishonest parties form a coalition whose only goal is to enforce a certain output, nothing can stop them from achieving it. In particular, in the two-party case Alice and Bob could, instead of executing the protocol, decide to play a game of chess and then the output of the interaction would be a complete account of a chess game. Clearly, no cryptographic statements can be made about a chess game. Therefore, we only consider scenarios in which at least one party is honest and that is why in the two-party setting we prefer to talk about \emph{security for honest Alice (Bob)} instead of \emph{security against dishonest Bob (Alice)}.

Once the primitive has been defined we propose a protocol (i.e.~a sequence of interactions between the players) that implements it. Verifying the correctness of a cryptographic protocol is simple since the honest parties behave in a well-defined manner. Showing security, on the other hand, is more complicated because we need to characterise all possible ways in which the dishonest parties might deviate from the protocol and argue that none of them violates the security requirements of the primitive. In a protocol that does not achieve perfect security, the final outcome of a security proof is an upper bound on how well the dishonest party can cheat. Since the level of security that we are happy to accept depends on the precise circumstances, protocols usually come in families parametrised by an integer $n \in \amsbb{N}$ and the security guarantee is a function of $n$ ideally satisfying $\varepsilon(n) \to 0$ as $n \to \infty$. Increasing the value of $n$ leads to protocols that use more resources (e.g.~computation, communication or randomness) but achieve better security. Ideally, we would like $\varepsilon(n)$ to decay exponentially but inverse polynomial decay might also be acceptable. Security analysis of such a family of protocols aims to find the tightest bound, i.e.~lowest $\varepsilon(n)$, as a function of $n$.

Having performed the theoretical analysis of a protocol, the last step is to actually implement it. In case of mature technologies (e.g.~modern digital devices) \index{fault-tolerance} \emph{fault-tolerance} (capability of terminating correctly even in the presence of errors) is ensured at the hardware level so there is no need to introduce any extra measures in the actual protocol. The multiround classical relativistic bit commitment protocol discussed in Chapter~\ref{chap:multiround} is a prime example: the simplest theoretical protocol is already suitable for implementation and no modifications are necessary. On the other hand, in case of less developed fields like quantum technologies the situation is a bit more complicated. Since we have not yet found a way to (generically) eliminate all the errors, we must consider how they will affect our protocol. What happens when honest parties follow the protocol but their communication or storage suffers from noise? Depending on how severe the errors are, the protocol either terminates with the wrong output or it aborts. To prevent such an undesirable outcome the protocol must be modified to become fault-tolerant. The exact modifications that need to be made depend on what type of noise we want to protect ourselves against. More specifically, we need to have a model of noise that is simple enough to analyse but remains a reasonably faithful description of the experimental setup. As a consequence, turning a theoretical protocol into an experimental proposal is not so straightforward and usually requires multiple rounds of communication between the theoretician and the experimentalist. The new fault-tolerant protocol admits a couple of parameters, which determine its error tolerance, and these should be chosen to ensure that the protocol terminates successfully (with high probability) when performed by the honest parties. In this case asymptotic analysis is sufficient, since it is the actual experiment that demonstrates correctness (while calculations simply give us an indication whether the experiment is worth setting up).

Having modified our protocol we need to reassess its security and it is clear that introducing fault-tolerant features makes a protocol more vulnerable to cheating. Moreover, since the security analysis is supposed to please the most paranoid cryptographers, we must make minimal assumptions about the adversary. In particular, we do not want to impose on him any technological restrictions. Our devices are imperfect due to our lack of skills and knowledge but we do not want to assume that about the adversary. The standard approach to quantum cryptography is to assume that the devices used by the honest party are trusted (i.e.~their precise description including potential imperfections is known) but the devices used by the adversary might be arbitrary (i.e.~they are only limited by the laws of physics).\footnote{As mentioned before the trust assumption can in fact be dropped. See Section~\ref{sec:quantum-cryptography} for references.}

Clearly, requiring that our protocol is correct for honest parties with imperfect devices and remains secure against an all-powerful adversary puts us in a difficult situation. As mentioned before, the fault-tolerant protocol takes a couple of parameters which we can try adjusting but we might nevertheless reach the conclusion that guaranteeing correctness and security simultaneously is not possible. This means that the quality of the devices available to the honest parties is not sufficient to allow for a secure execution of the protocol.

We can turn this statement around and ask about the minimal requirements on the honest devices. How much noise can we tolerate before the protocol becomes insecure? Note that now this is a property of the protocol alone and we should aim to design protocols with the highest possible noise tolerance. In case of quantum technologies a successful implementation of a cryptographic protocol often requires a collective effort of the experimentalist (who attempts to reduce the experimental noise to the absolute minimum) and the theoretician (who improves the theoretical security analysis). An example of such an analysis for a quantum bit commitment protocol is presented in Chapter~\ref{chap:transmitting}.
\section{Bit commitment}
\label{sec:bit-commitment}
Recall Example 2 from Section~\ref{sec:cryptography}, in which Alice wants to commit to a certain message without actually revealing it. Commitment schemes have multiple applications, for example they allow us to prove that we know something  or that we are able to predict some future event without revealing any information in advance. They are also a useful tool to force different parties to act simultaneously, even if the communication model is inherently sequential. Consider two bidders who want to take part in an auction but there is no trusted auctioneer at hand. In the usual, sequential communication model one of them has to announce his bid first, which gives an unfair advantage to the other bidder. This can be rectified if the first bidder commits to his bid (instead of announcing it) and opens it only after the second bidder has announced his price. Hence, given access to a commitment functionality, one can perform a fair auction without a trusted third party. Moreover, commitment schemes are often used in reductions to construct other cryptographic primitives. For the purpose of this section we restrict ourselves to schemes in which the committed message is just one bit.

As explained in Section \ref{sec:cryptographic-protocols} the protocol should be correct (it should succeed if executed by honest parties) and secure (the honest party should be protected even if the other party deviates arbitrarily from the protocol).\footnote{One can also design \emph{cheat-sensitive} protocols \cite{aharonov00, kent04}, in which one party constantly monitors the other party's action and might abort the protocol in the middle if they believe that the other party is cheating, but we do not consider them here. Similarly, cheat-sensitive coin flipping protocols have been proposed \cite{spekkens02}.} To make precise mathematical statements, we need a formal description of the protocol in the quantum language.
\subsection{Formal definition}
\label{sec:formal-definition}
The primitive of bit commitment is usually split into two phases: the commit phase and the open phase. In the commit phase Alice interacts with Bob and at the end of the commit phase she should be committed, i.e.~she should no longer have the freedom to choose (or change) her commitment. Nevertheless, Bob should remain ignorant about Alice's commitment. In the open phase Alice sends to Bob the bit she has committed to, along with a proof of her commitment, which he examines to decide whether to accept the opening or not.

While this description is sufficient for most purposes, it has some undesirable features. First of all, since it does not explicitly \emph{mention} the period in between the two phases, it might create the impression that there \emph{is} no interval in between, i.e.~it might lead to the false conclusion that the \emph{end of the commit phase} and the \emph{beginning of the open phase} correspond to the same point of time. This is clearly misleading as the whole point of a commitment scheme is to obtain a finite interval \emph{between} the two, i.e.~a period in which Alice is committed to a message which Bob remains ignorant about. The distinction is usually not made explicit because in most protocols nothing happens in between the two phases (Alice and Bob just savour the moment of being securely committed), which means that the two points are \emph{operationally} equivalent (e.g.~any information that Bob might extract about Alice's commitment just before the open phase he might also extract immediately after the commit phase). This is not true for protocols in which communication continues in between the two phases and there the distinction is important. Therefore, we explicitly introduce the \emph{sustain} phase, i.e.~the period during which the commitment is valid. For reasons which will become clear soon, we call the beginning of the sustain phase the \emph{commitment point} and the end the \emph{opening point}. We also split up what is usually called the open phase into two separate parts: in the \emph{open} phase Alice unveils $\cval$ to Bob and sends him a proof of her commitment (which we assume to be a single message\footnote{The assumption that the open phase consists of a single message from Alice to Bob might seem restrictive but it does not rule out any interesting protocols. Any interaction between Alice and Bob in the open phase can be simulated locally by Bob given that Alice provides him with all the relevant information. Since there is no need to protect Alice's privacy any more, the security of the protocol is not affected. In fact, we could consider an extreme case in which Alice does not even extract a proof and instead passes all the (possibly quantum) information in her possession to Bob. Again, this would not affect security but might unnecessarily increase the size of the message. Note that while this simulation argument does not change the situation of honest Alice, it might change (to worse) the situation of dishonest Alice but this shall not concern us.}), while in the \emph{verify} phase Bob decides whether to accept the opening or not. The phase structure is shown in Fig.~\ref{fig:phases}.

We use $A$ and $B$ to denote the subsystems of Alice and Bob, respectively. We use $P$ to denote the proof, which is generated (in the open phase) by Alice and sent to Bob. We implicitly assume that $P$ contains the information about the value $\cval$ that Alice is trying to unveil. Since the commit and sustain phases are interactive they do not admit a compact description in the quantum formalism. The open phase can be described as a quantum channel $\Phi_{A \to P}^{\textnormal{open}}$, which acts on Alice's subsystem ($A$) to produce a proof ($P$). Bob's decision whether to accept or reject the commitment in the verify phase can be described by a binary measurement $\cM = \{ M_{\textnormal{accept}}, M_{\textnormal{reject}} \}$ performed jointly on subsystems $B$ and $P$.
\begin{figure}
\centering
\begin{tikzpicture}[scale=1]
\phases{-0.4}
\end{tikzpicture}
\caption{The phase structure of a generic bit commitment protocol.}
\label{fig:phases}
\end{figure}
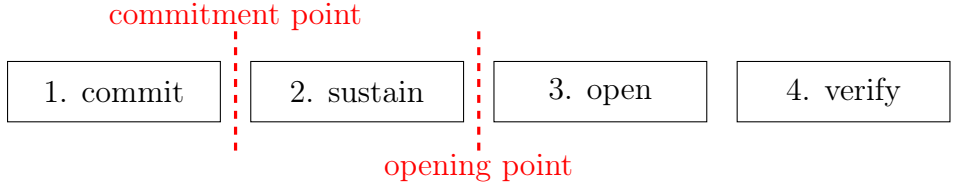

The honest scenario is relatively straightforward to analyse. The protocol specifies uniquely (for each value of Alice's commitment $\cval$) the state shared between Alice and Bob at every stage of the protocol and finding it explicitly is a matter of simple calculation.
\begin{df}
Let $\rho_{AB}^{\cval}$ be the state shared between Alice and Bob at the opening point, $\Phi_{A \to P}^{\textnormal{open}}$ be the opening map and $\cM = \{ M_{\textnormal{accept}}, M_{\textnormal{reject}} \}$ be the final measurement. A bit commitment protocol is \textbf{$(1 - \delta)$-correct} if for $\cval \in \{0, 1\}$ we have
\begin{gather*}
\tr (M_{\textnormal{accept}} \rho_{BP}^{\cval} ) \geq 1 - \delta,\\
\nbox{where} \rho_{BP}^{\cval} = \Phi_{A \to P}^{\textnormal{open}}(\rho_{AB}^{\cval}).
\end{gather*}
\end{df}

In the dishonest scenario the situation becomes a bit more complicated because the state shared between Alice and Bob is no longer uniquely specified. For example, if Alice is dishonest then the state of her subsystem might be completely arbitrary. For the purpose of defining security it is convenient to talk about the \emph{set} of states that dishonest Alice (or Bob) can \emph{enforce} during the protocol and we will use $\sigma_{AB}$ to denote such states (to distinguish them from the honest states denoted by $\rho_{AB}$).\footnote{Note that for a generic, multiround protocol characterising such sets might be a difficult task. Nevertheless, these sets are always well-defined and allow us to define security in a convenient way.} These sets are then used to quantify security.

As discussed before, coming up with the right security definition is not trivial because it requires us to turn the intuitive notion of security into a mathematical statement. It is useful to realise that the dishonest scenario is operationally equivalent to a game between the honest party (acting as a referee since their behaviour is determined by the protocol) and the dishonest party (a player who is allowed to adopt an arbitrary strategy). Thus, defining security is equivalent to specifying the exact rules of such a ``cheating game''.

To look at a concrete example let us start with the case of honest Alice and dishonest Bob. Bob's goal is to find out the value of Alice's commitment before the open phase begins, i.e.~at the opening point, and to achieve this he might deviate arbitrarily from the protocol. This admits a natural formulation as a game in which Alice (the referee) chooses $\cval \in \{0, 1\}$ uniformly at random and follows the honest protocol until the opening point. Then, Bob is challenged to guess $\cval$ at the opening point and the probability of guessing $\cval$ correctly is a natural measure of his cheating abilities. To phrase this in terms of quantum states, let $\sigma_{AB}^{\cval}$ be the state at the opening point and note that a particular strategy of dishonest Bob enforces two distinct states $(\sigma_{AB}^{0}, \sigma_{AB}^{1})$. 
\index{the hiding property}
\begin{df}
\label{df:hiding}
A bit commitment protocol is \textbf{hiding} if all pairs of states $(\sigma_{AB}^{0}, \sigma_{AB}^{1})$ that Bob can enforce at the opening point satisfy
\begin{equation}
\label{eq:equal-marginals}
\sigma_{B}^{0} = \sigma_{B}^{1},
\end{equation}
where $\sigma_{B}^{\cval} = \tr_{A} \sigma_{AB}^{\cval}$.
\end{df}
\noindent This definition implies that whatever strategy Bob employs, he obtains no information about Alice's commitment.\footnote{Note that this implies that the equality \eqref{eq:equal-marginals} holds not just \emph{at} at the opening point but at all times \emph{until} that point. If at any point the two reduced states were not equal Bob could simply store his system until the opening time (possibly sending Alice some freshly prepared states if necessary).} Note that this property is sometimes referred to as being \emph{perfectly hiding}, in contrast to schemes that only guarantee Alice partial security. Since all protocols considered in this thesis are perfectly hiding, we always use hiding to mean perfect security.

The case of dishonest Alice and honest Bob is a bit more complex. In order to claim that Alice's commitment begins \emph{at} the commitment point, we must show that at that point she no longer has the freedom to unveil both values, \emph{regardless} of the strategy adopted prior to that. In other words, the dishonest behaviour of Alice can be seen as two distinct strategies (corresponding to $\cval = 0$ and $\cval = 1$) which are identical until the commitment point and let us call such strategies \emph{compatible}. Intuitively, this means that she can delay the choice which strategy to follow until the commitment point.
\index{the binding property}
\begin{df}
\label{df:binding}
Let $(\sigma_{AB}^{0}, \sigma_{AB}^{1})$ be a pair of states that Alice can enforce at the opening point using compatible strategies and let $(\Phi^{\textnormal{cheat}, 0}_{A \to P}, \Phi^{\textnormal{cheat}, 1}_{A \to P})$ be opening maps. Define $p_{\cval}$ to be the probability that Alice's attempt to unveil $\cval$ is accepted by Bob
\begin{equation*}
p_{\cval} = \tr \big( M_{\textnormal{accept}} \big[ \Phi^{\textnormal{cheat}, \cval}_{A \to P} (\sigma_{AB}^{\cval}) \big] \big).
\end{equation*}
A bit commitment protocol is called \textbf{$\varepsilon$-binding} if for all states $(\sigma_{AB}^{0}, \sigma_{AB}^{1})$ and for all opening maps $(\Phi^{\textnormal{cheat}, 0}_{A \to P}, \Phi^{\textnormal{cheat}, 1}_{A \to P})$ we have
\begin{equation*}
p_{0} + p_{1} \leq 1 + \varepsilon.
\end{equation*}
\end{df}
\noindent Note that finding the optimal opening map for a particular intermediate state $\sigma_{AB}^{\cval}$ is a semidefinite program so it can be solved efficiently. Therefore, the cheating strategy is essentially specified by a pair of compatible strategies.

It is clear that the restriction that the two strategies are compatible is crucial. Clearly, Alice can enforce the honest pair of states $(\rho_{AB}^{0}, \rho_{AB}^{1})$, which leads to $p_{0} = p_{1} = 1$ (as long as the protocol is correct), but this cannot be achieved using compatible strategies (if it was possible, the protocol would be completely insecure).

Note that this formulation is equivalent to a game in which Alice employs some generic strategy until the commitment point and is \emph{immediately after} challenged to open either $\cval = 0$ or $\cval = 1$ chosen uniformly at random.\footnote{Note that this challenge must come from an external source like the referee. It is not convincing for Alice to unveil a bit chosen by herself.} A strategy that wins such a game with probability $\pwin$ is equivalent to a cheating strategy which achieves $p_{0} + p_{1} = 2 \pwin$. Thinking of the cheating scenario as a game often allows a more intuitive understanding of what the dishonest party is trying to achieve.

Note that it might seem natural to demand that there is \emph{only one} value that Alice might successfully open. However, this requirement is too strong because we cannot prevent Alice from committing to a random bit which leads to $p_{0} = p_{1} = \frac{1}{2}$. This idea can be developed further to produce a reachable notion of security~\cite{damgard05, wehner08a}, which has the advantage of being \emph{composable}, i.e.~security is guaranteed even if the primitive is used as a subroutine in a longer procedure \cite{canetti01}. However, it is known that this stronger notion of security cannot be reached in the relativistic setting (see Appendix \ref{app:classical-certification} for details). Therefore, in this thesis we only consider the weaker, non-composable Definition \ref{df:binding}.
\subsection{The Mayers-Lo-Chau impossibility result}
\label{sec:impossibility}
As explained in Section \ref{sec:quantum-cryptography}, in the early 1990s a significant effort went into investigating whether various two-party tasks can be solved using quantum protocols. Unfortunately, most of such tasks were ultimately shown to be impossible even in the quantum world. In this section we sketch out the impossibility proof for quantum bit commitment discovered independently by Mayers \cite{mayers97} and Lo and Chau \cite{lo97}.

Before going into the details of the quantum no-go argument let us sketch out the classical one.
Consider a protocol which is correct and hiding and suppose Alice follows the honest strategy for $\cval = 0$ until the commitment point. Since the protocol is correct Alice can proceed with the honest strategy and successfully unveil $\cval = 0$. What if she wants to cheat and unveil $\cval = 1$ instead? The protocol is hiding, which means that Bob does not know the value of her commitment. This implies that there \emph{exists} a particular strategy for Alice after the commitment point that will make him accept $\cval = 1$ (otherwise he could eliminate this possibility). The existence of such a strategy implies that Alice can unveil either value (with certainty), i.e.~achieve $p_{0} + p_{1} = 2$. This argument can be extended to show that in every classical protocol (which is correct) at least one of the parties can cheat with certainty. Such arguments are formalised using the notion of a transcript, which is simply the complete list of messages exchanged by Alice and Bob. It is clear that in the classical world each party can produce a transcript by copying all the messages to a private, auxiliary register. In the quantum world one cannot simply copy messages so it is not meaningful to talk about the transcript of a quantum protocol. That is why this simple classical argument does not apply to quantum protocols.

The quantum impossibility argument hinges on the fact that, without loss of generality, we can assume that Alice and Bob keep the entire quantum state pure until the commitment point. Since Alice and Bob share no correlations (these would count as a resource) and private randomness can be purified locally, we can assume that Alice and Bob start in a pure, product state. Then, all the measurements can be performed coherently (also known as \emph{keeping the measurements quantum}, see Section 1.3 of Ref.~\cite{colbeck06} for an explanation). This requires us to replace all the classical channels in the protocol by quantum channels. Since this might open up new, inherently quantum cheating strategies we must instruct each party to measure (coherently) each incoming message. This gives rise to a new (but equivalent from the security point of view) protocol, in which all the interactions happen at the quantum level and at the commitment point Alice and Bob share a pure state, which we denote by $\ket{\psi^{d}}_{AB}$.\footnote{This is one way of dealing with classical communication known as the \emph{indirect approach}. For more details on the indirect approach and also the alternative \emph{direct approach} see Ref.~\cite{brassard97}.} According to Definition \ref{df:hiding} the protocol is hiding if $\rho_{B}^{0} = \rho_{B}^{1}$. Unfortunately, due to Uhlmann's theorem \cite{uhlmann76} this implies that there exists a unitary $U_{A}$ acting on subsystem $A$ alone such that
\begin{equation*}
( U_{A} \otimes \mathbb{1}_{B} ) \ket{\psi^{0}}_{AB} = \ket{\psi^{1}}_{AB}.
\end{equation*}
In other words, Alice can switch between the two honest states by acting on her system alone, which allows her to unveil either bit (with certainty) and, therefore, renders the protocol completely insecure. Formally the impossibility can be stated in the following manner.
\begin{thm}
Any protocol which is perfectly correct and hiding allows Alice to cheat perfectly. In other words, there exists a cheating strategy for Alice which achieves $p_{0} = p_{1} = 1$.
\end{thm}

This simple argument applies only to the exact case where $\rho_{B}^{0} = \rho_{B}^{1}$ but it is easy to show that if $\rho_{B}^{0}$ is close to $\rho_{B}^{1}$ (i.e.~Bob finds it difficult to distinguish them) then Alice can cheat with high probability and trade-offs based on this idea were derived by Spekkens and Rudolph \cite{spekkens01}. Optimal bounds on quantum bit commitment have been found by Chailloux and Kerenidis \cite{chailloux11}.

While quantum mechanics does not allow for perfect bit commitment, it still beats classical protocols. As mentioned before every classical protocol is completely insecure against one of the parties. Quantum protocols, on the other hand, allow us to achieve some intermediate points, in which security is in some sense ``distributed'' between Alice and Bob.
\end{proof}
\chapter{Non-communicating models}
\label{chap:nc-models}
\emph{This chapter (excluding Section 3.1) is based on}
\paperA
It is well-known that interrogating suspects is more fruitful if they cannot communicate during the process, simply because coming up with two reasonable stories is more difficult than with one.\footnote{This only holds if the interrogation is an interactive and unpredictable process. If the suspects can predict all the possible questions in advance, nothing prohibits them from producing consistent answers.} This intuition was first made rigorous in the context of complexity theory but similar features can be seen in cryptography, in which non-communicating models allow us to implement primitives which would be otherwise forbidden. Care has to be taken, however, since the primitive implemented in the non-communicating model is usually subtly different (weaker) than the original one and, hence, might not be suitable for all applications.

Non-communicating models can be formalised using the concept of \emph{agents}. For example, if in the standard protocol Alice interacts with Bob, in the multiagent variant she might be required to interact with two distinct agents of Bob. In this case we think of Bob as being the main party, who decides on the strategy and briefs all his agents beforehand but steps back as soon as the protocol begins. During the protocol the agents follow the instructions but are not allowed to communicate with each other (or the main party). In this chapter we adopt the convention that if Alice (Bob) only needs to delegate one agent we do not make an explicit distinction between the main party and the agent. If more agents are involved we make a distinction by calling the agents \alice{1} and \alice{2} (\bob{1} and \bob{2}).

While it is entirely possible to discuss multiagent commitment schemes without any reference to complexity theory, we feel it is beneficial to explain where the idea of employing multiple agents originally came from. We explain how such models arose in the context of \index{interactive proof} \emph{interactive proofs} \cite{babai85, goldwasser85, benor88} and we discuss connections between \emph{zero-knowledge proofs} and cryptography \cite{goldreich86, benor88, goldreich08}. We also consider a multiagent variant of oblivious transfer and present a (trivial) protocol that implements it. This serves as a useful example to demonstrate that the functionality implemented in the multiagent scenario might differ significantly from the original primitive.

When it comes to analysing multiagent commitment schemes one of the major conceptual challenges is to establish a framework which encompasses all interesting schemes without being overly complicated. While for a particular scheme it is fairly straightforward to come up with an ad hoc treatment and security definition (which is often left implicit), a general framework is necessary for comparing various schemes. We propose such a framework based largely on results published in Ref.~\cite{kaniewski13}.

\noindent \textbf{Outline:} We start by explaining the concept of an interactive proof system and why employing multiple agents makes the model significantly more powerful. Then we discuss how such models can be used in the context of cryptography using \emph{distributed oblivious transfer} \cite{naor00} as an example. The last section of the chapter is dedicated to multiagent commitment schemes. We explain what kind of arrangements of agents are useful for the purpose of commitment schemes and propose how to quantify security in these new, multiagent models.
\section{Interactive proof systems}
The purpose of this section is to give a brief, non-technical introduction to the field of interactive proof systems. We are particularly interested in multiprover models\footnote{In the context of complexity theory, it is always the prover (one of the involved parties) who is required to delegate multiple agents.}, zero-knowledge proofs and their relation to cryptography. The notes of Oded Goldreich \cite{goldreich08} provide a thorough and accessible introduction to interactive proof systems. Readers interested in the early history of interactive proofs are referred to a wonderfully entertaining essay by L\'{a}szl\'{o} Babai \cite{babai90}.

Let us start with a motivating story. Suppose that Bob wants to be convinced that a certain statement is true. His own computational powers are limited (so he cannot simply verify the statement on his own) but he has access to an all-powerful computer called Alice. Unfortunately, Alice is a malicious machine and she will always assert that the statement is true (even if it is actually false). To make things worse she will even provide an incorrect proof, hoping that Bob will fall for it. Bob wants to interact with Alice in such a way that if she is honest and the proof is correct he accepts it, but if she misbehaves and outputs an incorrect proof her misconduct should be noticed. On a more fundamental level, we are asking whether it is possible to verify computations that are, by assumption, beyond our own capabilities.

Since we always assume that Alice knows the statement that Bob wants to be convinced of, she might simply produce a proof and send it to Bob. This coincides with the way we usually think of proofs as static, non-interactive objects, e.g.~something that can be published in a book. This is a valid solution but we know from everyday experience that the process of learning and understanding is often facilitated by the possibility of asking questions and receiving answers. The same phenomenon occurs in case of proofs and gives rise to the concept of an \emph{interactive proof}, which cannot be published in a book but can be explained in class. It turns out that allowing Alice and Bob to interact might significantly simplify certain proofs. Moreover, as counter-intuitive as it sounds it allows Alice to prove a statement without revealing anything about the actual proof. In complexity theory the setting described above is known as an \emph{interactive proof system}, with Alice being the \emph{prover} and Bob the \emph{verifier}, and was introduced independently by Babai \cite{babai85} and Goldwasser, Micali and Rackoff \cite{goldwasser85}.

To demonstrate the advantage of interactive proofs we need some concrete statements, which we will then construct (interactive) proofs for. It turns out that graph theory is a good source of intuitive and interesting examples. Given two graphs $G_{0}$ and $G_{1}$ we say that they are \emph{isomorphic} if we can map $G_{0}$ onto $G_{1}$ by simply relabelling the vertices. The problem of deciding whether two graphs are isomorphic is known as the \index{graph isomorphism} \emph{graph isomorphism problem} and we do not know how to solve it efficiently.\footnote{In fact, it is one of the few interesting problems that seem to sit in the middle between the ``easy problems'' (i.e.~the ones that can be solved efficiently) and the really hard ones (i.e.~the ones that we do not think can be solved efficiently, like the travelling salesman problem). Interested readers are encouraged to read a survey by Scott Aaronson on the distinction between the easy and the hard problems and how they relate to physical reality \cite{aaronson05}.}

This is exactly the setting we want to look at: Bob has two graphs $G_{0}$ and $G_{1}$ and he wants to know whether they are isomorphic. Since he is unable to solve this problem on his own, he asks Alice for help. If the graphs are isomorphic Alice simply sends Bob a valid relabelling (of vertices) and he verifies that it indeed maps $G_{0}$ onto $G_{1}$. This is an efficient, non-interactive proof. The problem becomes a bit more complex if the graphs are \emph{not} isomorphic. Alice could, of course, write down all possible relabellings and show that none of them achieves the goal but this does not really save Bob any computational effort. Verifying such a brute-force ``proof'' is not any easier than producing it. As of today, we do not know how to (generically) construct a non-interactive, efficient proof that two graphs are not isomorphic.

On the other hand, a beautifully simple solution exists if Alice and Bob are allowed to interact \cite{goldreich86}. Bob picks a random bit $b \in \{0, 1\}$, applies a random relabelling to $G_{b}$ and sends it to Alice, whose task is to guess $b$. If the graphs are not isomorphic then Alice can always correctly identify the original graph (she is all-powerful so she can simply try all possible relabellings) and successfully answer Bob's challenge. On the other hand, if the graphs are isomorphic then by applying a random relabelling Bob made the message that Alice receives independent of $b$.\footnote{More precisely, the probability distributions over graphs sent to Alice are identical for $b = 0$ and $b = 1$.} Hence, her probability of guessing $b$ correctly is exactly $\frac{1}{2}$. If we repeat this game multiple times the probability of correctly answering all the challenges decays exponentially. If Alice can reliably tell the two graphs apart, then Bob should be convinced that the two graphs are not isomorphic (except for exponentially small probability).

The connection between interactive proofs and cryptography appears when we impose an additional requirement that the proof should carry no information beyond \emph{the validity of the statement}. This concept introduced by Goldwasser, Micali and Wigderson goes under the name of a \index{zero-knowledge proof} \emph{zero-knowledge proof} \cite{goldreich86}. Note that this formulation sounds suspiciously similar to our initial motivation for commitment schemes in Section \ref{sec:cryptography}, in which Alice wants to prove to Bob that she knows something without revealing any additional information.

Let us go back to the problem of proving that two graphs are isomorphic. The obvious solution presented before is to provide a valid relabelling explicitly. Unfortunately, this reveals much more information than necessary: we want to prove the \emph{existence} of a relabelling rather to exhibit a particular one. Can we prove that two graphs are isomorphic in a zero-knowledge manner? Clearly, this cannot be done using a static proof but adding interactions helps as demonstrated below.\footnote{Requiring a static proof to be zero-knowledge reduces it to a trivial assertion ``this statement is true'', which the verifier will not find too convincing.}

Again, we assume that Alice knows both graphs $G_{0}$ and $G_{1}$. She applies a random relabelling to $G_{0}$ and sends it to Bob as $H$. Bob chooses a random bit $b$ and challenges Alice to reveal the relabelling that maps $H$ onto $G_{b}$. Clearly, if $G_{0}$ and $G_{1}$ are isomorphic Alice can always produce a valid answer. However, if they are not, she can find a valid answer to at most one of the two challenges (regardless of how she chose $H$). Again, by repeating this test a number of times Bob can be convinced that the two graphs are indeed isomorphic. Why is this proof zero-knowledge? This is clear if Bob acts honestly (i.e.~he chooses the bit $b$ at random), because then at the end of the protocol we can see $H$ as a random relabelling of $G_{b}$. This is something that Bob could have generated himself, hence, he has obtained no extra knowledge. The situation becomes more complex if we consider malicious Bob who might choose $b$ based on the graph $H$ he receives. A rigorous proof that this protocol remains zero-knowledge in this adversarial scenario is significantly more involved \cite{goldreich86}.

Once we know how to prove that two graphs are isomorphic in a zero-knowledge manner it is natural to ask what other statements can be proven in such a way. If we are happy to accept an extra computational assumption then it turns out that any statement that can be proven using a static proof can also be proven in a zero-knowledge fashion \cite{goldreich86}. Ben-Or, Goldwasser, Kilian and Wigderson realised that the computational assumption can be dropped by introducing an extra prover (who is not allowed to communicate with the first one during the protocol) and, in fact, their solution is quite simple \cite{benor88}. Before the protocol begins the provers generate a long, random string. During the protocol all the work is done by \prover{1}, while \prover{2} simply outputs segments of the shared randomness (randomly chosen by the verifier). Essentially, the goal is to convince the verifier that \prover{1} is using genuine, pre-existing randomness rather than generating (faking?) it on the spot. As a crucial step in the proof they propose a bit commitment scheme in the two-prover model and prove its security. They also present a construction for a particular flavour of distributed oblivious transfer.

The observation that computational security in Ref.~\cite{goldreich86} is used to provide commitment-like functionality is made explicit in Construction 2.4 from Goldreich's lecture notes \cite{goldreich08}, in which a generic zero-knowledge proof is constructed under the assumption that commitment functionality is available for free. This shows that the primitive of bit commitment establishes a connection between zero-knowledge proofs and multiprover models.

Multiprover models were introduced to remove computational assumptions in the context of zero-knowledge proofs but have since become an independent object of study in complexity theory. In fact, they have been shown to be significantly more powerful than the single-prover class \cite{fortnow94, babai91}. The quantum versions of these complexity classes have been proposed by allowing the provers to share entanglement either with \cite{kobayashi02} or without \cite{cleve04} quantum communication (with the verifier). The two classes have recently been shown to be equal \cite{reichardt13}.
\section{Applications in cryptography}
\label{sec:applications-in-cryptography}
We have seen that introducing multiple provers is useful in the context of interactive proofs and now we would like to see what can be gained in cryptography. Here, we consider a simple example and our main goal is to convince the reader that such models are not subject to the usual impossibility arguments and explain why that is the case.

Let us go back to the primitive of oblivious transfer explained by the example of an online movie service in Section \ref{sec:cryptography}. Alice has paid for one movie and wants to download it without revealing her choice to the company (Bob). In spirit of the previous section we consider a multiagent model in which Bob is required to delegate two agents, who interact with Alice but cannot communicate with each other. In the original primitive Bob should never find out which movie Alice chose to download. However, in the multiprover setting an interesting question arises: what happens to the agents after the protocol ends? Since it is hard to envision keeping them isolated until the end of time, we may first lean towards a model in which they are allowed to communicate after the protocol is finished. However, as pointed out in Appendix A.2 of Ref.~\cite{benor88}~in that case secure oblivious transfer is not possible. Temporary communication constraints are not sufficient as the standard no-go argument applies whenever the agents meet: if their combined knowledge does not allow them to deduce which message was retrieved, both messages must have leaked out to Alice.\footnote{It is possible to retain some security if we assume that the amount of communication between the provers is bounded \cite{benor88}.}

This encourages us to investigate the other extreme case in which the provers are not allowed to ever communicate again.\footnote{\label{footnote:split-forever}Note, however, a certain conceptual weakness of this model. The only manner in which Alice can ensure that the two agents never communicate again is to keep at least one of them isolated forever. But in that case it should not matter if that particular agent finds out which movie she wants to watch, hence, no cryptography is necessary. Note that keeping an agent isolated forever sounds morally wrong if we think of him as a human being but becomes more socially acceptable if we replace him by a disposable electronic device. Unfortunately, while in case of a human agent the assumption that he will only allow Alice to retrieve one movie is natural (an agent is capable of protecting the integrity of his laboratory), in case of an inanimate device this becomes essentially a technological assumption. Such devices have been proposed under the name of \emph{one-time memories} \cite{goldwasser08}.} Such a primitive is known as \index{distributed oblivious transfer} \emph{distributed oblivious transfer} \cite{naor00} (or \emph{symmetrically-private information retrieval} if we focus on the limit of a large number of messages \cite{gertner00, malkin00, gasarch04, kerenidis04}) and it admits the following simple solution based on secret sharing\footnote{We only use the simplest type of secret sharing in which an unknown string $x$ is split up into two shares: $s_{1} = x \oplus r$ and $s_{2} = r$, where $r$ is a string chosen uniformly at random. The two shares together allow us to reconstruct the string but it is easy to verify that having just one share conveys no information about $x$.}. For simplicity let us consider the case of Bob having only two messages $m_{0}, m_{1} \in \bs{n}$.
\phantomsection
\distributedOT
This protocol is secure because both \bob{1} and \bob{2} see Alice asking for a random message so neither of them obtains any knowledge about her choice. Moreover, it is easy to verify that no information is leaked about the message that Alice did not choose. Hence, this constitutes a secure multiagent implementation of oblivious transfer. However, as discussed in Section \ref{sec:simple-relativistic-protocols} we do not know how to usefully implement this protocol in a relativistic setting.

Why does such a protocol evade the standard no-go result\footnote{The intuition behind the standard no-go argument in the classical case is as follows. If at the end of the protocol Bob cannot tell which message Alice has decided to retrieve it must mean that through the interaction he has leaked both of them. In a world split only between Alice and Bob whatever Bob leaks becomes immediately available to Alice, which implies that she must have learnt both messages.}? It is important to realise that the no-go implicitly assumes that the whole world is split between Alice and Bob and there are no third parties: Alice can only be sure about the systems in her possession and everything else is fully controlled by Bob (this is equivalent to the assumption that the state shared between Alice and Bob is pure). In the multiagent model this must be modified as the state is now shared between Alice, \bob{1} and \bob{2}. Since \bob{1} and \bob{2} cannot communicate (their knowledge cannot be combined), the usual impossibility argument does not apply.
\section{Commitment schemes}
\label{sec:commitment-schemes}
The original zero-knowledge interactive proof proposed by Ben-Or et al.~relies on a multiagent bit commitment scheme \cite{benor88}. The proposed scheme is correct, hiding and $\varepsilon$-binding for $\varepsilon = \frac{1}{2}$. On the other hand, in Section \ref{sec:impossibility} we have argued that in the standard two-party model such schemes cannot exist.

Again, we must realise that the standard notion of a commitment scheme implicitly assumes that the protocol is executed by two parties only (no additional agents) and the impossibility result only holds for that case. Multiagent schemes require new security definitions and in general the usual limitations (proven in the standard two-party model) will not apply. While it is usually clear how security definitions should be extended to multiagent protocols, it is important to do it explicitly, as it helps to understand the exact nature of the primitives under consideration.

Requiring a party to delegate agents who are not allowed to communicate (which we also refer to as \emph{splitting}) restricts the range of actions available to that party. Clearly, this might only be useful for security purposes if communication constraints apply during the relevant party's ``turn to cheat''. According to the phase structure discussed in Section \ref{sec:bit-commitment}, this leads to either splitting Bob \emph{until the opening point} (which we call $\alpha$-split) or splitting Alice \emph{from the commitment point} ($\beta$-split).\footnote{Note that these are the \emph{minimal} splits, i.e.~they are necessary to evade the impossibility result. Later we will consider models which impose more than the minimal splits.} The two different splits are shown in Fig.~\ref{fig:splits}. Since we are interested in the fundamental possibilities and limitations, we will discuss protocols for both splits (and we will find that the resulting bit commitment primitives exhibit subtle differences).
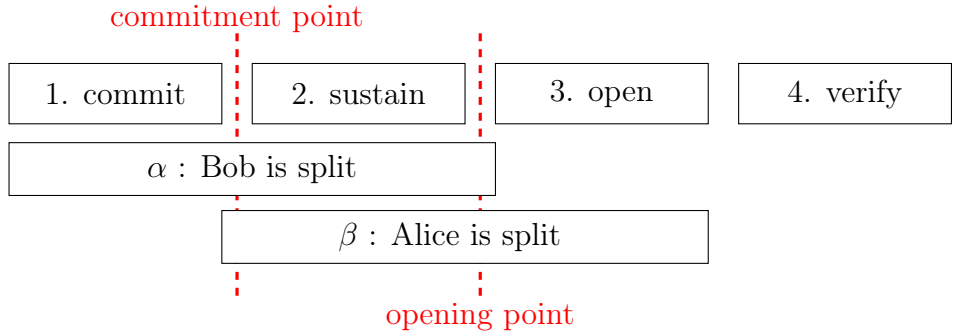
\begin{figure}
\centering
\begin{tikzpicture}[scale=1]
\phases{-2.36}
\draw[fill=white] (0, -0.25) rectangle (6.4, -0.95);
\node (alpha) at (3.2, -0.6) {$\alpha$ : Bob is split};
\draw[fill=white] (2.8, -1.15) rectangle (9.2, -1.85);
\node (beta) at (5.8, -1.5) {$\beta$ : Alice is split};
\end{tikzpicture}
\caption{The two types of minimal splits that are potentially useful for the purpose of commitment schemes.}
\label{fig:splits}
\end{figure}

Before proposing particular protocols, let us first adapt the security definitions to such multiagent scenarios. Since security requirements state what the dishonest party should not be able to achieve, it is clear that we need a new definition of the hiding property in the $\alpha$-split and a new definition of the binding property in the $\beta$-split.

A commitment scheme is hiding if at the opening point Bob remains ignorant about Alice's commitment. In the $\alpha$-split model at the opening point there are two agents \bob{1} and \bob{2}, who are not allowed to communicate. Similarly to the case of distributed oblivious transfer if we require that even their \emph{combined} knowledge does not allow them to learn the commitment, then the standard no-go applies (i.e.~Alice can cheat with certainty). However, we can instead require that \emph{neither} \bob{1} \emph{nor} \bob{2} can guess the commitment, which leads to a natural condition closely resembling Definition \ref{df:hiding}.
\begin{df}
\label{df:hiding-two}
A multiagent bit commitment protocol is \textbf{hiding} if all pairs of states $(\sigma_{ A B_{1} B_{2} }^{0}, \sigma_{ A B_{1} B_{2} }^{1})$ that \bob{1} and \bob{2} can enforce at the opening point satisfy
\begin{equation*}
\sigma_{B_{1}}^{0} = \sigma_{B_{1}}^{1} \nbox{and} \sigma_{B_{2}}^{0} = \sigma_{B_{2}}^{1},
\end{equation*}
where $\sigma_{B_{c}}^{d} = \tr_{ A B_{1 - c} } \sigma_{A B_{1} B_{2}}^{d}$.
\end{df}
\noindent This definition means that neither of the agents has learnt anything about Alice's commitment but it says nothing about their combined knowledge. This naturally leads to the following protocol based on secret sharing. For bit commitment protocols we adopt the convention that $a$ ($b$) denotes to the private randomness of Alice (Bob) while $x$ ($y$) are the messages sent during the protocol by Alice (Bob). Note that the labels $x$ and $y$ are used regardless of whether the parties are honest or not.
\phantomsection
\secretsharingnc
This protocol is so simple that neither party can even attempt to cheat! In the commit phase whatever combination of messages Alice decides to produce, it will correspond to an honest commitment, which she has no influence over once the messages are received by \bob{1} and \bob{2} (i.e.~this is exactly the commitment point of the protocol). On the other hand, if Alice is honest then both \bob{1} and \bob{2} receive a uniform bit (regardless of the value of $\cval$) so the protocol is hiding according to Definition~\ref{df:hiding-two}.

A potential drawback of this protocol is that in certain scenarios, we might want to give Alice the right to \emph{refuse} opening a commitment. Clearly, in this protocol this could only be done if \bob{1} and \bob{2} were never allowed to communicate again, which is a problematic assumption (cf.~footnote \ref{footnote:split-forever} in Section \ref{sec:applications-in-cryptography}).

It turns out that this feature (of allowing Alice to keep the commitment value hidden forever) is much easier to achieve in the $\beta$-split model. As explained in Section \ref{sec:formal-definition} security for honest Bob can be quantified through a game in which Alice performs some generic strategy until the commitment point and is then challenged (by an external referee) to open either $\cval = 0$ or $\cval = 1$ with equal probabilities. The commitment is considered secure if she is not able to win this game with probability significantly exceeding $\frac{1}{2}$ (an honest commitment achieves at least $\frac{1}{2}$ as long as the scheme is correct). In case of \alice{1} and \alice{2} performing the open phase in a non-communicating fashion, we need to specify who actually receives the challenge. Is it both \alice{1} and \alice{2} or just, say, \alice{1}? The former scenario might arise if \alice{1} and \alice{2} despite not being able to communicate with each other might still receive messages from an external source (it might be easier to isolate the agents from each other than from the external world). For example, what they attempt to unveil might depend on the latest stock market news. It turns out that this distinction is important and gives rise to two different models, which we call \index{global command model} \emph{global} and \index{local command model} \emph{local command}, respectively. This choice does not affect \alice{1}: in both cases her cheating behaviour is determined by two compatible strategies\footnote{See Section~\ref{sec:formal-definition} for an explanation what it means for two strategies to be compatible.}, just like in the standard single-agent model. However, the allowed behaviour of \alice{2} is affected. In the global command model she chooses two compatible strategies but in the local command she may only choose one (since she never actually finds out what they are trying to unveil).
\begin{df}
\label{df:binding-two}
Let $(\sigma_{AB}^{0}, \sigma_{AB}^{1})$ be a pair of states that \alice{1} and \alice{2} can enforce at the opening point given that \alice{1} employs two compatible strategies and \alice{2} employs
\begin{itemize}
\item \textbf{local command:} only one strategy (regardless of the value of $\cval$).
\item \textbf{global command:} two compatible strategies.
\end{itemize}
Let $(\Phi^{\textnormal{cheat}, 0}_{A \to P}, \Phi^{\textnormal{cheat}, 1}_{A \to P})$ be opening maps of the form
\begin{itemize}
\item \textbf{local command:} $\Phi^{\textnormal{cheat}, \cval}_{A \to P} = \Phi^{\textnormal{cheat}, \cval}_{A_{1} \to P_{1}} \otimes \Phi^{\textnormal{cheat}}_{A_{2} \to P_{2}}$.
\item \textbf{global command:} $\Phi^{\textnormal{cheat}, \cval}_{A \to P} = \Phi^{\textnormal{cheat}, \cval}_{A_{1} \to P_{1}} \otimes \Phi^{\textnormal{cheat}, \cval}_{A_{2} \to P_{2}}$.
\end{itemize}
Define $p_{\cval}$ to be the probability that Alice's attempt to unveil $\cval$ is accepted by Bob
\begin{equation*}
p_{\cval} = \tr \big( M_{\textnormal{accept}} \big[ \Phi^{\textnormal{cheat}, \cval}_{A \to P} (\sigma_{AB}^{\cval}) \big] \big).
\end{equation*}
A multiagent bit commitment protocol is called \textbf{$\varepsilon$-binding} in the \textbf{local/global command} model if for all states $(\sigma_{AB}^{0}, \sigma_{AB}^{1})$ and for all opening maps $(\Phi^{\textnormal{cheat}, 0}_{A \to P}, \Phi^{\textnormal{cheat}, 1}_{A \to P})$ allowed by the model we have
\begin{equation*}
p_{0} + p_{1} \leq 1 + \varepsilon.
\end{equation*}
\end{df}
\noindent To see that the distinction between the two models is important, note that the local command model allows for the following trivial bit commitment protocol.
\phantomsection
\begin{prot}{3}{Bit commitment in the local command model}
\label{prot:local-command}
\begin{enumerate}
\item (commit) Alice sends $\cval$ to \alice{1} and \alice{2}.
\item (open) \alice{1} sends $x_{1} = d$ and \alice{2} sends $x_{2} = d$ to Bob.
\item (verify) Bob verifies that $x_{1} = x_{2}$.
\end{enumerate}
\end{prot}
In the local command model dishonest \alice{1} receives the challenge and knows what they are trying to unveil but \alice{2} does not. Since the value they are challenged to unveil is chosen uniformly at random, she cannot guess it too well. In fact, the best she can do is to always output the same value, which essentially corresponds to an honest commitment. Here, security is a direct consequence of the fact that \alice{2} does not know what she is supposed to be unveiling. It is clear that in this protocol Alice is committed as soon as communication between \alice{1} and \alice{2} is forbidden. Protocol \hyperref[prot:local-command]{3} is secure in the local command model but it is easy to see that it is completely insecure in the more stringent global command model. Does there exist a protocol that remains secure in the global command model?

It turns out that no classical protocol in the $\beta$-split model can meet this requirement and the argument is similar to the standard no-go for bit commitment. Let us assume that the protocol is correct and hiding, i.e.~it allows \alice{1} and \alice{2} to make an honest commitment, which until the opening point leaks no information to Bob and the opening is always accepted. Suppose \alice{1} and \alice{2} honestly commit to $\cval = 0$. Clearly, unveiling $\cval = 0$ in the open phase is easy but since Bob cannot rule out Alice's commitment to $\cval = 1$, there must also exist a sequence of messages from \alice{1} and \alice{2} which will make him accept $\cval = 1$. Since now both of them know what they are trying to unveil, this strategy can be implemented and the protocol is completely insecure.

The intuitive argument presented above makes a subtle assumption that all information that Alice and Bob exchange in the commit phase is available to \emph{both} \alice{1} and \alice{2} in the open phase. There are two ways of invalidating this assumption.
\begin{enumerate}
\item Make the information that Bob shares with Alice in the commit phase quantum. Then, by the no-cloning theorem \cite{wootters82} it will not (in general) be possible for both \alice{1} and \alice{2} to have an exact copy.
\item Strengthen the communication constraint, i.e.~require that only \alice{1} takes part in the commit phase while \alice{2} is already isolated.
\end{enumerate}
The first solution was explored under the name of \emph{quantum relativistic bit commitment} by Kent \cite{kent11, kent12} and a rigorous security analysis of the latter protocol (including experimental imperfections like noise and losses) can be found in Chapter \ref{chap:transmitting} of this thesis. Moreover, two new protocols based on different features of quantum theory were recently proposed \cite{adlam15a, adlam15b}. The second solution corresponds to the original proposal of Ben-Or et al.~\cite{benor88}, further developed in Refs.~\cite{simard07, crepeau11}. Since the protocol is simple and intuitive we present it here but we defer rigorous security analysis until Chapter \ref{chap:multiround}.

The bit commitment scheme proposed in Ref.~\cite{benor88} is sufficient from the complexity point of view but it is not the most convenient formulation for cryptographic purposes. As described in Section \ref{sec:cryptographic-protocols} in cryptography it is convenient to have a family of protocols with a parameter $n \in \amsbb{N}$ which can be chosen to guarantee the desired level of security. Such a protocol was presented under the name \texttt{simplified-BGKW} (\texttt{sBGKW}) in Refs.~\cite{simard07, crepeau11}. In this case \alice{1} and \alice{2} are not allowed to communicate throughout the entire protocol. Let $a$ and $b$ be $n$-bit strings chosen uniformly at random by Alice and Bob, respectively.
\phantomsection
\label{prot:sBGKW-nc}
\sBGKWnc
(The bit-by-string multiplication was defined in Section \ref{sec:strings-of-bits}.) In a protocol which requires \alice{1} and \alice{2} to be already isolated in the commit phase, it becomes important whether the value of the commitment must be known to both or just one of them. In this particular case \alice{1} can single-handedly decide on the value of the commitment.\footnote{It is interesting to note that \alice{2} (the only agent of Alice who takes part in the open phase) does not need to know the value she is unveiling.} Correctness of the protocol is easy to verify while the hiding property is a simple consequence of the fact that the message that \alice{1} sends to Bob in the commit phase is ``one-time padded'' with a uniformly random string. On the other hand, we intuitively see that the binding property is a direct consequence of the communication constraint between \alice{1} and \alice{2} (cheating would be easy if \alice{2} knew $b$). Moreover, note that in this protocol \alice{2} can simply refuse to take part in the open phase and then the commitment made by \alice{1} (if she indeed followed the protocol in the commit phase) will remain secret forever. In this aspect, this protocol differs significantly from Protocol \hyperref[prot:secret-sharing-nc]{2}. This difference will have quite interesting consequences when we consider relativistic variants of these protocols in Section \ref{sec:simple-relativistic-protocols}.
\chapter{Relativistic protocols}
\label{chap:relativistic-protocols}
\emph{This chapter is based on}
\paperA
In Chapter \ref{chap:nc-models} we saw that communication constraints are useful in a variety of situations. In particular, they enable us to implement cryptographic primitives which are not possible otherwise. Non-communicating models are widely studied in computer science but unless one can justify such communication constraints, they should be treated on equal footing with other technological limitations and we already know that assumptions concerning computational power or storage capabilities make two-party cryptography possible.

How could Alice possibly ensure that \bob{1} and \bob{2} cannot communicate? Well, in principle she could lock each of them up in separate rooms. First of all, \bob{1} and \bob{2} might not be happy with such a solution but even if they are, how does she ensure that the rooms are perfectly shielded from the outside world? Does this not lead to yet another technological assumption?

One way out of the vicious circle of technological assumptions is relativity. Imposing an upper bound on the speed at which information spreads implies that communication between any two distinct locations incurs some minimal delay (proportional to the distance between them). This gives rise to temporary communication constraints, which rely solely on the correctness of the theory of relativity. It is worth pointing out that this is the \emph{only} feature of relativity used in relativistic cryptography.

It is important to stress the difference between non-communicating and relativistic protocols. In a non-communicating protocol (like the ones discussed in Chapter \ref{chap:nc-models}) we first explicitly specify communication constraints and then the interactions between the agents. On the other hand, in a relativistic protocol one cannot simply impose such arbitrary communication constraints. Instead, they must \emph{arise} from the arrangement of agents in space and appropriately chosen timing of the protocol. Therefore, the description of the protocol must specify where and when each interaction takes place and then the resulting communication constraints may be used to prove security. Note that not every combination of communication constraints might be achieved in this model, e.g.~if Alice simultaneously communicates with \bob{1} and \bob{2}, they must be able to communicate too.

To the best of our knowledge, the idea of combining relativity and quantum mechanics for cryptographic purposes first appeared in writing in a summary article by Gilles Brassard and Claude Cr\'{e}peau \cite{brassard96, crepeau96}, who attributed it to Louis Salvail. The foundations were laid by Adrian Kent (first relativistic commitment schemes \cite{kent99, kent05}) and Roger Colbeck (proposals for various flavours of coin-tossing and impossibility results for secure two-party computation \cite{colbeck06, colbeck06a, colbeck07a}). More recently, significant interest was sparked by position-verification schemes \cite{kent11a, buhrman14, tomamichel13, unruh14, ribeiro15}. A relativistic quantum key distribution scheme has also been proposed \cite{radchenko14}.

The defining feature of relativistic cryptography is the requirement that different phases of the protocol take place at distinct locations. With the appropriate choice of timing this imposes communication constraints, which are no longer due to technological limitations but result directly from the physical theory (security of such schemes is often advertised to be ``guaranteed by the laws of physics''). Unfortunately, this desirable feature comes at a price. Communication constraints guaranteed by relativity are \emph{temporary}, which means that we must leave the neat and tidy world of non-communicating models, in which we are free to impose arbitrary communication constraints, and enter the complex world of relativistic models, in which communication is only \emph{delayed} rather than \emph{forbidden}.\footnote{It is useful to contrast this aspect of relativistic models with the non-communicating case. In the non-communicating world we can choose whether or not the agents are allowed to communicate once the protocol is finished and both options are equally valid. In the relativistic setting there is only one natural solution, which lies somewhere in between the two extremes: the agents can communicate but their communication is not instantaneous.} The analysis of such scenarios becomes significantly more involved if the agents are required to handle quantum information (or when dishonest parties use quantum devices to cheat in a classical protocol). In fact, this has led to interesting and fundamental questions about how to \emph{define} the location of a quantum system. Consider the process of teleportation \cite{bennett93}, in which a quantum state $\rho$ located initially at one place is reconstructed at another place by using entanglement (pre-shared between the two locations) and sending classical data. Interestingly enough, during this procedure there is a period of time when the state seemingly ``ceases to exist'', in the sense that there is \emph{no location} at which any information about $\rho$ can be \emph{immediately} extracted. Where is the state then? This counter-intuitive phenomenon is captured operationally through the task of \emph{summoning} recently investigated by Kent \cite{kent13, kent12a}, Hayden and May \cite{hayden12}.

\noindent \textbf{Outline:} In this chapter we first show how some of the protocols discussed in Chapter \ref{chap:nc-models} can be implemented in the relativistic setting and what limitations such a ``translation'' brings about. We then present an explicit procedure for mapping a relativistic protocol onto a communication-constrained model. We show that in the fully classical setting communication-constrained models can be further mapped onto non-communicating models and we discuss why such a simplifying reduction cannot be done when quantum information is involved. Finally, we discuss the power and limitations of relativistic cryptography.
\section{Non-communicating schemes in the relativistic setting}
\label{sec:simple-relativistic-protocols}
We start by considering how some of the non-communicating schemes discussed in Chapter \ref{chap:nc-models} can be implemented in the relativistic setting. Since the communication constraints imposed by relativity are temporary, the resulting commitment schemes cannot guarantee everlasting security.\footnote{Unless the parties keep communicating, see Section \ref{sec:limitations-relativistic-cryptography} for more details.} Understanding exactly the ``mode of failure'', i.e.~how different commitment schemes ``expire'', provides valuable insight into the power of relativistic cryptography.

The only realistic implementation of a relativistic protocol involves stationary agents exchanging information at the speed of light. The protocol specifies a set of locations and each party is required to delegate a (stationary) agent to each location. All communication between Alice and Bob occurs locally, i.e.~between agents occupying the same location, and for simplicity we assume that all local communication is instantaneous.\footnote{Note that this is the only reasonable model. If an agent of Alice were to send a message to a far-away agent of Bob, she would either have to ``escort'' the message until it reaches the agent of Bob (which is equivalent to placing an extra agent at the receiving end as in our model) or she would let the message out unguarded, in which case there is no guarantee that the message will not be intercepted by some other agent of Bob at some earlier location.} Communication between distinct agents of the same party is unrestricted (and assumed to be secure) but must respect the speed-of-light constraint (for simplicity we take $c = 1$).\footnote{Security of internal communication can be ensured by using teleportation to transmit quantum states and information-theoretic encryption (one-time pad) for classical information. Alternatively, we can assume that distinct agents occupy different locations within the same laboratory (e.g.~the model of two long laboratories in a single spatial dimension as in Section 1.7.2 and Fig.~1.6 of Ref.~\cite{colbeck06}).}

All examples considered in this thesis take place in a single spatial dimension labelled by $x$ and as usual time is labelled by $t$. All considerations in this chapter extend in a straightforward fashion to more spatial dimensions but we are not aware of any examples in which this gives any advantage. We label the locations by integers and refer to the agents occupying Location $k$ as \alice{k} and \bob{k}. For convenience we define the following three locations.
\begin{center}
\setlength{\tabcolsep}{0.7cm}
\begin{tabular}{l l}
\hline
Location 0 & $x = 0$\\
Location 1 & $x = -1$\\
Location 2 & $x = 1$\\
\hline
\end{tabular}
\end{center}
It is important to bear in mind that in relativistic protocols all the interactions are performed by agents occupying well-defined locations. We avoid referring to the main party (whose location during the protocol is not specified) as it might create the impression that there exists some higher form of life that is able to instantaneously communicate with all its agents. The existence of such a being is forbidden by relativity and would indeed render all the relativistic protocols insecure.

Let us first present a relativistic variant of Protocol \hyperref[prot:secret-sharing-nc]{2}. Before the protocol begins \alice{1} and \alice{2} must be provided with a random bit $a \in \{0, 1\}$ (e.g.~generated by \alice{0} at $t = -1$).
\phantomsection
\secretsharingrel
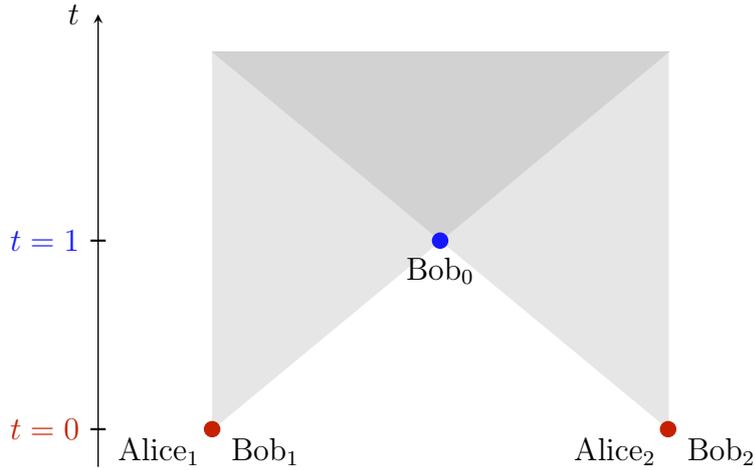
\begin{figure}
\centering
\begin{tikzpicture}[scale=1, line width=0.5]
\timeaxis{5.5}
\draw [-, thick] (-0.1, 2.5) to (0.1, 2.5);
\node[left, color=mycolour1] at (-0.1, 2.5) {$t = 1$};
\node at (0.8, -0.3) {\alice{1}};
\node at (2.2, -0.3) {\bob{1}};
\node at (6.8, -0.3) {\alice{2}};
\node at (8.2, -0.3) {\bob{2}};
\draw [fill, lightgray] (1.5, 0) -- (7.5, 5) -- (1.5, 5);
\draw [fill, lightgray] (7.5, 0) -- (7.5, 5) -- (1.5, 5);
\draw [fill, darkgray] (4.5, 2.5) -- (7.5, 5) -- (1.5, 5);
\draw[fill, color=mycolour3] (1.5, 0) circle [radius=0.1];
\draw[fill, color=mycolour3] (7.5, 0) circle [radius=0.1];
\draw[fill, color=mycolour1] (4.5, 2.5) circle [radius=0.1];
\node at (4.5, 2.1) {\bob{0}};
\end{tikzpicture}
\caption{Spacetime diagram for Protocol \hyperref[prot:secret-sharing-rel]{5}. The red dots represent the commit phase while the blue dot represents the open phase. The shaded areas correspond to the future light cones of the interactions in the commit phase.}
\label{fig:secret-sharing-rel}
\end{figure}
In a sense this protocol is easier to understand than the original, non-communicating version (cf.~the spacetime diagram in Fig.~\ref{fig:secret-sharing-rel}). It is clear that Alice becomes committed at $t = 0$ (the commitment point) and that the commitment becomes known to Bob (\bob{0} to be more specific) at $t = 1$ (the opening point), hence, the commitment is valid for $t \in (0, 1)$. On the other hand, in the non-communicating variant it is not a priori clear when (and why!) communication constraints vanish and the commitment opens. Just like in Protocol \hyperref[prot:sBGKW-nc]{4}, \alice{1} can single-handedly decide on the commitment and the choice can be delayed until $t = 0$.

Our second example is a relativistic variant of Protocol \hyperref[prot:sBGKW-nc]{4}. This time \alice{1} and \alice{2} must share a random $n$-bit string $a \in \{0, 1\}^{n}$.
\phantomsection
\label{prot:sBGKW-rel}
\sBGKWrel{t \in (0, 2)}{= 2}
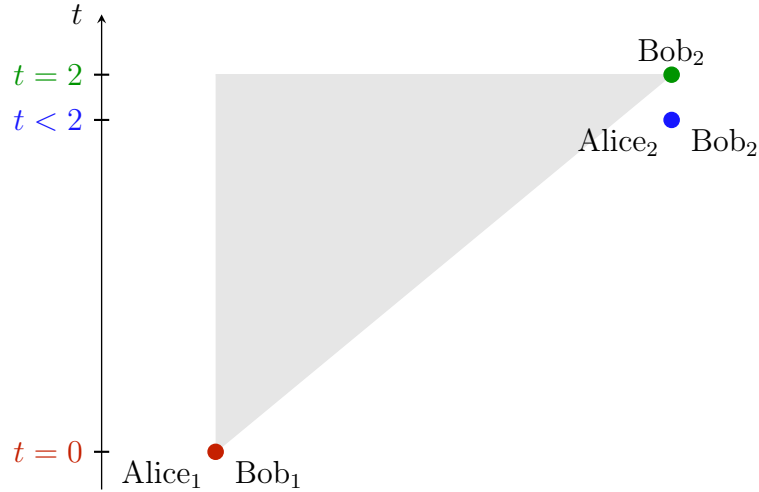
\begin{figure}
\centering
\begin{tikzpicture}[scale=1, line width=0.5]
\timeaxis{5.8}
\node at (0.8, -0.3) {\alice{1}};
\node at (2.2, -0.3) {\bob{1}};
\draw [fill, lightgray] (1.5, 0) -- (7.5, 5) -- (1.5, 5);
\draw[fill, color=mycolour3] (1.5, 0) circle [radius=0.1];
\draw [-, thick] (-0.1, 4.4) to (0.1, 4.4);
\node[left, color=mycolour1] at (-0.1, 4.4) {$t < 2$};
\draw[fill, color=mycolour1] (7.5, 4.4) circle [radius=0.1];
\node at (6.8, 4.1) {\alice{2}};
\node at (8.2, 4.1) {\bob{2}};
\draw [-, thick] (-0.1, 5) to (0.1, 5);
\node[left, color=mycolour2] at (-0.1, 5) {$t = 2$};
\draw[fill, color=mycolour2] (7.5, 5) circle [radius=0.1];
\node at (7.5, 5.3) {\bob{2}};
\end{tikzpicture}
\caption{Spacetime diagram for Protocol \hyperref[prot:sBGKW-rel]{6}. The red dot represents the commit phase, the blue dot represents the open phase, the green dot corresponds to the point at which \bob{2} verifies the commitment.}
\label{fig:sBGKW-rel}
\end{figure}
Just like in Protocol \hyperref[prot:sBGKW-nc]{4}, \alice{1} can choose the value of the commitment single-handedly and this choice can be delayed until $t = 0$ (\alice{2} does not need to know it). The requirement that the open phase happens at $t < 2$ ensures that no signals can be sent between the commit and open phases (cf.~Fig.~\ref{fig:sBGKW-rel}). It is easy to see that \alice{2} could cheat perfectly if she knew $b$ so the timing must be chosen such that $b$, which is announced by \bob{1} at $t = 0$, is not available to \alice{2} during the open phase. Under this condition the relativistic protocol and the original, non-communicating version are equivalent as far as security is concerned.

Note that in this relativistic scheme there is always a non-zero delay in verifying the commitment but it can be made arbitrarily small.\footnote{The possibility of immediate verification of the opening would imply that the commit phase and the open phase are \emph{not} space-like separated. Then, there would have been enough time for $b$ to reach \alice{2}, which would render the protocol insecure.} Whether this constitutes a severe limitation or not depends on the particular application but this feature, which appears often in relativistic protocols, should be always kept in mind, especially when considering composability (i.e.~executing a relativistic scheme as a subroutine in a longer procedure).

It is instructive to consider what happens if for some reason the open phase does not happen in the interval $t \in (0, 2)$. At $t = 2$ dishonest \alice{2} receives $b$ (sent by dishonest \alice{1} at $t = 0$) and at this point she can provide a valid proof for either value of $\cval$, which makes the protocol completely insecure. In other words, the commitment expires at $t = 2$ and no opening should be accepted at (or after) that point. If \alice{2} does not perform the opening during $t \in (0, 2)$, Bob will never find out whether \alice{1} made an honest commitment, let alone its value.

Having presented two cases in which non-communicating protocols can be turned in a straightforward manner into relativistic protocols, let us briefly discuss one case in which such a simple translation is not possible. Recall Protocol \hyperref[prot:distributed-OT]{1} for distributed oblivious transfer presented in Section \ref{sec:applications-in-cryptography}. Security of this protocol hinges on the assumption that \bob{1} and \bob{2} cannot communicate from the beginning of the protocol until the end of time. We know that permanent communication constraints cannot be enforced by relativity so we cannot hope for everlasting security but temporary security is not immediately ruled out. To restrict communication between \bob{1} and \bob{2} we would have to place them at distant locations, as usual accompanied by their communication partners \alice{1} and \alice{2}. During the protocol each Alice receives a single message and the message that they actually want to obtain is the \texttt{XOR} of the two. Unfortunately, the earliest point at which the transmitted message might be reconstructed coincides with the point at which the information gathered by \bob{1} and \bob{2} can be recombined to reveal which message Alice chose to retrieve (the spacetime diagram is essentially identical to the one shown in Fig.~\ref{fig:secret-sharing-rel}). We could have hoped for some finite interval during which Alice already knows the message but Bob still remains ignorant about her choice but in case of Protocol \hyperref[prot:distributed-OT]{1} this is not possible. This shows that not all non-communicating protocols can be mapped directly onto the relativistic setting in a meaningful way.
\section{Explicit analysis of relativistic protocols}
\label{sec:explicit-analysis}
We have seen how simple non-communicating protocols can be implemented in the relativistic setting but so far the security analysis was rather ad hoc. While this is sufficient for simple schemes, for more complex protocols (involving more agents and/or multiple rounds, which might be necessary to achieve improved security features, e.g.~longer commitment time) a systematic approach is desirable. In this section we provide a solution to a subclass of these problems and discuss the complications arising while dealing with the most general case.

A relativistic protocol is classical if all the messages exchange between agents of Alice and agents of Bob are classical. A protocol is quantum if there is at least one quantum message. Since classical protocols are designed to be executed by classical parties they should not require the agents of Alice or Bob to perform quantum operations \emph{in the honest scenario}. However, this cannot be ruled out in the dishonest case and it is natural to study the security of classical protocols against quantum adversaries.

We first consider classical protocols and we show that analysing the dishonest scenario is equivalent to a certain multiplayer game with \emph{partial communication constraints}\footnote{Similar models have been previously studied from the foundational point of view under the name of \emph{time-ordered models} \cite{gallego14} or \emph{correlation scenarios} \cite{fritz12, fritz14}.} played by the agents of the dishonest party.\footnote{Note that this procedure is not specific to commitment schemes and applies to any relativistic protocol in which cheating can be cast as a game.} In Section \ref{sec:classical-players} we show that if the agents are restricted to classical strategies, the situation is equivalent to a multiplayer game of non-communicating players. In Section \ref{sec:quantum-players} we mention some complications that arise when analysing such games against quantum players. Finally, in Section \ref{sec:quantum-relativistic-protocols} we discuss briefly the problems related to quantum relativistic protocols.

For the sake of concreteness let us consider the case of honest Alice. Since the agents of Alice follow the protocol, we might think of them as an omnipresent referee, who interacts with the agents of Bob. The following simple procedure explains how to turn a relativistic protocol into a multiplayer game (similar to those described in Section \ref{sec:multiplayer-games}) such that winning the game is equivalent to cheating in the protocol.
\begin{enumerate}
\item Identify all points of spacetime at which the agents of Alice and Bob interact, order them by their time coordinate and label by (positive) integers.\footnote{For interactions occurring at the same time the order does not matter.} Without loss of generality we assume that every interaction consists of a \emph{challenge} from Alice followed by a \emph{response} from Bob, which for the $j\th$ interaction are denoted by $c_{j}$ and $r_{j}$, respectively.\footnote{If the protocol requires more rounds of communications in a sequence, consider them as separate interactions.} Let $(x_{j}, t_{j})$ be the spacetime coordinates of the $j\th$ interaction and let $n$ be the total number of interactions in the protocol. Construct the \index{communication graph} \emph{communication graph} $G = ([n], E)$, in which each vertex corresponds to an interaction and the set of (directed) edges is determined by the causality constraints. More precisely, $(j, k)$ is an edge iff $k$ is in the future light cone of $j$
\begin{equation*}
(j, k) \in E \iff \abs{x_{k} - x_{j}} \leq t_{k} - t_{j}.
\end{equation*}
Note that $G$ is an oriented and acyclic graph.
\item Without loss of generality the challenge issued by Alice in the $j\th$ interaction is a deterministic function of some pre-shared randomness (represented by a random variable $Z$) and the previous responses of Bob. For a particular value of the random variable $Z = z$ we have
\begin{equation*}
c_{j} = f_{j}(z, r_{1}, r_{2}, \ldots, r_{j - 1}).
\end{equation*}
(Clearly, $f_{j}$ might not depend on the responses which do not belong to the past light cone of the $j\th$ interaction but to keep the notation simple we do not indicate this restriction explicitly.) The collection of functions $f_{1}, f_{2}, \ldots, f_{n}$ together with the probability distribution of $Z$ fully determines the distribution of challenges issued by Alice.
\item Deciding whether a cheating attempt is successful, i.e.~the predicate function for the game, might without loss of generality be taken to depend only on the initial randomness and the responses from Bob, i.e.~$V(z, r_{1}, r_{2}, \ldots, r_{n})$.
\end{enumerate}
This procedure provides us with three components: the communication graph, the distribution of challenges and the predicate function. Clearly, this triple defines a multiplayer game in which communication, instead of being completely forbidden, is restricted. More specifically, we can identify the $j\th$ interaction with a player $\cP_{j}$ and starting from $j = 1$ every player takes part in the following procedure.
\begin{enumerate}
\item Player $\cP_{j}$ receives messages sent by previous players.
\item Player $\cP_{j}$ receives a challenge $c_{j}$ and issues a response $r_{j}$.
\item Player $\cP_{j}$ might send a message to any player $\cP_{k}$ such that $(j, k) \in E$.
\end{enumerate}
At the end all the answers are collected and the predicate function $V$ is evaluated to determine whether the game is won or lost.

Since quantum communication can be implemented by teleportation (and we do not impose any restrictions on the amount of entanglement shared by the players) we can assume all communication to be classical.

Let us summarise what we have accomplished so far. We have started from a classical relativistic protocol and we have turned it into an equivalent classical multiplayer game with communication constraints. Note that by classical we mean that all the challenges and responses are classical but this does not prevent the players from using quantum systems to generate them. As discussed in the next section, the case of quantum players (i.e.~players using quantum systems to generate their classical responses) is significantly harder to analyse than the case of classical players.

Note that multiplayer games with communication constraints include many interesting scenarios as special cases. For example if $E = \emptyset$ (i.e.~the communication graph $G$ has no edges) we recover the standard scenario of multiplayer non-communicating games. The other extreme case is when the players satisfy a ``total order'', i.e.~$(j, k) \in E \iff k > j$, which is equivalent to a single player responding to a sequence of challenges. This is exactly the scenario that arises in classical non-relativistic two-party cryptography.\footnote{These two special cases have also been studied if the challenges and/or responses are quantum. For some recent results on two-player quantum games see Refs.~\cite{regev13, cooney15} while for sequential quantum games see papers on quantum non-relativistic two-party cryptography listed in Section \ref{sec:quantum-cryptography}.}
\subsection{Classical players}
\label{sec:classical-players}
Any strategy available to classical players can be expressed as a convex combination of deterministic strategies. Since randomness can be shared among the players in advance and their goal is to achieve the optimal winning probability (which is determined by a fixed and known function), we might restrict our attention to deterministic strategies. What is the most general strategy of $\cP_{j}$, i.e.~what is his response allowed to depend on? Clearly, it might depend on the challenge that he receives $c_{j}$ but it might also depend on messages received by him from the ``previous'' players. This seems to complicate the situation, since these might be arbitrary and depend on anything that was available to the sender, etc. However, a simple observation allows us to simplify this seemingly complicated structure. Since the message sent by a particular player is a function of the data available to him, he could alternatively send the whole data set to the receiver, who can then generate the message himself. This leads to the simple conclusion that it is optimal\footnote{Optimal in the sense of spreading information to the largest number of players, certainly not in terms of efficiency.} to broadcast any challenge received from the referee to all eligible players. Then the response of $\cP_{j}$ becomes a deterministic function of all the challenges \emph{in his past}. If we supply every player with these additional inputs, they no longer need to communicate. This reduction works because there exists a trivial but optimal communication strategy for the players, namely ``broadcast everything''.
\begin{obs}
\label{obs:games-equivalence}
Let $\cG_{1}$ be the game in which $\cP_{j}$ receives $c_{j}$ and the allowed communication pattern is specified by $G = ([n], E)$. Let $\cG_{2}$ be the game in which $\cP_{j}$ receives $\{ c_{k} \}_{k \in \cS_{j}}$ where
\begin{equation*}
\cS_{j} := \{ k \in [n] :  (k, j) \in E \}
\end{equation*}
and no communication is allowed $G = ([n], \emptyset)$. The sets of strategies available to the classical players in games $\cG_{1}$ and $\cG_{2}$ are identical.
\end{obs}
\noindent This observation plays a crucial role in the analysis of a multiround classical relativistic bit commitment protocol in Chapter \ref{chap:multiround}.
\subsection{Quantum players}
\label{sec:quantum-players}
We have seen that for classical players games with communication constraints can be reduced to fully non-communicating games. What happens if we attempt such a reduction for quantum players?

As one might expect the quantum case is not so simple and it is instructive to consider the following example. Consider a game of three players where $G$ contains only one edge, $E = \{ (1, 2) \}$. Clearly, $\cP_{3}$ cannot communicate with the other players but his presence is necessary to hope for a quantum advantage.\footnote{\label{footnote} Without $\cP_{3}$ we would have a game equivalent to asking a sequence of classical questions to a single player and in such games no quantum advantage is possible.} The response of $\cP_{1}$ is determined by some measurement he performs on his quantum system (and the measurement setting depends on the challenge $c_{1}$). Then he passes whatever is left of the quantum system along with the classical messages $c_{1}$ and $r_{1}$ to $\cP_{2}$ who then receives $c_{2}$ and completely measures the quantum system to obtain $r_{2}$.

This scenario is difficult to analyse because the measurement performed by $\cP_{1}$ affects how much information $\cP_{2}$ (who learns a new piece of information $c_{2}$) might extract from the state. This problem goes under the name of sequential measurements and is currently an active area of research \cite{heinosaari15}. Note that in the classical setting such trade-offs do not exist: generating the response for the current round does not affect the information that might be sent to other players.

To the best of our knowledge, this is the simplest example in which finding the quantum value of a classical game cannot be reduced to any of the previously studied models. Interestingly enough, this is precisely the scenario which arises when analysing security of the multiround protocol presented in Chapter \ref{chap:multiround} against quantum adversaries.
\subsection{Quantum relativistic protocols}
\label{sec:quantum-relativistic-protocols}
While presenting the procedure to map a relativistic protocol onto a communication-constrained model, we have explicitly restricted ourselves to classical protocols. This was mainly to avoid the trouble of specifying the most general way in which the referee may choose the challenge. While this is conceptually not difficult, formalising these notions would be quite cumbersome. In particular, we would need to explicitly define the Hilbert spaces corresponding to the referee's memory, the ``message'' space, define the class of operations the referee might use to prepare the challenge, argue what the new predicate is, etc.\footnote{Note that as a special case we must recover the standard model for quantum protocols of Yao \cite{yao95}, which puts a lower bound on the complexity of the description.}

While mapping a quantum relativistic protocol onto a quantum game is not difficult, we do not know how to analyse the resulting ``quantum'' games. Without aiming for full generality let us just sketch out two quantum games, which demonstrate difficulties that might arise in these scenarios.

The first game is a variation on the example presented in the previous section. Basically, by making the first challenge $c_{1}$ quantum we can eliminate $\cP_{3}$ without trivialising the problem. Consider a game of two players $\cP_{1}, \cP_{2}$ such that $E = \{ (1, 2) \}$. The challenge received by $\cP_{1}$ is an unknown quantum state and he is required to give a classical response $r_{1}$. $\cP_{1}$ passes the remaining quantum state together with his classical response to $\cP_{2}$, who receives a new (classical) challenge and must produce another classical response. Clearly, the information extractable in the second round depends on the measurement performed in the first one, hence, the two rounds cannot be decoupled and mapped onto a non-communicating model. Games of this type arise when considering quantum non-relativistic protocols.

The second game is arguably the simplest manifestation of no-cloning. Consider a game of three players whose communication graph contains two edges: $E = \{ (1, 2), (1, 3) \}$. The challenge issued to $\cP_{1}$ is an unknown quantum state and no response is required. Players $\cP_{2}$ and $\cP_{3}$ are then challenged to unveil one out of two incompatible properties of the original state. Clearly, this would be easy if each of them could hold a copy of the original state but this is forbidden by the no-cloning theorem. One solution is for $\cP_{1}$ to measure one of the two properties and send the classical outcomes to $\cP_{2}$ and $\cP_{3}$. However,  this only allows them to answer one of the challenges correctly. This is exactly the quantum feature used in Kent's quantum relativistic bit commitment protocol \cite{kent12}, whose complete analysis can be found in Chapter~\ref{chap:transmitting}.
\section{Limitations of relativistic cryptography}
\label{sec:limitations-relativistic-cryptography}
We have made contributions towards understanding of relativistic commitment schemes but in general the exact power of relativistic cryptography is not yet completely understood. The goal of this section is to summarise what is known to be possible and what the known limitations are. It turns out that between the two there is a sizeable piece of land yet to be discovered.

Let us start with the simplest task: coin tossing. The trivial classical protocol (described for example as Protocol 2.3 in Ref.~\cite{colbeck06}), in which \alice{1} sends a random bit to \bob{1} and simultaneously \bob{2} sends a random bit to \alice{2} and the outcome of the coin toss is the \texttt{XOR} of the two bits, achieves perfect security and is easily implemented in the simplest relativistic model with just two locations. More sophisticated flavours of coin tossing, in which Alice and Bob can partially influence the bias of the coin, are also possible \cite{colbeck06}.

The situation becomes a bit more complicated when it comes to bit commitment. All the commitment protocols we have discussed so far expire in some way: in case of Protocol \hyperref[prot:secret-sharing-rel]{5} the commitment automatically opens, in case of Protocol \hyperref[prot:sBGKW-rel]{6} the commitment vanishes. In principle these commitments can be made arbitrarily long but only at the price of increasing the spatial separation between the sites. This is clearly not a desirable solution, since in practice we are restricted to a fixed region of space (we have easy access to the surface of the Earth but going beyond that seems somewhat impractical). Can we achieve an arbitrarily long commitment while performing the protocol in a finite region of space? Let us first consider protocols in which the commit phase only requires a finite amount of communication, i.e.~at some point the communication stops and no more messages need to be exchanged until the open phase. It is clear that at that point both parties could bring all their systems together and within some period of time (proportional to the size of the accessible region of space) we would be back in the standard scenario, in which the usual trade-offs apply. Hence, arbitrarily long commitment cannot be achieved by a protocol with a bounded number of messages in the commit phase. What about protocols in which the agents keep communicating? The multiround scheme presented in Chapter \ref{chap:multiround} belongs to this class and implements bit commitment which is secure against classical adversaries and can be made arbitrarily long. We conjecture that the protocol remains secure against quantum adversaries but we currently do not have a proof.

Commitments with a finite period of validity (which at some point expire) have been previously studied under the name of \emph{timed commitments}. For example Boneh and Naor~\cite{boneh00} study commitments which fail in the same way as Protocol \hyperref[prot:secret-sharing-rel]{5}, i.e.~after some fixed time the committed value is revealed to Bob.\footnote{Their motivation comes from schemes which only offer computational security. Such schemes can always be forced open given enough time and computational power.} They show that such commitments can be used for contract signing or honesty-preserving auctions. Generally speaking, such temporary secrecy is sufficient if the goal is to force parties to act simultaneously (in the sense that their respective actions should not depend on each other) even if the communication model is sequential. Broadbent and Tapp considered the task of secure voting, for which such commitments would be sufficient \cite{broadbent08}. Timed commitments that vanish (i.e.~the commitment is no longer valid but the committed value, if there was one, remains secret) can be used in similar situations if we want to give Alice more power to protect her privacy. This type of commitment might also be used in multiparty protocols which are robust against a certain fraction of dishonest parties (then any party that refuses to open the commitment would be declared dishonest).

To see the limitations of relativistic commitment schemes it is instructive to investigate whether they can be used to implement other, more powerful primitives. For example, a well-known construction shows how to use bit commitment and quantum communication to implement oblivious transfer \cite{bennett92, yao95}. Are relativistic schemes suitable for this canonical construction? Without going into too many details let us describe one important feature of this construction. At some point of the procedure Bob is required to make several commitments. Later, Alice asks Bob to open a random subset of them but the rest he keeps untouched. Security for Bob hinges on the fact that some commitments remain closed, which rules out relativistic schemes that expire by opening (like Protocol \hyperref[prot:secret-sharing-rel]{5}). The commitments that vanish without revealing any information (like Protocol \hyperref[prot:sBGKW-rel]{6}) might seem perfectly suited for the task. However, a simple conceptual problem referred to as \index{classical certification} \emph{classical certification} or \index{retractability} \emph{retractability} arises \cite{kent12b, colbeck06}. Basically, the canonical construction implicitly assumes that every commitment (including the unopened ones) \emph{has a value}. While it might not be immediately clear what it means for an unopened commitment to have a value, this concept can be made rigorous and it is possible to show that relativistic protocols do not satisfy this property. A more detailed discussion on the issue of classical certification of relativistic commitment schemes and an explicit example how it renders the canonical construction insecure can be found in Appendix \ref{app:classical-certification}. While this is by no means a proof that no relativistic commitment scheme can be used for the canonical construction, we have at least ruled out the ones considered so far.

What about implementing relativistic oblivious transfer directly without going through canonical construction? This possibility seems unlikely due to the following informal argument. In every (correct) oblivious transfer protocol at some fixed point Alice must receive the chosen message. At this point the knowledge of Alice (or Bob) might be scattered among all their agents but in an attempt to cheat it can be (within some finite time) gathered at one location and then the usual impossibility results apply. Investigating whether this intuition can be turned into a rigorous argument would be an interesting research problem for two reasons: it would require us to propose a meaningful definition of relativistic oblivious transfer and the actual impossibility result (if true) would determine an important boundary point of quantum relativistic cryptography. Alternatively, one can relax the requirements and look for a protocol whose security is only guaranteed for a finite period of time. Such protocols might exist and it would be interesting to know how useful they are.
\chapter{Bit commitment by transmitting measurement outcomes}
\label{chap:transmitting}
\emph{This chapter is based on}
\paperB
In the previous chapters we have seen why non-communicating models are useful in cryptography and how such models can be implemented using relativity. In this chapter we use the tools introduced before and present a complete analysis of a particular quantum relativistic bit commitment protocol proposed by Kent \cite{kent12}. Our initial approach to this problem, presented in Ref.~\cite{kaniewski13}, relied on some tools from non-asymptotic quantum information theory and a recently discovered uncertainty relation \cite{tomamichel11, tomamichel12}. Later, however, we found another, simpler approach (which does not explicitly use any uncertainty relation), which we then extended to apply to experimental implementations \cite{lunghi13}. In this chapter we only present the latter, superior method. Note that similar techniques found applications to other interesting problems in quantum cryptography \cite{tomamichel13}. (During the lifetime of these projects an independent security analysis of the same protocol was provided by Croke and Kent \cite{croke12} and an independent experiment was performed by Liu, Cao, Curty, Liao, Wang, Cui, Li, Lin, Sun, Li, Zhang, Zhao, Chen, Peng, Zhang, Cabello and Pan~\cite{liu14}).

\noindent \textbf{Outline:} We start by proving security of the original protocol. Our methods are robust, as they also apply to the case of imperfect state preparation and noisy transmission. This would be sufficient to prove security of an implementation that uses a single-photon source, a lossless quantum channel and perfect detectors (or devices which are good approximations thereof). Unfortunately, such devices are not available at the moment, hence, we modify the protocol so it can be implemented using currently available devices, in our case a weak-coherent source and inefficient and noisy detectors. We describe the new security model, extend the previous security analysis and determine the minimum requirements on the honest devices that allow for a secure implementation of the protocol. We also present an explicit calculation of the security parameter of the protocol. We finish this chapter by giving a brief overview of an experiment performed between Geneva and Singapore in collaboration with an experimental group at the University of Geneva.
\section{The original protocol}
\label{sec:the-original-protocol}
We use the following notation for the BB84 states
\begin{equation}
\label{eq:bb84-states}
\ket{\psi_{x}^{\theta}} = H^{\theta} \ket{x},
\end{equation}
where $x, \theta \in \{0, 1\}$. For a sequence of BB84 states described by $x, \theta \in \{0, 1\}^{n}$ we use
\begin{equation}
\label{eq:bob-measurement}
\ket{x^{\theta}} = \bigotimes_{k = 1}^{n} \ket{\psi_{x_{k}}^{\theta_{k}}}.
\end{equation}
For a particular basis string $\theta \in \bs{n}$ we define $S$ and $T$ to be the rounds in which Alice encoded her qubits in the computational and Hadamard basis, respectively,
\begin{gather*}
S = \{k \in [n] : \theta_{k} = 0\},\\
T = \{k \in [n] : \theta_{k} = 1\}.
\end{gather*}
The protocol proposed by Kent \cite{kent12} uses the same locations as described in Section \ref{sec:simple-relativistic-protocols} and Fig.~\ref{fig:tmo} shows the relevant spacetime diagram. The parameter $n \in \amsbb{N}$ determines the usual cost vs.~security trade-off (see Section \ref{sec:cryptographic-protocols}), while $\delta \in [0, 1)$ specifies the noise tolerance of the protocol. Recall that $\dham(\cdot, \cdot)$ is the fractional Hamming distance (defined in Section \ref{sec:strings-of-bits}).
\phantomsection
\backreportingfalse
\tmoprot
\begin{figure}
\centering
\begin{tikzpicture}[scale=1, line width=0.5]
\timeaxis{3.5}
\draw [-, thick] (-0.1, 2.5) to (0.1, 2.5);
\node[left, color=mycolour1] at (-0.1, 2.5) {$t = 1$};
\draw [fill, lightgray] (1.5, 0) -- (1.5, 2.5) -- (4.5, 0) -- (7.5, 2.5) -- (7.5, 0);
\node at (3.8, -0.3) {\alice{0}};
\node at (5.2, -0.3) {\bob{0}};
\draw[fill, color=mycolour3] (4.5, 0) circle [radius=0.1];
\node at (0.8, 2.2) {\alice{1}};
\node at (2.2, 2.2) {\bob{1}};
\node at (6.8, 2.2) {\alice{2}};
\node at (8.2, 2.2) {\bob{2}};
\draw[fill, color=mycolour1] (1.5, 2.5) circle [radius=0.1];
\draw[fill, color=mycolour1] (7.5, 2.5) circle [radius=0.1];
\end{tikzpicture}
\caption{Spacetime diagram for Protocol \hyperref[prot:tmo]{7}. The red dot represents the commit phase while the blue dots represent the open phase. The shaded area corresponds to the past light cones of the events of the open phase.}
\label{fig:tmo}
\end{figure}
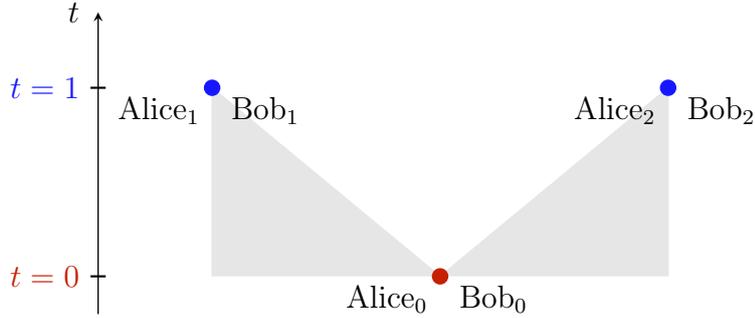
\noindent Before we proceed with a complete analysis let us mention a couple of unusual features of the protocol.

First of all, it seems that Alice becomes committed without actually sending any information to Bob (no communication from Alice to Bob happens until the open phase). Is that possible? How can Bob be sure that for $t > 0$ Alice is indeed committed?

To answer this question it is instructive to consider a slight variation on Protocol \hyperref[prot:tmo]{7}, in which the open phase is delayed until $t = 2$. Clearly, in this case \alice{0} could keep the quantum states untouched until $t = 1$ and only then perform the measurement. For this modified protocol Alice only becomes committed at $t = 1$, which does not occur \emph{immediately after} the communication in the commit phase is over (as was the case in all the previous protocols).

This quantum protocol challenges the preconception that the timing of the commitment point is determined by the interactions in the commit phase. Strangely enough, in this case the commitment point is determined by the timing of and locations used in the open phase. In fact, it is easy to see that the commitment point is determined by the latest point in the common past of the openings performed by \alice{1} and \alice{2}.

Protocols discussed in Section~\ref{sec:simple-relativistic-protocols} result in commitments which are only valid for a finite amount of time (in case of Protocol \hyperref[prot:secret-sharing-rel]{5} the commitment at some point automatically opens while in case of Protocol \hyperref[prot:sBGKW-rel]{6} at some point security for honest Bob is lost). It is interesting to note that the quantum commitment scheme we consider now does not expire. A commitment initiated at $t = 0$ may be opened at any $t = t_{\textnormal{open}} > 1$ but a successful opening only demonstrates that Alice was committed for $t \in (t_{\textnormal{open}} - 1, t_{\textnormal{open}})$. It is important to stress that at $t = t_{\textnormal{open}} - 1$ Alice is not yet committed, so we must take the commitment point to be $t = t_{\textnormal{open}} - 1 + r$, where $r > 0$ is an arbitrarily small (but non-zero) constant.

Finally, let us point out that verifying whether an opening should be accepted or not is not immediate. Moreover, in contrast to Protocol \hyperref[prot:sBGKW-rel]{6}, the delay cannot be made arbitrarily small. Since the conditions that Bob needs to verify depend on data unveiled at both opening locations, the delay is proportional to the distance between them.
\subsection{Correctness}
Correctness in the noiseless setting is clear by inspection while for an experimental implementation the only relevant quantity turns out to be the total bit-flip error rate between (honest) Alice and Bob (this rate includes contributions coming from imperfect state preparation, transmission noise and measurement errors). For simplicity we assume that noise acts independently on every qubit. Let $\err$ be the bit-flip error rate, i.e.~the probability of obtaining the wrong outcome despite the qubit having been prepared and measured in the same basis.\footnote{If the error probabilities are different for the two bases, we take the larger value to be on the safe side.} The protocol is asymptotically correct (i.e.~the probability of honest parties aborting decays exponentially in $n$) if
\begin{equation}
\label{eq:tmo-correctness}
\err < \delta.
\end{equation}
Note that this depends solely on the numerical value of $\err$ and not on the exact effects that contribute to it.
\subsection{Security for honest Alice}
Since Bob receives no information before the open phase, he remains completely ignorant about Alice's commitment and so the protocol is hiding.
\subsection{Security for honest Bob}
\label{sec:tmo-honest-bob}
To investigate security for honest Bob we first turn the original prepare-and-measure scheme of Kent into an entanglement-based scheme (equivalence for security purposes was explained in Section \ref{sec:equivalence}). In the entanglement-based formulation instead of generating BB84 states \bob{0} generates EPR pairs, he keeps one half of each (to be measured later) and sends the other halves to \alice{0}. The most general attack performed by \alice{0} during the commit phase is to perform an isometry that ``splits up'' the entire quantum system received from \bob{0} into two parts, which she then sends to \alice{1} and \alice{2}. In the open phase, \alice{1} and \alice{2} measure their respective quantum systems and pass the outcomes to \bob{1} and \bob{2}, respectively.

Since we want to prove security with respect to the global command variant of Definition \ref{df:binding-two} let us spell out how this definition applies to the protocol.\footnote{Security in the local command can be achieved by the trivial Protocol \hyperref[prot:local-command]{3}, cf.~Section \ref{sec:commitment-schemes}.}

First, we need to characterise the set of states that \alice{0} might enforce at the commitment point. While at $t = 0$ the state is (without loss of generality) only shared between \alice{0} and \bob{0}, at the commitment point ($t = r > 0$) the share of \alice{0} is already explicitly split up into two parts (that will reach \alice{1} and \alice{2} in time for the open phase) and this partitioning is essential to determine the commitment. We denote the relevant subsystems by $A_{1}$ and $A_{2}$ (even if at the commitment point these subsystems are not with the agents \alice{1} and \alice{2} yet). It is straightforward to see that at the commitment point any tripartite state $\sigma_{B A_{1} A_{2}}$ can be enforced as long as the marginal state held by \bob{0} remains unchanged, i.e.
\begin{equation*}
\tr_{ A_{1} A_{2} } \sigma_{ B A_{1} A_{2} } = \Big( \frac{ \mathbb{1} }{2} \Big)^{\otimes n}.
\end{equation*}
Interestingly enough, our proof does not make use of this property. In other words, the protocol remains secure even if \alice{0} were allowed to provide an arbitrary tripartite state compatible with the measurements that \bob{0} will later perform (i.e.~the subsystem of \bob{0} must consist of $n$ qubits).

In the open phase \alice{1} and \alice{2} must provide proofs, which are just classical strings of length $n$. Hence, the opening maps correspond to measurements. Each of \alice{1} and \alice{2} has two different measurements used to unveil the two different values of the commitment and we denote the measurement operators of \alice{1} (\alice{2}) attempting to unveil $\cval$ by $\{P_{y}^{\cval}\}_{y \in \bs{n}}$ \big($\{Q_{y}^{\cval}\}_{y \in \bs{n}}$\big).
Since we do not impose any constraints on the local dimensions of $A_{1}$ and $A_{2}$ we may without loss of generality assume that these measurements are projective. The fact that \alice{2} might use different measurements for $d = 0$ and $d = 1$ indicates that we work in the global command model.

If $\theta \in \bs{n}$ is the basis string (picked by \bob{0} uniformly at random), then his measurement is described by operators $\{ \ketbraq{x^{\theta}} \}_{x \in \bs{n}}$ as defined in Eq.~\eqref{eq:bob-measurement}. The commitment is accepted if the strings supplied by \alice{1} and \alice{2} are consistent with the classical outcomes obtained by \bob{0}. This condition can be written as a projector acting on the original tripartite state and it is easy to see that the projector $\Pi_{d}^{\theta}$ corresponding to \bob{0} accepting the unveiling of $d$ for a particular basis string $\theta$ equals
\begin{align*}
\Pi_{0}^{\theta} &= \sum_{x \in \bs{n}} \ketbraq{x^{\theta}} \otimes \sum_{\substack{y \in \bs{n}\\ \dham(x_{S}, y_{S}) \leq \delta}} P_{y}^{0} \otimes Q_{y}^{0},\\
\Pi_{1}^{\theta} &= \sum_{x \in \bs{n}} \ketbraq{x^{\theta}} \otimes \sum_{\substack{y \in \bs{n}\\ \dham(x_{T}, y_{T}) \leq \delta}} P_{y}^{1} \otimes Q_{y}^{1},
\end{align*}
and the three registers correspond to the subsystems held by \bob{0}, \alice{1} and \alice{2}, respectively (and the latter two result from an isometry applied by \alice{0} to subsystem $A_{0}$, which she received in the commit phase). These projectors require that \alice{1} and \alice{2} unveil the same string, which on the relevant subset ($S$ for $d = 0$ and $T$ for $d = 1$) is $\delta$-close (in terms of fractional Hamming distance) to the string obtained by Bob. To calculate the probability of successfully unveiling $d$ we must average over all possible basis choices
\begin{equation*}
p_{d} = 2^{-n} \sum_{\theta \in \bs{n}} \tr( \Pi_{d}^{\theta} \sigma_{ B A_{1} A_{2} } ).
\end{equation*}
In fact, our technique allows us to generalise the definition~\eqref{eq:bb84-states} to any pair of bases on a qubit
\begin{equation*}
\braket{\psi_{0}^{0}}{\psi_{1}^{0}} = \braket{\psi_{0}^{1}}{\psi_{1}^{1}} = 0.
\end{equation*}
This requirement comes directly from the fact that the equivalence between prepare-and-measure and entanglement-based schemes as presented in Section \ref{sec:equivalence} only applies if the average state is fully mixed. It turns out that the final bound depends only on the overlap between the bases
\begin{equation}
\label{eq:overlap}
c := \max_{x, y} \abs{ \braket{\psi_{x}^{0}}{\psi_{y}^{1}} },
\end{equation}
which is a well-known measure of incompatibility used extensively in the study of uncertainty relations \cite{deutsch83, maassen88, berta09, tomamichel11}.
\begin{prop}
\label{prop:tmo-binding}
Let
\begin{equation*}
\lambda_{0} = \frac{1 + c}{2} \nbox{and} \lambda_{1} = \frac{1 - c}{2},
\end{equation*}
where $c$ is the overlap as defined in Eq.~\eqref{eq:overlap}. For any strategy of dishonest Alice, the probabilities of Bob accepting the commitment satisfy
\begin{equation*}
p_{0} + p_{1} \leq 1 + \varepsilon,
\end{equation*}
for
\begin{equation}
\label{eq:tmo-security-bound}
\varepsilon =
\begin{cases}
\lambda_{0}^{n} &\nbox{for} \delta = 0,\\
\exp \bigg(- \frac{1}{2} \Big( \sqrt{\lambda_{1}} - \frac{\delta}{\sqrt{\lambda_{1}}} \Big)^{2} n \bigg) &\nbox{for} 0 < \delta < \lambda_{1}.
\end{cases}
\end{equation}
\end{prop}
\begin{proof}
Let us write the sum out explicitly
\begin{equation}
\label{eq:p0-plus-p1}
p_{0} + p_{1} = 2^{-n} \sum_{\theta \bsr{n}} \tr \big( [\Pi_{0}^{\theta} + \Pi_{1}^{\theta}] \sigma_{ B A_{1} A_{2} } \big).
\end{equation}
Adding up the two projectors (for a particular value of $\theta$) gives
\begin{equation*}
\Pi_{0}^{\theta} + \Pi_{1}^{\theta} = \sum_{x \bsr{n}} \ketbraq{x^{\theta}} \otimes \bigg[ \sum_{\substack{y \bsr{n}\\ \dham(x_{S}, y_{S}) \leq \delta}} P_{y}^{0} \otimes Q_{y}^{0} \quad + \sum_{\substack{y \bsr{n}\\ \dham(x_{T}, y_{T}) \leq \delta}} P_{y}^{1} \otimes Q_{y}^{1} \bigg].
\end{equation*}
The terms in the square bracket can be upper bounded by replacing one of the measurement operators by the identity matrix. Therefore,
\begin{gather*}
\sum_{\substack{y \bsr{n}\\ \dham(x_{S}, y_{S}) \leq \delta}} P_{y}^{0} \otimes Q_{y}^{0} \quad + \sum_{\substack{y \bsr{n}\\ \dham(x_{T}, y_{T}) \leq \delta}} P_{y}^{1} \otimes Q_{y}^{1} \leq \sum_{\substack{y \bsr{n}\\ \dham(x_{S}, y_{S}) \leq \delta}} P_{y}^{0} \otimes \mathbb{1} \quad + \mathbb{1} \otimes \sum_{\substack{y \bsr{n}\\ \dham(x_{T}, y_{T}) \leq \delta}} Q_{y}^{1}\\
\leq \mathbb{1} \otimes \mathbb{1} \quad + \sum_{\substack{y \bsr{n}\\ \dham(x_{S}, y_{S}) \leq \delta}} P_{y}^{0} \otimes \sum_{\substack{z \bsr{n}\\ \dham(x_{T}, z_{T}) \leq \delta}} Q_{z}^{1},
\end{gather*}
where the last step follows from the following operator inequality
\begin{equation}
\label{eq:operator-inequality}
A \otimes \mathbb{1} + \mathbb{1} \otimes B = \mathbb{1} \otimes \mathbb{1} + A \otimes B - (\mathbb{1} - A) \otimes (\mathbb{1} - B)\leq \mathbb{1} \otimes \mathbb{1} + A \otimes B,
\end{equation}
which holds for any $0 \leq A, B \leq \mathbb{1}$. Therefore,
\begin{equation}
\label{eq:projector-inequality}
\Pi_{0}^{\theta} + \Pi_{1}^{\theta} \leq \sum_{x \bsr{n}} \ketbraq{x^{\theta}} \otimes \mathbb{1} \otimes \mathbb{1} + \Pi_{c}^{\theta} = \mathbb{1} \otimes \mathbb{1} \otimes \mathbb{1} + \Pi_{c}^{\theta},
\end{equation}
where
\begin{equation*}
\Pi_{c}^{\theta} = \sum_{x \bsr{n}} \ketbraq{x^{\theta}} \otimes \sum_{\substack{y \bsr{n}\\ \dham(x_{S}, y_{S}) \leq \delta}} P_{y}^{0} \otimes \sum_{\substack{z \bsr{n}\\ \dham(x_{T}, z_{T}) \leq \delta}} Q_{z}^{1}
\end{equation*}
is a projector for the ``cross-game'', in which \alice{1} has to unveil a string consistent with $\cval = 0$ and \alice{2} has to unveil a string consistent with $\cval = 1$.\footnote{Our security analysis goes through a thought experiment in which \alice{1} and \alice{2} are challenged to unveil \emph{different} bits and in the current method this connection is made through the operator inequality \eqref{eq:operator-inequality}. Interestingly enough, our previous method relies on the same idea but expressed at the level of no-signalling probability distributions (see Lemma V.1 of Ref.~\cite{kaniewski13}).} Combining Eqs.~\eqref{eq:p0-plus-p1} and \eqref{eq:projector-inequality} gives
\begin{equation*}
p_{0} + p_{1} \leq 1 + 2^{-n} \sum_{\theta \bsr{n}} \tr \big( \Pi_{c}^{\theta} \sigma_{ B A_{1} A_{2} } \big) = 1 + \tr \big( \ave{\Pi_{c}^{\theta}} \sigma_{ B A_{1} A_{2} } \big) \leq 1 + \norm{\ave{\Pi_{c}^{\theta}}},
\end{equation*}
where $\ave{\cdot}$ denotes averaging over $\theta$, i.e.~$\ave{\Pi_{c}^{\theta}} = 2^{- n} \sum_{\theta \bsr{n}} \Pi_{c}^{\theta}$, and $\norm{\cdot}$ denotes the Schatten $\infty$-norm (defined in Section \ref{sec:linear-algebra}). Changing the order of summation in $\Pi_{c}^{\theta}$ gives
\begin{equation*}
\Pi_{c}^{\theta} = \sum_{y, z \bsr{n}} \sum_{\substack{x \bsr{n}\\ \dham(x_{S}, y_{S}) \leq \delta\\ \dham(x_{T}, z_{T}) \leq \delta}} \ketbraq{x^{\theta}} \otimes P_{y}^{0} \otimes Q_{z}^{1}.
\end{equation*}
Now, it is clear that only the $x$-dependent part needs to be averaged:
\begin{equation*}
\ave{\Pi_{c}^{\theta}} = 2^{-n} \sum_{\theta \bsr{n}} \Pi_{c}^{\theta} = \sum_{y, z \bsr{n}} B_{yz} \otimes P_{y}^{0} \otimes Q_{z}^{1},
\end{equation*}
where
\begin{equation*}
B_{yz} = 2^{-n} \sum_{\theta \bsr{n}} \sum_{\substack{x \bsr{n}\\ \dham(x_{S}, y_{S}) \leq \delta\\ \dham(x_{T}, z_{T}) \leq \delta}} \ketbraq{x^{\theta}}.
\end{equation*}
Since the product $P_{y}^{0} \otimes Q_{z}^{1}$ yields orthogonal projectors, we have
\begin{equation}
\label{eq:cross-projector-norm}
\norm{\ave{\Pi_{c}^{\theta}}} = \max_{y, z \bsr{n}} \norm{B_{yz}}.
\end{equation}
To identify values of $y$ and $z$ which maximise the norm we take a closer look at the matrices $B_{yz}$. For every $\theta$ define $u(\theta) \bsr{n}$ to be the string that satisfies $[u(\theta)]_{S} = y_{S}$ and $[u(\theta)]_{T} = z_{T}$. Relabelling $x \mapsto x \oplus u(\theta)$ yields
\begin{equation*}
B_{yz} = 2^{-n} \sum_{\theta \bsr{n}} \sum_{\substack{x \bsr{n}\\ \wham(x_{S}) \leq \delta\\ \wham(x_{T}) \leq \delta}} \ketbraq{(x \oplus u(\theta))^{\theta}}.
\end{equation*}
The constraints on the second sum can be relaxed by noting that $\wham(x_{S}) \leq \delta$ and $\wham(x_{T}) \leq \delta$ imply $\wham(x) \leq \delta$. Therefore,
\begin{equation*}
B_{yz} \leq B_{yz}' = 2^{-n} \sum_{\theta \bsr{n}} \sum_{\substack{x \bsr{n}\\ \wham(x) \leq \delta}} \ketbraq{(x \oplus u(\theta))^{\theta}},
\end{equation*}
which makes the second sum independent of $\theta$. Hence, the summation over $\theta$ can be performed first and due to the tensor product structure we have
\begin{equation*}
\ketbraq{x^{\theta}} = \bigotimes_{k = 1}^{n} \ketbraq{\psi_{x_{k}}^{\theta_{k}}}
\end{equation*}
and
\begin{equation*}
\sum_{\theta \bsr{n}} \ldots \iff \bigotimes_{k = 1}^{n} \sum_{\theta_{k} \in \{0, 1\}} \ldots.
\end{equation*}
Therefore,
\begin{equation*}
2^{-n} \sum_{\theta \bsr{n}} \ketbraq{(x \oplus u(\theta))^{\theta}} = \bigotimes_{k = 1}^{n} \rho_{x_{k} \oplus y_{k}, x_{k} \oplus z_{k}}.
\end{equation*}
where
\begin{equation*}
\rho_{b, c} = \frac{1}{2} (\ketbraq{\psi_{b}^{0}} + \ketbraq{\psi_{c}^{1}})
\end{equation*}
for $b, c \in \{0, 1\}$. Note that $\rho_{b, c} + \rho_{1 - b, 1 - c} = \mathbb{1}$ so they are diagonal in the same basis. Therefore, without loss of generality we can write
\begin{equation*}
\rho_{b, c} = \sum_{t \in \{0, 1\}} \lambda_{t}^{b \oplus c} \ketbraq{e_{t}^{b \oplus c}},
\end{equation*}
for $b, c \in \{0, 1\}$, where $\lambda_{0}^{b \oplus c} + \lambda_{1}^{b \oplus c} = 1$. In particular, we have
\begin{equation*}
\Big( \bigotimes_{k = 1}^{n} \rho_{x_{k} \oplus y_{k}, x_{k} \oplus z_{k}} \Big) \Big( \bigotimes_{k = 1}^{n} \ket{e_{v_{k}}^{y_{k} \oplus z_{k}}} \Big) = \bigotimes_{k = 1}^{n} \lambda_{x_{k} \oplus y_{k} \oplus v_{k}}^{y_{k} \oplus z_{k}} \ket{e_{v_{k}}^{y_{k} \oplus z_{k}}}.
\end{equation*}
Therefore, we also know the eigenbasis of
\begin{equation*}
B_{yz}' = \sum_{\substack{x \bsr{n}\\ \wham(x) \leq \delta}} \bigotimes_{k = 1}^{n} \rho_{x_{k} \oplus y_{k}, x_{k} \oplus z_{k}},
\end{equation*}
and the largest eigenvalue equals
\begin{equation*}
\norm{B_{yz}'} = \max_{v \bsr{n}} \sum_{\substack{x \bsr{n}\\ \wham(x) \leq \delta}} \prod_{k = 1}^{n} \lambda_{x_{k} \oplus y_{k} \oplus v_{k}}^{y_{k} \oplus z_{k}}.
\end{equation*}
Recall from Eq.~\eqref{eq:cross-projector-norm} that the expression we want to bound is 
\begin{equation*}
\max_{y, z \bsr{n}} \norm{B_{yz}'} = \max_{v, y, z \bsr{n}} \sum_{\substack{x \bsr{n}\\ \wham(x) \leq \delta}} \prod_{k = 1}^{n} \lambda_{x_{k} \oplus y_{k} \oplus v_{k}}^{y_{k} \oplus z_{k}} = \max_{a, b \bsr{n}} \sum_{\substack{x \bsr{n}\\ \wham(x) \leq \delta}} \prod_{k = 1}^{n} \lambda_{x_{k} \oplus b_{k}}^{a_{k}} .
\end{equation*}
It is clear that every bit of $a$ and $b$ should be chosen to satisfy $\lambda_{b_{k}}^{a_{k}} = \max_{s, t} \lambda_{s}^{t} := \lambda_{0}$. Then
\begin{equation*}
\max_{y, z \bsr{n}} \norm{B_{yz}'} = \sum_{\substack{x \bsr{n}\\ \wham(x) \leq \delta}} \prod_{k = 1}^{n} \lambda_{x_{k}} = \sum_{k = 0}^{ \lfloor \delta n \rfloor } {n \choose k} \lambda_{0}^{n - k} \lambda_{1}^{k},
\end{equation*}
where $\lambda_{1} = 1 - \lambda_{0}$.
Finally, since we know that
\begin{equation*}
\norm{\ave{\Pi_{c}^{\theta}}} = \max_{y, z \bsr{n}} \norm{B_{yz}} \leq \max_{y, z \bsr{n}} \norm{B_{yz}'},
\end{equation*}
we obtain the security guarantee of the form
\begin{equation*}
\varepsilon = \sum_{k = 0}^{ \lfloor \delta n \rfloor } {n \choose k} \lambda_{0}^{n - k} \lambda_{1}^{k}.
\end{equation*}
For $\delta = 0$ there is only one term in the sum while for $0 < \delta < \lambda_{1}$ we use the Chernoff bound (Lemma \ref{lem:chernoff}) to obtain the final result of the lemma.
\end{proof}
The protocol is secure as long as $\delta < \lambda_{1}$, which combined with Eq.~\eqref{eq:tmo-correctness} implies that correctness and security is possible as long as
\begin{equation*}
\err < \lambda_{1}.
\end{equation*}
This allows us to check whether a particular experimental setup (characterised by $\err$ and $\lambda_{1}$) allows for a secure implementation of the protocol. For example, if the source emits perfect BB84 states (or, in fact, any two mutually unbiased bases on a qubit) we can tolerate up to $14.6\%$ of errors.
\section{Modelling imperfect devices}
\label{sec:modelling-devices}
While we have allowed our states to be imperfect and undergo some noise process in transit, we have implicitly assumed that every time \bob{0} pushes a button a qubit in a well-defined state is sent towards \alice{0}, who always detects it to obtain a particular classical outcome. As of today there is no physical system which matches this idealised description to a reasonable degree. Therefore, in collaboration with an experimental group at the University of Geneva, we have developed a new version of the protocol, which can be implemented using currently available devices.

First of all, instead of a single-photon source we use a weak-coherent source with phase randomisation\footnote{We have decided to use phase randomisation because then the number of photons in a pulse can be modelled as a classical random variable, which turns out to be convenient for the security analysis.}, which emits pulses of light in which the number of photons is a Poisson-distributed random variable. Let $\ket{r}$ be the Fock state of $r$ photons and $\mu$ be the average number of photons per pulse (an adjustable parameter of the source). Then, the ensemble emitted by a weak-coherent source can be written as
\begin{equation*}
\rho = \sum_{r = 0}^{\infty} p_{r} \ketbraq{r} \nbox{for} p_{r} = e^{- \mu} \cdot \frac{\mu^{r}}{r!}.
\end{equation*}
A direct consequence of this model is that some pulses might contain more than one photon and we refer to those as \emph{multiphoton emissions}. Such pulses constitute a deadly threat to our protocol since \alice{0} could measure the first photon in the computational basis, the second photon in the Hadamard basis and, hence, obtain enough information to open either value with certainty. Clearly, multiphoton emissions do not contribute to security and so their number must be rigorously controlled.

Besides the imperfections of the source, there is also a certain probability that a photon might be lost either in transit or during the detection process. Let $\eta$ be the \emph{detection efficiency}, i.e.~the probability that a photon sent by \bob{0} is detected by \alice{0}. We assume that the loss process affects every photon independently so the number of photons detected by \alice{0} is, again, a Poisson-distributed random variable. The probability of detecting $r$ photons equals
\begin{equation}
\label{eq:pr-definition}
p_{r}(\mu, \eta) = e^{-\mu \eta} \cdot \frac{(\mu \eta)^{r}}{r!}.
\end{equation}
We assume that the detection efficiency depends neither on the measurement setting nor on the incoming state. Note that while this is a natural assumption from the theoretical point of view, it does not always hold for an experimental setup (it is common to have slightly different detection efficiencies for measurements in different bases) and this issue needs to be addressed while analysing experimental data as explained in Section \ref{sec:tmo-experimental}. Moreover, it has recently been demonstrated that strong pulses of light might allow Bob to learn some information about Alice's basis setting \cite{lydersen10}, which is not included in our analysis.

We follow the standard approach (presented in Section \ref{sec:cryptographic-protocols}), i.e.~we assume that the devices used by the honest party are trusted (their characterisation including any imperfections is known) but the dishonest party is limited only by the laws of physics. The following table lists the models used for each of the three distinct scenarios.
\begin{center}
\begin{tabular}{c c c c c}
Alice & Bob & source & losses & errors\\
\hline
honest & honest & weak-coherent source & yes & yes\\
honest & dishonest & perfect & detectors only & N/A\\
dishonest & honest & weak-coherent source & no & no\\
\end{tabular}
\end{center}
\section{Protocol with backreporting}
If we try to implement the original protocol using the equipment described above, we run into a very simple problem: in most rounds \alice{0} simply does not see a click (either because the photon was never emitted or it was not detected). What is she supposed to do then? One solution would be to simply generate a random bit and act as if this was the outcome of the measurement. This solution works as long as losses are infrequent and can be ``hidden'' within the error threshold. However, in the experimental setup described above losses are extremely common. In fact, it is the detection events that are rare. Therefore, flipping a coin for every loss is not a feasible solution.

We solve this problem using a standard technique known as \emph{backreporting}, which requires \alice{0} to inform \bob{0} at the end of the quantum exchange which rounds were successful (i.e.~a photon was detected) and only these rounds are used for the protocol (all the remaining data is discarded). This clearly restores correctness but, unfortunately, it opens a new security loophole as dishonest \alice{0} might also backreport single-photon rounds in order to increase the contribution of multiphoton emissions, which she can win with certainty. To avoid this threat \bob{0} must carefully monitor the number of rounds backreported by Alice. Let $\cM$ be the set of rounds in which \alice{0} observed a click, which we call \emph{the valid set}. \bob{0} only continues with the protocol if the size of the valid set exceeds a certain threshold, $m := \abs{ \cM } \geq \gamma n$, where $\gamma \in [0, 1)$ is an adjustable parameter of the protocol called the \emph{detection threshold}.
\phantomsection
\backreportingtrue
\tmoprot
\subsection{Correctness}
To guarantee correctness we must first ensure that Alice registers a sufficient number of clicks. Asymptotically, we simply require that the probability of seeing a click (i.e.~detecting at least one photon) is larger than the detection threshold
\begin{equation*}
\sum_{r = 1}^{\infty} p_{r} = 1 - p_{0} > \gamma.
\end{equation*}
Using Eq.~\eqref{eq:pr-definition} to express $p_{0}$ in terms of $\mu$ and $\eta$ gives
\begin{equation*}
e^{-\mu \eta} + \gamma < 1.
\end{equation*}
Since in our model errors are independent of losses or multiphoton emissions, the second correctness condition remains the same as before, i.e.~Eq.~\eqref{eq:tmo-correctness}.
\subsection{Security for honest Alice}
\label{sec:honest-alice-backreporting}
Since backreporting introduces communication from \alice{0} to \bob{0} in the commit phase, security for honest Alice is no longer unconditionally true. To make sure that the valid set $\cM$ does not contain any information about the commitment we must ensure that the detection efficiencies do not depend on the basis choice regardless of the state that the dishonest \bob{0} sends in. We assume that the detection system used by Alice satisfies these properties (see Section \ref{sec:modelling-devices}).
\subsection{Security for honest Bob}
Security analysis for honest Bob is an extension built on top of the previous argument. We take advantage of the fact that all the experimental imperfections (e.g.~multiphoton emissions or no-detection events) can be modelled as classical random variables and that for particular values of these random variables Proposition \ref{prop:tmo-binding} provides an explicit security bound.

In every round a certain number of photons (between $0$ and $\infty$) is emitted (recall that in this case we assume that \alice{0} has perfect detectors, i.e.~$\eta = 0$). Pulses with no photons affect correctness but do not constitute a security threat. Pulses with one photon is what the original protocol calls for and what we analysed in the previous section. Finally, multiphoton pulses are a serious threat as they allow \alice{0} to obtain sufficient information to successfully open both values of the commitment. To simplify our analysis we replace all the zero-photon emissions by single-photon emissions (which only gives Alice more power). Eq.~\eqref{eq:pr-definition} with $\eta = 0$ implies that the probability of a multiphoton emission in a particular round equals
\begin{equation*}
p_{m} = 1 - e^{- \mu} ( 1 + \mu ).
\end{equation*}
The number of multiphoton rounds $N_{m}$ is a binomially-distributed random variable
\begin{equation}
\label{eq:multiphoton-rounds}
\Pr[N_{m} = k] = {n \choose k} p_{m}^{k} (1 - p_{m})^{n - k}.
\end{equation}
The optimal strategy of dishonest \alice{0} is to discard as many single-photon rounds as possible. It is clear that if she can discard all of them, she is left with multiphoton emissions only and no security can be guaranteed. Therefore, the necessary condition for security is that the number of multiphoton rounds is lower than the detection threshold
\begin{equation*}
N_{m} < \lceil \gamma n \rceil.
\end{equation*}
After using up the entire backreporting allowance the number of valid rounds equals $m = \lceil \gamma n \rceil$ but there are only $\lceil \gamma n \rceil - N_{m}$ (which is now guaranteed to be a positive number) single-photon rounds among them. Honest Bob believes that they are performing a bit commitment protocol of $\lceil \gamma n \rceil$ rounds but there is a certain number of multiphoton ones, which Alice can win ``for free''. Hence, she can concentrate her error allowance on the single-photon rounds and the security we achieve is that of playing a game of $\lceil \gamma n \rceil - N_{m}$ rounds with the absolute (non-fractional) error allowance of $\delta \lceil \gamma n \rceil$, which gives the effective (fractional) error allowance of
\begin{equation}
\label{eq:tmo-fractional-error}
\delta' = \frac{ \delta \lceil \gamma n \rceil }{ \lceil \gamma n \rceil - N_{m} }.
\end{equation}
The proof in Section~\ref{sec:tmo-honest-bob} is valid only if the effective (fractional) error allowance, $\delta'$, satisfies
\begin{equation*}
\delta' < \lambda_{1},
\end{equation*}
where $\lambda_{1}$ measures the incompatibility of the measurements performed by \bob{0}. Hence, in our case we require
\begin{equation}
\label{eq:security-bound}
\delta \lceil \gamma n \rceil < (\lceil \gamma n \rceil - N_{m}) \lambda_{1}.
\end{equation}
In the asymptotic limit it is sufficient to look at the expectation value
\begin{equation*}
\amsbb{E}[N_{m}] = p_{m} n = [1 - e^{-\mu} (1 + \mu)] n,
\end{equation*}
which substituted into Eq.~\eqref{eq:security-bound} gives
\begin{equation*}
e^{-\mu} (1 + \mu) + (1 - \delta/\lambda_{1} ) \gamma > 1.
\end{equation*}
\subsection{Requirements on the honest devices}
Having derived explicit criteria for correctness and security we can check whether a given experimental setup allows for a secure implementation of the protocol. The correctness and security constraints are
\begin{gather*}
e^{-\mu \eta} + \gamma < 1,\\
\err < \delta,\\
e^{-\mu} (1 + \mu) + (1 - \delta/\lambda_{1} ) \gamma > 1.
\end{gather*}
It is clear that $\delta$ and $\gamma$ (parameters of the protocol) can be taken arbitrarily close to the values which would turn the first two conditions into equalities. This leaves us with only one, but rather complicated, condition
\begin{equation*}
e^{-\mu} (1 + \mu) + (1 - \err/\lambda_{1} ) (1 - e^{-\mu \eta}) > 1.
\end{equation*}
This expression allows us to check whether for devices of certain quality (quantified by $\lambda_{1}$, $\err$ and $\eta$) there exists a value of $\mu$ that makes the protocol both correct and secure.\footnote{Note that changing the mean photon number $\mu$ affects the physical aspect of the protocol so it might influence other physical parameters like the error rate. On the other hand, parameters like $\delta$ or $\gamma$, which are only relevant for the post-processing, do not have such an effect.}
\subsection{Explicit security calculation}
The asymptotic analysis is relevant as $n \to \infty$ but in any practical scenario the number of rounds is finite. Therefore, we want to explicitly calculate security guarantees as a function of $n$. Since correctness is verified experimentally and security for honest Alice is perfect by assumption, we only need to calculate security for honest Bob.

Let $E$ denote the event that Alice successfully cheats in the bit commitment protocol. We have shown in Eq.~\eqref{eq:tmo-fractional-error} that the probability of cheating successfully depends on the number of multiphoton emissions. Therefore, let us write
\begin{equation}
\label{eq:summation}
\Pr[E] = \sum_{k = 0}^{n} \Pr[E | N_{m} = k] \Pr[N_{m} = k].
\end{equation}
Equation~\eqref{eq:tmo-security-bound} allows us to bound $\Pr[E | N_{m} = k]$ as long as the number of multiphoton emissions is below the threshold. For $k < k_{t} := \gamma n (1 - \delta / \lambda_{1})$ we have
\begin{equation*}
\Pr[E | N_{m} = k] \leq \exp \bigg(- \frac{1}{2} \Big( \sqrt{(\gamma n - k) \lambda_{1}} - \frac{\delta \gamma n}{\sqrt{(\gamma n - k) \lambda_{1}}} \Big)^{2} \bigg).
\end{equation*}
On the other hand, for $k \geq k_{t}$ there is no security and the trivial bound, $\Pr[E | N_{m} = k] \leq 1$, is the best we can hope for. The second term is given by Eq.~\eqref{eq:multiphoton-rounds}. Performing the summation \eqref{eq:summation} for any particular values of the parameters is a straightforward exercise in any package for numerical calculations (e.g.~Octave \cite{eaton09}).
\section{Experimental implementation}
\label{sec:tmo-experimental}
Protocol \hyperref[prot:tmo-backreporting]{8} requires the use of three equidistant locations on a line. This is clearly quite difficult to do if we want to take maximal advantage of the size of the Earth. Therefore, we have implemented a modified protocol, in which only two locations are used. The experiment is performed between Geneva (Location 1) and Singapore (Location 2) and achieves secure commitment for 15.6 ms (the maximal duration achievable on the Earth, corresponding to antipodal locations on the surface, equals 21.2 ms). As explained in Section \ref{sec:the-original-protocol} the maximal commitment time equals \emph{half} the time it takes to travel at lightspeed between the two opening locations (we cannot rule out the possibility that dishonest Alice deployed an extra agent exactly in between the two locations, who receives the quantum states at $t = 1$ and only then performs the measurements, cf.~Fig.~\ref{fig:tmo-experimental}).
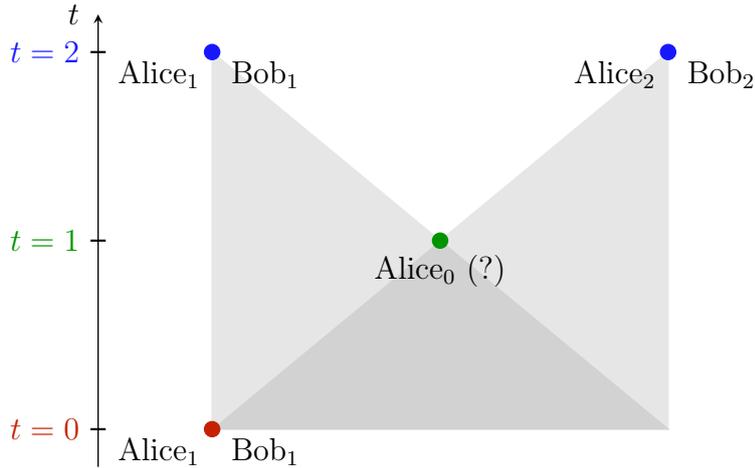
\begin{figure}
\centering
\begin{tikzpicture}[scale=1, line width=0.5]
\timeaxis{5.5}
\draw [-, thick] (-0.1, 5) to (0.1, 5);
\node[left, color=mycolour1] at (-0.1, 5) {$t = 2$};
\draw [fill, lightgray] (1.5, 0) -- (1.5, 5) -- (4.5, 2.5) -- (7.5, 5) -- (7.5, 0);
\draw [fill, darkgray] (4.5, 2.5) -- (7.5, 0) -- (1.5, 0);
\node at (0.8, -0.3) {\alice{1}};
\node at (2.2, -0.3) {\bob{1}};
\draw[fill, color=mycolour3] (1.5, 0) circle [radius=0.1];
\draw [-, thick] (-0.1, 2.5) to (0.1, 2.5);
\node[left, color=mycolour2] at (-0.1, 2.5) {$t = 1$};
\draw[fill, color=mycolour2] (4.5, 2.5) circle [radius=0.1];
\node at (4.5, 2.1) {\alice{0} (?)};
\node at (0.8, 4.7) {\alice{1}};
\node at (2.2, 4.7) {\bob{1}};
\node at (6.8, 4.7) {\alice{2}};
\node at (8.2, 4.7) {\bob{2}};
\draw[fill, color=mycolour1] (1.5, 5) circle [radius=0.1];
\draw[fill, color=mycolour1] (7.5, 5) circle [radius=0.1];
\end{tikzpicture}
\caption{Spacetime diagram for the experimental implementation of Protocol \hyperref[prot:tmo-backreporting]{8}. The red dot represents the commit phase while the blue dots represent the open phase. The shaded area corresponds to the past light cones of the events of the open phase. The green dot determines the commitment point, i.e.~the \emph{latest} point at which Alice can still perform an honest commitment.}
\label{fig:tmo-experimental}
\end{figure}

Each location hosts one agent for each player synchronised to universal time using a global positioning system (GPS) clocks. The protocol consists of two parts: (i) exchange of quantum information in the commit phase and (ii) exchange of classical information in the open phase. Phase (i) was implemented using a commercial quantum key distribution system Vectis 5100 from ID Quantique located at the University of Geneva (see Fig.~\ref{fig:experiment1}). This system is based on the two-way ``Plug\&Play'' configuration \cite{muller97}: strong optical pulses travel from \alice{1} to \bob{1}, who uses them to encode the BB84 states. Moreover, he attenuates the optical power down to single photon level and sends these weak-coherent pulses back to \alice{1}. Trojan-horse attacks on Bob's side are particularly effective against the ``Plug\&Play'' configuration so the power of the incoming beam is continuously monitored by \bob{1}. To use the quantum key distribution system for bit commitment some software changes must be made, in particular communication from \alice{1} to \bob{1} is restricted to backreporting only. As discussed in Section \ref{sec:honest-alice-backreporting} extra care has to be taken to ensure that the backreported data does not leak any information about the choice of measurements performed by \alice{1}. In our experimental setup the two bases exhibit slightly different detection probabilities so the raw data must be artificially ``equalised''.\footnote{More specifically, with probability $p$ we discard an outcome obtained for the more efficient setting, where $p$ is chosen to make the expected number of detections backreported for both settings equal.}

The duration of the commitment is ultimately limited by the distance between the locations under the assumption that all local communication is instantaneous. However, since the exchange of quantum states in the commit phase takes a considerate amount of time, which would reduce the achievable commitment time, we have performed a \emph{delayed commitment}. In a delayed commitment scheme \alice{1} first commits to a random bit $r$ and only later announces $\cval \oplus r$, which initiates the actual commitment (and determines its value). This allows us to perform all the quantum information exchange in advance, which maximises the commitment time.

Classical information needed by \alice{2} in Singapore to open the commitment was transmitted in advance through the internet (one-time-pad encrypted using pre-shared keys generated by a quantum random generator from ID Quantique). The classical information exchange in the open phase was performed using stand-alone computers equipped with field-programmable gate array (FPGA) to make the transmission time negligible relative to the commitment time (around 3 $\mu$s).
\begin{figure}
\centering
\includegraphics[scale=0.55]{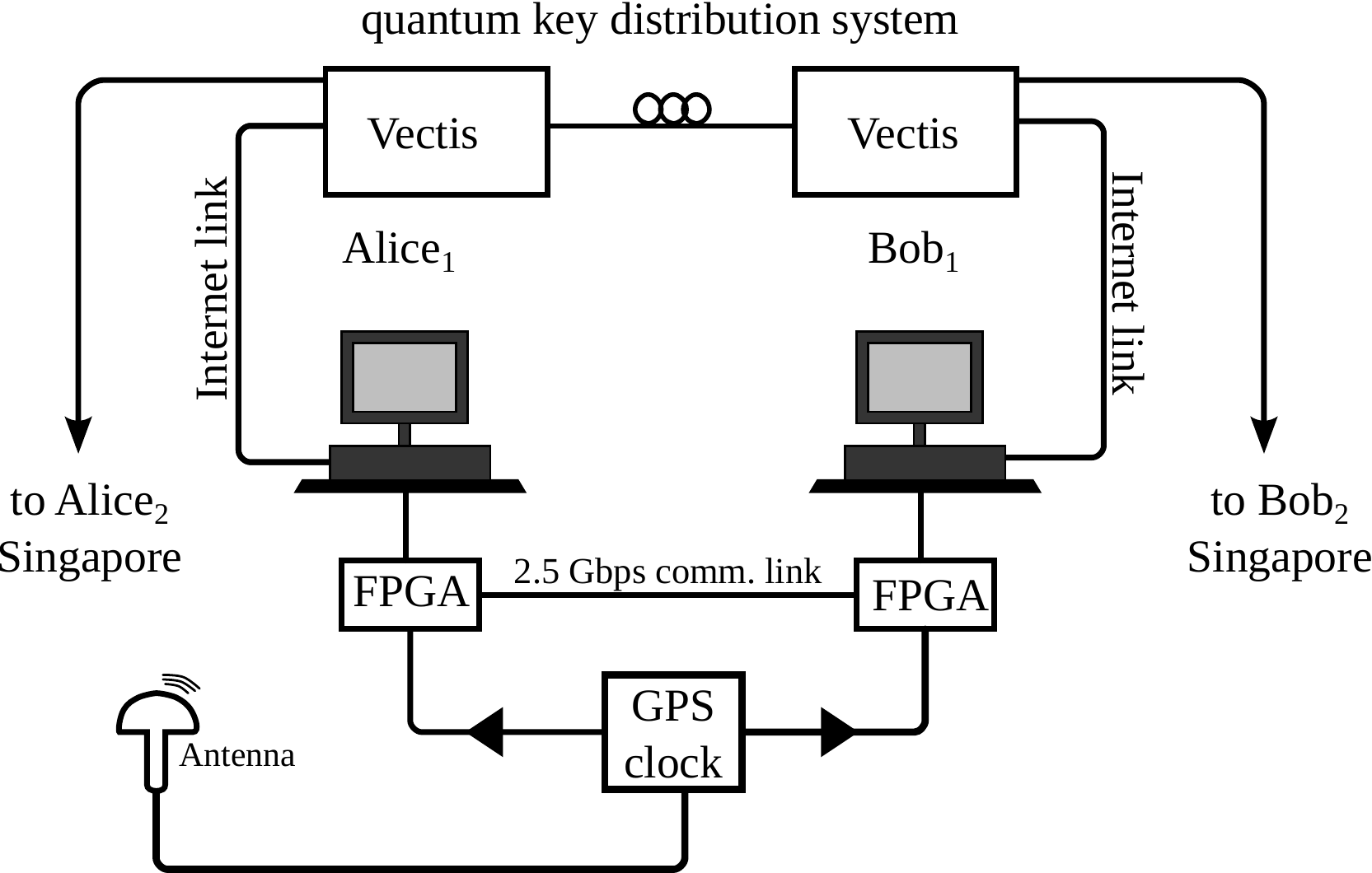}
\caption{Experimental setup located in Geneva. The setup located in Singapore is identical except that \alice{1} and \bob{1} are replaced by \alice{2} and \bob{2}, and that there is no quantum key distribution system (the protocol only requires quantum interactions in the commit phase which as explained in Section~\ref{sec:tmo-experimental} happens in Geneva). Figure reproduced with permission from Tommaso Lunghi.}
\label{fig:experiment1}
\end{figure}
\chapter{Multiround relativistic bit commitment protocol}
\label{chap:multiround}
\emph{This chapter is based on}
\paperC
The quantum protocol discussed in the previous chapter implements secure bit commitment in the $\beta$-split model (one of the minimal splits, in which, as discussed in Section~\ref{sec:commitment-schemes}, no classical protocol can give security). This demonstrates that quantum protocols are strictly more powerful and in this particular case the extra power results from the no-cloning theorem. However, a weak point of that scheme is the length of the commitment, which is limited by the spatial separation of the two opening sites. In particular, for a protocol taking place on the Earth, the commitment time is limited to 21 ms. While this might be sufficient for some purposes, extending the commitment time would be highly desirable. It is clear that if the spatial arrangement is fixed, the only manner to extend the commitment time is to introduce additional rounds of communication. In fact, this idea was proposed by Kent quite early on \cite{kent99, kent05}, where instead of \emph{opening} the commitment Alice \emph{commits} to the information she \emph{would have used} in the unveiling. The original way of ``chaining'' commitments has the drawback that the communication required grows exponentially with the number of rounds \cite{kent99}. This was later rectified by adding a compression scheme on top of the protocol \cite{kent05}. Security of such a chained scheme would follow directly if the individual commitments were composably secure\footnote{In fact, this is exactly the idea of composable security: the protocol is indistinguishable from the ideal primitive in \emph{all possible scenarios}.}. However, we do not know how to prove composable security of relativistic commitment schemes and there is evidence that this might indeed be impossible (e.g.~in Appendix~\ref{app:classical-certification} we show that the classical relativistic \texttt{sBGKW} scheme is not secure according to the usual composable definition). Therefore, the security proof must explicitly consider all the intermediate commitments. Security argument against classical adversaries presented by Kent \cite{kent05} is of asymptotic nature and, therefore, not sufficient for implementation purposes.

\noindent \textbf{Outline:} With the goal of finding a new classical multiround protocol in mind we first revisit the two-round protocol proposed in Ref.~\cite{simard07, crepeau11} and show that security for honest Bob relies on the difficulty of a certain non-local game, whose quantum value was recently investigated by Sikora, Chailloux and Kerenidis \cite{sikora14}. This completes the security analysis and shows that the two-round protocol is secure against quantum adversaries.\footnote{Security for honest Alice is obvious against classical Bob. Security against quantum Bob follows from a simple observation that if only one agent (\bob{1}) is involved in the commit phase of a classical protocol, no advantage can be gained by using quantum systems (cf.~footnote \ref{footnote} in Section \ref{sec:quantum-players}).} With the two-round scheme as our starting point we propose a new classical multiround protocol (with constant communication rate) and analyse its security against \emph{classical} adversaries.\footnote{For more than two rounds security analysis against quantum adversaries becomes rather involved as explained in Section~\ref{sec:quantum-players}.} Correctness can be verified by inspection, security for honest Alice is intuitively obvious (nevertheless formalising it requires some work) and security for honest Bob turns out to be more involved but we eventually derive an explicit and easily computable security bound. Unfortunately, the bound suffers from rather undesirable scaling in the number of rounds (which is proportional to the length of the commitment). While the commitment time is not subject to any fundamental limitations, an implementation using standard digital equipment only allows modest commitment times. In collaboration with the Geneva group we have implemented the protocol to achieve a secure commitment of 2 ms.\footnote{An attentive reader might be puzzled that the new multiround protocol yields commitment time which is significantly shorter than the one previously achieved by the quantum protocol. The solution of this conundrum is that while the previous experiment was performed between Geneva and Singapore, the new multiround one was (for practical reasons) executed between two locations within Switzerland. Performing the multiround protocol between Geneva and Singapore would give a commitment time of 156 ms (a tenfold improvement over the quantum protocol).}
\vspace{12pt}

\noindent \textbf{Note added:} After the completion of this work, two groups independently provided security proofs that give significantly improved security bounds \cite{fehr15a, chakroborty15}. These results imply that Protocol \hyperref[prot:multiround]{9} can be used to realise long-lasting commitments with very modest resources, hence, making it truly practical.
\section{Two-round protocol}
\label{sec:two-round-protocol}
Since the two-round protocol has already been discussed in Section \ref{sec:simple-relativistic-protocols} let us go directly to the security analysis. Recall that the protocol requires two sites labelled Location 1 and Location 2. To simplify notation in this chapter we assume that these locations are separated by $(1 + \epsilon)$ units of length for some $\epsilon > 0$. Therefore, one unit of time is \emph{not} sufficient to transmit information between them.
\sBGKWrel{t = 1}{= 1 + \epsilon}
\noindent (Note that the verification step could be performed slightly earlier using an agent positioned somewhere in between \bob{1} and \bob{2} but since $\varepsilon$ is chosen to be essentially $0$ this makes no difference.)

Correctness is straightforward to check and security for honest Alice comes from the fact that the message that \bob{1} receives in the commit phase is one-time-padded with a uniformly random string.

Security for honest Bob is quantified using Definition \ref{df:binding-two}. Since only one agent (\alice{2}) is involved in the open phase, the distinction between local and global command does not arise. As explained in Section \ref{sec:formal-definition} it is often helpful to explicitly state the cheating game that Alice attempts to win. In the commit phase \alice{1} receives $b \in \{0, 1\}^{n}$ (chosen uniformly at random) and in the open phase \alice{2} receives nothing from Bob but she is challenged to open $\cval \in \{0, 1\}$ chosen uniformly at random. In each phase she is required to output an $n$-bit string, which we denote by $x_{1}$ and $x_{2}$, respectively. Since there is no communication between \alice{1} and \alice{2} this is equivalent to a two-player game which (using the formalism presented in Section \ref{sec:multiplayer-games}) is specified by the uniform input distribution $p(b, \cval) = \frac{1}{2} \cdot \frac{1}{2^{n}}$ and the predicate function
\begin{equation*}
V( b, x_{1}, \cval, x_{2} ) =
\begin{cases}
1 &\nbox{if} x_{1} \oplus x_{2} = \cval \cdot b,\\
0 &\nbox{otherwise.}
\end{cases}
\end{equation*}
Since this game can be seen as a generalisation of the CHSH game (which corresponds to $n = 1$), let us call it $\chshn$ after Ref.~\cite{sikora14}. Calculating the classical value is straightforward but let us do it explicitly for completeness.

As explained in Section \ref{sec:multiplayer-games} we might without loss of generality assume that \alice{1} and \alice{2} employ deterministic strategies, which we denote by functions $f_{1}(b)$ and $f_{2}(\cval)$. Define $H_{\cval}$ as the event over $B$ (private randomness of Bob) that the opening of $\cval$ is accepted
\begin{equation}
\label{eq:two-round-event}
H_{\cval} \iff \cval \cdot B = f_{1}(B) \oplus f_{2}(\cval).
\end{equation}
Since both $H_{0}$ and $H_{1}$ are defined over $B$ it is meaningful to talk about $H_{0} \lor H_{1}$ and $H_{0} \wedge H_{1}$.\footnote{This is exactly where the argument breaks down when applied to the quantum world, in which $H_{0}$ and $H_{1}$ are not in general defined \emph{simultaneously}.} Note that $\Pr[H_{0}] + \Pr[H_{1}] = \Pr[H_{0} \lor H_{1}] + \Pr[H_{0} \wedge H_{1}] \leq 1 + \Pr[H_{0} \wedge H_{1}]$. The event $H_{0} \wedge H_{1}$ happens when condition \eqref{eq:two-round-event} is satisfied for both values of $d$. Define $K$ to be the event that the \texttt{XOR} of the two is satisfied
\begin{equation*}
K \iff B = f_{2}(0) \oplus f_{2}(1).
\end{equation*}
Since the left-hand side is a uniformly distributed random variable and the right-hand side is a constant $\Pr[K] = 2^{-n}$. Moreover, as $H_{0} \wedge H_{1} \implies K$ we have $\Pr[H_{0} \wedge H_{1}] \leq \Pr[K]$. Combining the two statements implies that for any strategy of classical Alice we have $\Pr[H_{0}] + \Pr[H_{1}] \leq 1 + 2^{-n}$ which leads to
\begin{equation*}
\omega( \chshn ) \leq \frac{1}{2} + \frac{1}{ 2^{n + 1} }.
\end{equation*}
It is easy to check that the trivial strategy of always outputting $x_{1} = x_{2} = 0^{n}$ saturates this bound. Intuitively this game should be difficult because unveiling $\cval = 0$ and $\cval = 1$ requires \alice{2} to know $x_{1}$ and $x_{1} \oplus b$, respectively. She cannot guess both of these too well since $b$ is chosen uniformly at random by \bob{1}.

Applying this reasoning to quantum adversaries is not so straightforward because \alice{2} might have two distinct measurements that reveal $x_{1}$ and $x_{1} \oplus b$, respectively, but they might be incompatible so the implications on her ability to guess $b$ are not so obvious. Fortunately, the following bound on the quantum value was recently proven \cite{sikora14}
\begin{equation*}
\omega^{*}( \chshn ) \leq \frac{1}{2} + \frac{1}{\sqrt{2^{n + 1}}}.
\end{equation*}
This is sufficient for our purposes as it implies that
\begin{equation*}
p_{0} + p_{1} \leq 1 + \sqrt{2} \cdot 2^{- n / 2}
\end{equation*}
for all strategies of dishonest quantum Alice. Therefore, the protocol is $\varepsilon$-binding with $\varepsilon = 2^{(1 - n) / 2}$ decaying exponentially in $n$ (but note that the decay rate is half of the decay rate against classical adversaries).

This is a prime example that analysing classical protocols against quantum adversaries is \emph{not always} a futile task and sometimes leads to interesting observations. An immediate question arises regarding super-quantum adversaries: what if \alice{1} and \alice{2} have access to stronger-than-quantum correlations? It is straightforward to see that the protocol is completely insecure against \alice{1} and \alice{2} who have access to no-signalling correlations \cite{simard07, crepeau11}. In fact, it was shown recently that in this particular model it is not possible to have a protocol secure against no-signalling adversaries \cite{fehr15}. This is a consequence of the fact that in this context the hiding property coincides with the definition of no-signalling and so if the protocol is hiding then there exists a perfect cheating strategy consistent with no-signalling. This is yet another example of how classical and quantum theories are qualitatively different from the no-signalling world.
\section{Multiround protocol}
To extend the commitment time we must introduce additional rounds of communication, which keep the commitment ``alive'' (the sustain phase).
If Alice and Bob require the commitment to be valid for $m$ units of time (for some $m \in \amsbb{N}$) they need to execute the following protocol of $m + 1$ rounds. We use $k$ as a label for the round under consideration and since the rounds alternate between the two locations we define
\begin{equation*}
l(k) =
\begin{cases}
1 \nbox{if} k \equiv 1 \pmod 2,\\
2 \nbox{if} k \equiv 0 \pmod 2.
\end{cases}
\end{equation*}
We use $a_{k}$ and $b_{k}$ to denote private strings (chosen uniformly at random) of Alice and Bob, respectively and $x_{k}$ and $y_{k}$ to denote messages announced by Alice and Bob, respectively, in the $k\th$ round of the protocol. All the $n$-bit strings are interpreted as elements of the finite field $\amsbb{F}_{2^{n}}$ and ``$*$'' denotes the finite field multiplication.
\phantomsection
\begin{prot}{9}{Multiround bit commitment}
\label{prot:multiround}
\begin{enumerate}
\item (commit, $k = 1$) At $t = 0$, \bob{1} sends $y_{1} = b_{1}$ to \alice{1}. \alice{1} returns $x_{1} = d \cdot y_{1} \oplus a_{1}$.
\item (sustain, $2 \leq k \leq m$) At $t = k - 1$, \bob{l(k)} sends $y_{k} = b_{k}$ to \alice{l(k)}. \alice{l(k)} returns $x_{k} = (y_{k} * a_{k - 1}) \oplus a_{k}$.
\item (open, $k = m + 1$) At $t = m$, \alice{l(m + 1)} sends $d$ and $x_{m + 1} = a_{m}$ to \bob{l(m + 1)}.
\item (verify) At $t = m + \epsilon$, \bob{l(m + 1)} receives $x_{m}$ and accepts the opening if
\begin{equation}
\label{eq:acceptance-condition}
\begin{aligned}
x_{m + 1} \; = \; x_{m} \; &\oplus \; b_{m} * x_{m - 1} \; \oplus \; b_{m} * b_{m - 1} * x_{m - 2} \; \oplus \; \ldots\\
\ldots \; &\oplus \; b_{m} * b_{m - 1} * \ldots * b_{2} * x_{1} \; \oplus \; d \cdot b_{m} * b_{m - 1} * \ldots * b_{1}.
\end{aligned}
\end{equation}
\end{enumerate}
\end{prot}
\noindent It is easy to see that the timing is chosen precisely to make every two consecutive rounds space-like separated. In the formalism presented in Section \ref{sec:explicit-analysis} there are $m + 1$ interactions and the communication graph is
\begin{equation}
\label{eq:communication-graph}
\begin{gathered}
G = ( [m + 1], E) \nbox{for}\\
E = \{ (j, k) \in [m + 1]^{2}: j + 2 \leq k \}.
\end{gathered}
\end{equation}
Correctness of the protocol is straightforward to check by substituting the honest responses of \alice{1} and \alice{2} into the acceptance condition \eqref{eq:acceptance-condition}.
\subsection{Security for honest Alice}
Security for honest Alice is a direct consequence of the fact that every message she announces is one-time-padded with a fresh secret $n$-bit string. Hence, we would intuitively expect the transcripts corresponding to $\cval = 0$ and $\cval = 1$ to be statistically indistinguishable. We prove this statement by considering an arbitrary adaptive attack (consistent with relativity) that classical \bob{1} and \bob{2} might implement. While we do not believe that \bob{1} and \bob{2} can gain anything by using quantum systems, we currently do not have a rigorous argument to justify this belief.

We start with a lemma which formalises the intuition that if we take an arbitrary random variable taking values in $\amsbb{F}_{q}$ and add it to a uniform and uncorrelated random variable (over $\amsbb{F}_{q}$) then there will be no correlations between the input and the output (or any function thereof). More specifically, in the following lemma $Y$ is a random variable from which the input is generated using function $g$, $X$ is the fresh (finite field) randomness and $h$ is a function allowing us to condition on a certain subset of values of $Y$.
\begin{lem}
\label{lem:random-variables}
Let $\cX = \amsbb{F}_{q}$ and $\cY, \cZ$ be arbitrary finite sets. Let $X$ and $Y$ be two random variables taking values in $\cX$ and $\cY$, respectively, such that $X$ is uniform and independent from $Y$
\begin{equation}
\label{eq:assumption}
\Pr[X = x, Y = y] = q^{-1} \cdot \Pr[Y = y],
\end{equation}
for all $x \in \cX$ and $y \in \cY$. Then for arbitrary functions $g : \cY \to \cX$, $h : \cY \to \cZ$ and arbitrary fixed $x \in \cX$, $z \in \cZ$ it holds that
\begin{equation*}
\Pr[X + g(Y) = x \, | \, h(Y) = z] = q^{-1}.
\end{equation*}
\end{lem}
\begin{proof}
Note that
\begin{align*}
&\Pr[X + g(Y) = x, h(Y) = z] = \sum_{y \in \cY} \Pr[X = x - g(y), h(y) = z , Y = y]\\
&= \sum_{y \in \cY \atop h(y) = z} \Pr[X = x - g(y), Y = y] = \sum_{y \in \cY \atop h(y) = z} q^{-1} \cdot \Pr[Y = y] = q^{-1} \cdot \Pr[h(Y) = z],
\end{align*}
where the second last equality follows from applying the assumption~\eqref{eq:assumption} to every term of the sum.
\end{proof}
\noindent Lemma \ref{lem:random-variables} allows us to prove security for honest Alice.
\begin{prop}
If Alice is honest then the protocol is hiding.
\end{prop}
\begin{proof}
We want to show that the transcripts for $\cval = 0$ and $\cval = 1$ at the opening point (i.e.~after the last sustain round) are indistinguishable, i.e.
\begin{equation*}
\Pr[X_{1} = x_{1}, X_{2} = x_{2}, \ldots, X_{m} = x_{m} | \cval = 0] = \Pr[X_{1} = x_{1}, X_{2} = x_{2}, \ldots, X_{m} = x_{m} | \cval = 1]
\end{equation*}
for all $x_{1}, x_{2}, \ldots, x_{m}$. In fact, we show a stronger statement, namely
\begin{equation}
\label{eq:transcript-distribution}
\Pr[X_{1} = x_{1}, X_{2} = x_{2}, \ldots, X_{t} = x_{t} | d = c] = 2^{-n t},
\end{equation}
for all $t \in [m]$ and both values of $c \in \{0, 1\}$.

Honest Alice follows the protocol, which means that $\{A_{k}\}_{k = 1}^{m}$ are drawn independently, uniformly at random from $\{0, 1\}^{n}$ and so Alice's message in the $k\th$ round (represented as a random variable) equals
\begin{equation}
\label{eq:yk}
X_{k} =
\begin{cases}
d \cdot Y_{1} \oplus A_{1} &\nbox{for} k = 1,\\
(Y_{k} * A_{k - 1}) \oplus A_{k} &\nbox{for} 2 \leq k \leq m.
\end{cases}
\end{equation}
\bob{1} and \bob{2}, on the other hand, are only limited by the causal constraints, which means that the message in the $k\th$ round might depend on some pre-shared randomness denoted by $R_{B}$ and all the responses of \alice{1} and \alice{2} which belong to the past of the $k\th$ round. Therefore, without loss of generality the message in the $k\th$ round is
\begin{equation}
\label{eq:xk}
Y_{k} = f_{k}(R_{B}, X_{1}, X_{2}, \ldots X_{k - 2})
\end{equation}
for some arbitrary function $f_{k}$ (we include all randomness used by Bob in $R_{B}$ so $f_{k}$ is deterministic).

In this scenario the full transcript is a deterministic function of Alice's commitment $d$, her private randomness $\{A_{k}\}_{k = 1}^{m}$ and Bob's pre-shared randomness $R_{B}$. For every string announced by Alice and Bob we can explicitly find the subset of random variables it may depend on as listed in the table below
\begin{center}
\setlength{\tabcolsep}{0.6cm}
\begin{tabular}{c c}
message & random variables it might depend on\\
\hline
$Y_{1}$ & $R_{B}$\\
$Y_{2}$ & $R_{B}$\\
$Y_{3}$ & $d, R_{B}, A_{1}$\\
\vdots & \vdots\\
$Y_{k}$ & $d, R_{B}, A_{1}, A_{2}, \ldots, A_{k - 2}$\\
\vdots & \vdots\\
$Y_{m}$ & $d, R_{B}, A_{1}, A_{2}, \ldots, A_{m - 2}$\\
$X_{1}$ & $d, R_{B}, A_{1}$\\
$X_{2}$ & $d, R_{B}, A_{1}, A_{2}$\\
$X_{3}$ & $d, R_{B}, A_{1}, A_{2}, A_{3}$\\
\vdots & \vdots\\
$X_{k}$ & $d, R_{B}, A_{1}, A_{2}, \ldots, A_{k}$\\
\vdots & \vdots\\
$X_{m}$ & $d, R_{B}, A_{1}, A_{2}, \ldots, A_{m}$\\
\end{tabular}
\end{center}
First, we verify that condition \eqref{eq:transcript-distribution} holds for $t = 1$
\begin{equation*}
\Pr[X_{1} = x_{1} | d = c] = \Pr[ b \cdot Y_{1} \oplus A_{1} = x_{1}] = \Pr[ b \cdot f_{1}(R_{B}) \oplus A_{1} = x_{1}] = 2^{-n},
\end{equation*}
where the first two equalities follow from Eqs.~\eqref{eq:yk} and \eqref{eq:xk}, respectively. The last equality is a direct consequence of Lemma~\ref{lem:random-variables} (in a simplified form: no conditioning) applied to $X = A_{1}$, $Y = (c, R_{B})$, $g(Y) = c \cdot f_{1}(R_{B})$. Now, suppose that Eq.~\eqref{eq:transcript-distribution} holds for $t = k$. Then
\begin{align*}
\Pr&[X_{1} = x_{1}, \ldots, X_{k + 1} = x_{k + 1} | d = c]\\
&= \Pr[X_{k + 1} = x_{k + 1} | d = c, X_{1} = x_{1}, \ldots, X_{k} = x_{k}] \cdot \Pr[X_{1} = x_{1}, \ldots, X_{k} = x_{k} | d = c]\\
&= \Pr[(Y_{k + 1} * A_{k}) \oplus A_{k + 1} = x_{k + 1} | d = b, X_{1} = x_{1}, \ldots, X_{k} = x_{k}] \cdot 2^{-n k}\\
&= 2^{-n} \cdot 2^{-n k} = 2^{- (n + 1) k},
\end{align*}
where the second last inequality follows from applying Lemma~\ref{lem:random-variables} to
\begin{align}
X &= A_{k + 1},\nonumber\\
Y &= (c, R_{B}, A_{1}, \ldots, A_{k}),\nonumber\\
\label{eq:gy}
g(Y) &= Y_{k + 1} * A_{k},\\
h(Y) &= (X_{1}, X_{2}, \ldots, X_{k})\nonumber.
\end{align}
Note that it is not immediately obvious and the reader should verify (using the table presented above) that the quantities on the right-hand side of Eq.~\eqref{eq:gy} are functions of $Y$ alone, and therefore satisfy the assumptions of the lemma. This shows that Eq.~\eqref{eq:transcript-distribution} holds for $t = k + 1$ and so by induction it must hold for all $t \in [m]$. This shows that even at the opening point the transcript contains no information about Alice's commitment, which implies that the protocol is hiding.
\end{proof}
\subsection{Security for honest Bob}
Security for honest Bob is where the framework developed in Section \ref{sec:explicit-analysis} comes in useful. We immediately identify the case of honest Bob as a game of $m + 1$ players $\cP_{1}, \ldots, \cP_{m + 1}$ whose communication is restricted by $G$ defined in Eq.~\eqref{eq:communication-graph}. Player $\cP_{k}$ (for $1 \leq k \leq m$) receives a uniformly random $n$-bit string represented by the random variable $B_{k}$. Moreover, all the players except for $\cP_{1}$ receive $d$ (the value they are challenged to unveil) chosen uniformly at random. It is clear that $d$ must not be available to $\cP_{1}$ (this would correspond to an honest commitment) but it must be available to all the other players. This is important as it fixes the commitment point to occur immediately after $t = 0$, i.e.~allows us to claim that Alice becomes committed immediately after the first round.

Having explicitly determined the inputs received by each player and the communication constraints we apply Observation \ref{obs:games-equivalence} to turn this scenario into a non-communicating game. It is easy to verify that the output of $\cP_{1}$ denoted by $X_{1}$ can be written as
\begin{equation*}
X_{1} = f_{1}(R_{A}, B_{1}),
\end{equation*}
where $R_{A}$ corresponds to any randomness shared by the players. For $\cP_{k}$ for $2 \leq k \leq m$ we have
\begin{equation*}
X_{k} = f_{k}(R_{A}, B_{1}, B_{2}, \ldots, B_{k - 2}, B_{k}, d).
\end{equation*}
Finally, for $\cP_{m + 1}$ we have
\begin{equation*}
X_{m + 1} = f_{m + 1}(R_{A}, B_{1}, B_{2}, \ldots, B_{m - 1}, d).
\end{equation*}
This allows us to prove security for honest Bob.
\begin{prop}
If Bob is honest then the protocol is $\varepsilon$-binding for $\varepsilon = \omega_{m}$ defined in Eq.~\eqref{eq:omega-definition}.
\end{prop}
\begin{proof}
The argument is essentially identical to the one presented in Section \ref{sec:two-round-protocol} and we use the same definitions for the events $H_{0}, H_{1}$ and $K$. Since $K$ is defined as the \texttt{XOR} of Eq.~\eqref{eq:acceptance-condition} for $d = 0$ and $d = 1$ we have
\begin{align*}
K \iff B_{1} * B_{2} * \ldots * B_{m} &= g_{m + 1}(R_{A}, B_{1}, B_{2}, \ldots, B_{m - 1}) \oplus g_{m}(R_{A}, B_{1}, B_{2}, \ldots, B_{m - 2}, B_{m})\\
&\bigoplus_{k = 2}^{m - 1} B_{m} * B_{m - 1} * \ldots * B_{k + 1} * g_{k}(R_{A}, B_{1}, B_{2}, \ldots, B_{k - 2}, B_{k}),
\end{align*}
where
\begin{align*}
g_{k}(R_{A}, B_{1}, B_{2}, \ldots, B_{k - 2}, B_{k}) &= f_{k}(R_{A}, B_{1}, B_{2}, \ldots, B_{k - 2}, B_{k}, d = 0)\\
&\oplus f_{k}(R_{A}, B_{1}, B_{2}, \ldots, B_{k - 2}, B_{k}, d = 1).
\end{align*}
To bound $\Pr[K]$ note that the right-hand side contains exactly $m$ terms, but each of them depends on $(m - 1)$ $B$'s; none of the terms depends on \emph{all $B$'s simultaneously}. The terms corresponding to $2 \leq k \leq m - 1$ have some internal structure (e.g.~the dependence on $B_{m}$ is \emph{not} arbitrary) but we can relax the problem to the case where the $k\th$ term is an arbitrary function of all the $B$'s except for $B_{k}$ denoted by $h_{k}$. The winning condition for the relaxed game is
\begin{equation*}
B_{1} * B_{2} * \ldots * B_{m} = \bigoplus_{k = 1}^{m} h_{k}( B_{ [m] \setminus \{k\} } ).
\end{equation*}
In Section~\ref{sec:game-definition} we define the optimal winning probability for this game to be $\omega_{m}$, which concludes the proof since
\begin{equation*}
p_{0} + p_{1} \leq 1 + \Pr[K] \leq 1 + \omega_{m}.\qedhere
\end{equation*}
\end{proof}
The recursive argument presented in Section~\ref{sec:recursive-upper-bound} allows us to obtain explicit upper bounds on $\omega_{m}$. In particular, we have $\omega_{m} \leq c_{m}$, where
\begin{equation}
\label{eq:bound-multi}
c_{m} =
\begin{cases}
2^{-n} &\nbox{for} m = 1,\\
\frac{1}{2^{n + 1}} + \sqrt{c_{m - 1}} &\nbox{for} m \geq 2.
\end{cases}
\end{equation}
For large $n$, to a good approximation we have $c_{m} \approx 2^{-n / 2^{m}}$. The decay is exponential in $n$ but since the decay rate strongly depends on $m$, security deteriorates rapidly as we increase the number of rounds. The tightness of these bounds is an interesting open problem and is briefly discussed in Appendix B.5 of Ref.~\cite{lunghi15}. No explicit cheating strategy is known, whose winning probability would approach our security bounds.
\section{Experimental implementation}
Both the two-round and the multiround protocols have been implemented between University of Geneva and University of Berne. The straight-line distance between these two locations is $s = 131$ km, which corresponds to $\Delta t = 437 \mu$s. Each classical agent consists of a standalone computer equipped with a FPGA and the agents are connected by an optical link (see Fig.~\ref{fig:experiment2}). Synchronisation to universal time is achieved via a GPS clock. While the task of exchanging classical information is on its own quite straightforward, the challenge in our case is to take maximal advantage of the relativistic constraints. To maximise the commitment time, it is crucial to ensure that the devices are synchronised up to high accuracy and that classical data manipulation (communication with an external memory to load and store data, data exchange between the two agents and local computation) are optimised to produce the highest feasible rate.
\begin{figure}
\centering
\includegraphics[scale=0.4]{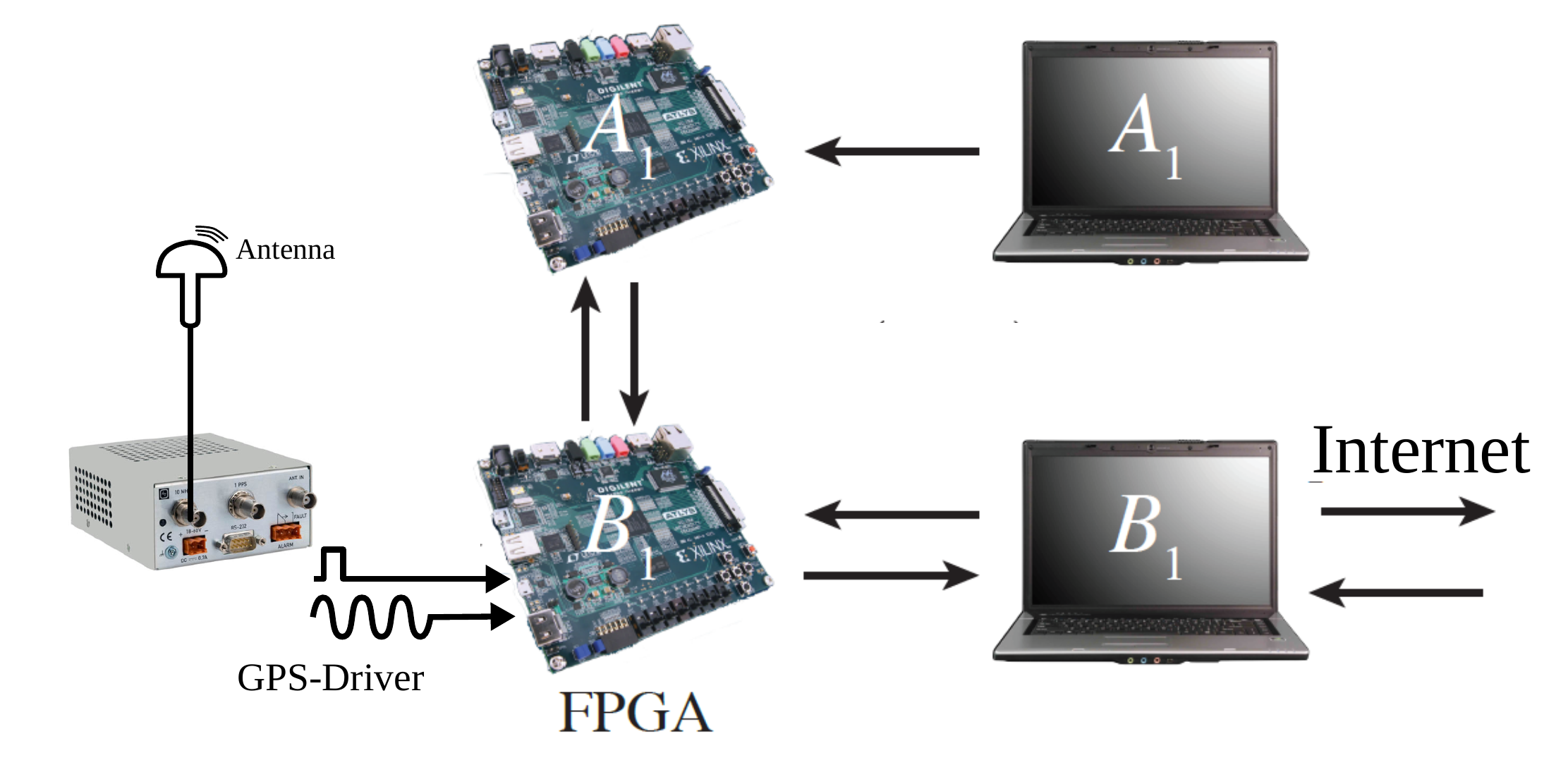}
\caption{Experimental setup at Location 1 where $A_{1}$ and $B_{1}$ represent \alice{1} and \bob{1	}, respectively. Figure reproduced with permission from Tommaso Lunghi.}
\label{fig:experiment2}
\end{figure}

The protocol was implemented with $n = 512$ bits, which for the two-round protocol gives the security parameter of $\varepsilon \approx 10^{-77}$ (against quantum adversaries). The multiround protocol was implemented with $m + 1 = 6$ rounds which gives the security parameter of $\varepsilon \approx 2.3 \times 10^{-10}$ (against classical adversaries). The total commitment time was 2 ms but placing exactly the same setup at the antipodes of the Earth would allow for commitment time of 212 ms. Of course, the commitment time could be made longer by employing more sophisticated hardware, which allows us to exchange more data within the relativistic constraints, but since our main goal was to demonstrate the feasibility of implementing multiround schemes we decided not to do it.
\chapter{Conclusions}
\label{chap:conclusions}
The central theme of this thesis is the study of how communication constraints can be used in classical and quantum cryptography, with a particular focus on commitment schemes. Communication constraints resulting from fundamental physical principles (like the fact that the speed of light is finite) are of particular interest. While relativity does not permit to implement the ideal commitment functionality, some weaker variants are possible and understanding similarities and differences between these schemes lies at the heart of this thesis.

It seems fair to claim that we now have a good understanding of relativistic commitment schemes, in particular the simplest class, in which no messages are exchanged when the commitment is valid. The main drawback of such schemes is the limitation on the commitment time, which is proportional to the distance between the agents. In order to increase the length of the commitment (without moving the agents further apart) one needs to resort to multiround schemes and we have presented a particular classical protocol along with a security proof against classical adversaries. This shows that in the classical world one can achieve arbitrary long commitments even if the agents are forced to occupy a finite region of space.

A natural follow-up question is to ask whether this statement remains true in the quantum world. Security analysis of the aforementioned multiround protocol against quantum adversaries is currently out of reach but we conjecture that the protocol remains secure against quantum adversaries (although with weaker security guarantees). Moreover, security analysis of relativistic protocols against quantum adversaries leads to a new, interesting class of problems: multiplayer games with communication constraints, which have not been studied before (except for a few special cases) and might be of independent interest.
Note that multiround commitment schemes against no-signalling adversaries have recently been shown to be impossible \cite{fehr15}.

Having understood the features and limitations of relativistic commitment schemes, the next step would be to look at more powerful primitives and oblivious transfer would be a natural choice. While we know that perfect oblivious transfer is not possible, one might think of weaker variants which have not been ruled out. We do not see a straightforward way of translating distributed oblivious transfer protocols into the relativistic setting (see Section~\ref{sec:simple-relativistic-protocols} for details) and, to the best of our knowledge, no explicitly relativistic schemes have been proposed. Investigating the possibility of relativistic oblivious transfer would constitute an important step towards characterising the exact power of quantum relativistic cryptography.
\bibliographystyle{alphaarxiv}
\bibliography{/home/jedrek/Projekty/tex/librarysan}

\newcommand{\etalchar}[1]{$^{#1}$}
\begin{thebibliography}{LWW{\etalchar{+}}10}

\bibitem[Aar05]{aaronson05}
S.~Aaronson.
\newblock {NP-complete Problems and Physical Reality}.
\newblock {\em ACM SIGACT News}, 36, 2005.
\newblock \\
  \texttt{DOI:\,\href{http://dx.doi.org/10.1145/1052796.1052804}{10.1145/1052796.1052804}}.

\bibitem[AK15a]{adlam15a}
E.~Adlam and A.~Kent.
\newblock {Deterministic Relativistic Quantum Bit Commitment}.
\newblock 2015.
\newblock \\
  \texttt{arXiv:\,\href{http://arxiv.org/abs/1504.00943}{1504.00943}}.

\bibitem[AK15b]{adlam15b}
E.~Adlam and A.~Kent.
\newblock {Device-Independent Relativistic Quantum Bit Commitment}.
\newblock 2015.
\newblock \\
  \texttt{arXiv:\,\href{http://arxiv.org/abs/1504.00944}{1504.00944}}.

\bibitem[ATVY00]{aharonov00}
D.~Aharonov, A.~Ta{-Shma}, U.~Vazirani, and A.~C.-C. Yao.
\newblock {Quantum Bit Escrow}.
\newblock {\em Proc. 32nd ACM STOC}, 2000.
\newblock \\
  \texttt{DOI:\,\href{http://dx.doi.org/10.1145/335305.335404}{10.1145/335305.335404}}.

\bibitem[Bab85]{babai85}
L.~Babai.
\newblock {Trading group theory for randomness}.
\newblock {\em Proc. 17th ACM STOC}, 1985.
\newblock \\
  \texttt{DOI:\,\href{http://dx.doi.org/10.1145/22145.22192}{10.1145/22145.22192}}.

\bibitem[Bab90]{babai90}
L.~Babai.
\newblock {E-mail and the unexpected power of interaction}.
\newblock {\em Proc. 5th Annual Structure in Complexity Theory Conference},
  1990.
\newblock \\
  \texttt{DOI:\,\href{http://dx.doi.org/10.1109/SCT.1990.113952}{10.1109/SCT.1990.113952}}.

\bibitem[BB84]{bennett84}
C.~H. Bennett and G.~Brassard.
\newblock {Quantum cryptography: Public key distribution and coin tossing}.
\newblock {\em Proc. IEEE Conference on Computers, Systems and Signal
  Processing}, 1984.
\newblock \\ Online: \url{http://www.cs.ucsb.edu/~chong/290N-W06/BB84.pdf}.

\bibitem[BBB{\etalchar{+}}92]{bennett92b}
C.~H. Bennett, F.~Bessette, G.~Brassard, L.~Salvail, and J.~Smolin.
\newblock {Experimental quantum cryptography}.
\newblock {\em J. Cryptology}, 5(1), 1992.
\newblock \\
  \texttt{DOI:\,\href{http://dx.doi.org/10.1007/BF00191318}{10.1007/BF00191318}}.

\bibitem[BBBW83]{bennett83}
C.~H. Bennett, G.~Brassard, S.~Breidbart, and S.~Wiesner.
\newblock {Quantum cryptography, or unforgeable subway tokens}.
\newblock {\em Advances in Cryptology: Proc. CRYPTO '82}, 1983.
\newblock \\
  \texttt{DOI:\,\href{http://dx.doi.org/10.1007/978-1-4757-0602-4\_26}{10.1007/978-1-4757-0602-4\_26}}.

\bibitem[BBC{\etalchar{+}}93]{bennett93}
C.~H. Bennett, G.~Brassard, C.~Cr{\'e}peau, R.~Jozsa, A.~Peres, and W.~K.
  Wootters.
\newblock {Teleporting an Unknown Quantum State via Dual Classical and
  Einstein-Podolsky-Rosen Channels}.
\newblock {\em Phys. Rev. Lett.}, 70(13), 1993.
\newblock \\
  \texttt{DOI:\,\href{http://dx.doi.org/10.1103/PhysRevLett.70.1895}{10.1103/PhysRevLett.70.1895}}.

\bibitem[BBCS92]{bennett92}
C.~H. Bennett, G.~Brassard, C.~Cr{\'e}peau, and M.~H. Skubiszewska.
\newblock {Practical Quantum Oblivious Transfer}.
\newblock {\em Advances in Cryptology: Proc. CRYPTO '91, LNCS}, 576, 1992.
\newblock \\
  \texttt{DOI:\,\href{http://dx.doi.org/10.1007/3-540-46766-1\_29}{10.1007/3-540-46766-1\_29}}.

\bibitem[BBM92]{bennett92a}
C.~H. Bennett, G.~Brassard, and N.~D. Mermin.
\newblock {Quantum Cryptography without Bell's Theorem}.
\newblock {\em Phys. Rev. Lett.}, 68(5), 1992.
\newblock \\
  \texttt{DOI:\,\href{http://dx.doi.org/10.1103/PhysRevLett.68.557}{10.1103/PhysRevLett.68.557}}.

\bibitem[BC91]{brassard91}
G.~Brassard and C.~Cr{\'e}peau.
\newblock {Quantum Bit Commitment and Coin Tossing Protocols}.
\newblock {\em Advances in Cryptology: Proc. CRYPTO '90, LNCS}, 537, 1991.
\newblock \\
  \texttt{DOI:\,\href{http://dx.doi.org/10.1007/3-540-38424-3\_4}{10.1007/3-540-38424-3\_4}}.

\bibitem[BC96]{brassard96}
G.~Brassard and C.~Cr{\'e}peau.
\newblock {Cryptology column -- 25 years of quantum cryptography}.
\newblock {\em ACM SIGACT News}, 27(3), 1996.
\newblock \\
  \texttt{DOI:\,\href{http://dx.doi.org/10.1145/235666.235669}{10.1145/235666.235669}}.

\bibitem[BCC88]{brassard88}
G.~Brassard, D.~Chaum, and C.~Cr{\'e}peau.
\newblock {Minimum Disclosure Proofs of Knowledge}.
\newblock {\em J. Comput. System Sci.}, 37(2), 1988.
\newblock \\
  \texttt{DOI:\,\href{http://dx.doi.org/10.1016/0022-0000(88)90005-0}{10.1016/0022-0000(88)90005-0}}.

\bibitem[BCC{\etalchar{+}}09]{berta09}
M.~Berta, M.~Christandl, R.~Colbeck, J.~M. Renes, and R.~Renner.
\newblock {The uncertainty principle in the presence of quantum memory}.
\newblock {\em Nat. Phys.}, 6, 2009.
\newblock \\
  \texttt{DOI:\,\href{http://dx.doi.org/10.1038/nphys1734}{10.1038/nphys1734}}.

\bibitem[BCF{\etalchar{+}}11]{buhrman14}
H.~Buhrman, N.~Chandran, S.~Fehr, R.~Gelles, V.~Goyal, R.~Ostrovsky, and
  C.~Schaffner.
\newblock {Position-based quantum cryptography: impossibility and
  constructions}.
\newblock {\em Advances in Cryptology: Proc. CRYPTO '11, LNCS}, 6841, 2011.
\newblock \\
  \texttt{DOI:\,\href{http://dx.doi.org/10.1007/978-3-642-22792-9\_24}{10.1007/978-3-642-22792-9\_24}}.

\bibitem[BCH{\etalchar{+}}06]{buhrman06a}
H.~Buhrman, M.~Christandl, P.~Hayden, H.-K. Lo, and S.~Wehner.
\newblock {Security of Quantum Bit String Commitment Depends on the Information
  Measure}.
\newblock {\em Phys. Rev. Lett.}, 97(25), 2006.
\newblock \\
  \texttt{DOI:\,\href{http://dx.doi.org/10.1103/PhysRevLett.97.250501}{10.1103/PhysRevLett.97.250501}}.

\bibitem[BCH{\etalchar{+}}08]{buhrman08}
H.~Buhrman, M.~Christandl, P.~Hayden, H.-K. Lo, and S.~Wehner.
\newblock {Possibility, impossibility, and cheat sensitivity of quantum-bit
  string commitment}.
\newblock {\em Phys. Rev. A}, 78(2), 2008.
\newblock \\
  \texttt{DOI:\,\href{http://dx.doi.org/10.1103/PhysRevA.78.022316}{10.1103/PhysRevA.78.022316}}.

\bibitem[BCJL93]{brassard93}
G.~Brassard, C.~Cr{\'e}peau, R.~Jozsa, and D.~Langlois.
\newblock {A quantum bit commitment scheme provably unbreakable by both
  parties}.
\newblock {\em Proc. 34th IEEE FOCS}, 1993.
\newblock \\
  \texttt{DOI:\,\href{http://dx.doi.org/10.1109/SFCS.1993.366851}{10.1109/SFCS.1993.366851}}.

\bibitem[BCMS97]{brassard97}
G.~Brassard, C.~Cr{\'e}peau, D.~Mayers, and L.~Salvail.
\newblock {A brief review on the impossibility of quantum bit commitment}.
\newblock 1997.
\newblock \\
  \texttt{arXiv:\,\href{http://arxiv.org/abs/quant-ph/9712023}{quant-ph/9712023}}.

\bibitem[BCP{\etalchar{+}}14]{brunner14}
N.~Brunner, D.~Cavalcanti, S.~Pironio, V.~Scarani, and S.~Wehner.
\newblock {Bell nonlocality}.
\newblock {\em Rev. Mod. Phys.}, 86(2), 2014.
\newblock \\
  \texttt{DOI:\,\href{http://dx.doi.org/10.1103/RevModPhys.86.419}{10.1103/RevModPhys.86.419}}.

\bibitem[BCS12]{buhrman12}
H.~Buhrman, M.~Christandl, and C.~Schaffner.
\newblock {Complete insecurity of quantum protocols for classical two-party
  computation}.
\newblock {\em Phys. Rev. Lett.}, 109(16), 2012.
\newblock \\
  \texttt{DOI:\,\href{http://dx.doi.org/10.1103/PhysRevLett.109.160501}{10.1103/PhysRevLett.109.160501}}.

\bibitem[BCU{\etalchar{+}}06]{buhrman06}
H.~Buhrman, M.~Christandl, F.~Unger, S.~Wehner, and A.~Winter.
\newblock {Implications of superstrong non-locality for cryptography}.
\newblock {\em Proc. R. Soc. A}, 462, 2006.
\newblock \\
  \texttt{DOI:\,\href{http://dx.doi.org/10.1098/rspa.2006.1663}{10.1098/rspa.2006.1663}}.

\bibitem[Bel64]{bell64}
J.~S. Bell.
\newblock {On the Einstein-Podolsky-Rosen paradox}.
\newblock {\em Physics}, 1, 1964.

\bibitem[Bel11]{bellovin11}
S.~M. Bellovin.
\newblock {Frank Miller: Inventor of the One-Time Pad}.
\newblock {\em Cryptologia}, 35(3), 2011.
\newblock \\
  \texttt{DOI:\,\href{http://dx.doi.org/10.1080/01611194.2011.583711}{10.1080/01611194.2011.583711}}.

\bibitem[BFL91]{babai91}
L.~Babai, L.~Fortnow, and C.~Lund.
\newblock {Non-deterministic exponential time has two-prover interactive
  protocols}.
\newblock {\em Comput. Complex.}, 1(1), 1991.
\newblock \\
  \texttt{DOI:\,\href{http://dx.doi.org/10.1007/BF01200056}{10.1007/BF01200056}}.

\bibitem[BFW14]{berta14}
M.~Berta, O.~Fawzi, and S.~Wehner.
\newblock {Quantum to Classical Randomness Extractors}.
\newblock {\em IEEE Trans. Inf. Theory}, 60(2), 2014.
\newblock \\
  \texttt{DOI:\,\href{http://dx.doi.org/10.1109/TIT.2013.2291780}{10.1109/TIT.2013.2291780}}.

\bibitem[BGKW88]{benor88}
M.~Ben{-Or}, S.~Goldwasser, J.~Kilian, and A.~Wigderson.
\newblock {Multi-Prover Interactive Proofs: How to Remove Intractability
  Assumptions}.
\newblock {\em Proc. 20th ACM STOC}, 1988.
\newblock \\
  \texttt{DOI:\,\href{http://dx.doi.org/10.1145/62212.62223}{10.1145/62212.62223}}.

\bibitem[Bha97]{bhatia97}
R.~Bhatia.
\newblock {\em {Matrix Analysis}}.
\newblock Springer New York, 1997.

\bibitem[Bha09]{bhatia09}
R.~Bhatia.
\newblock {\em {Positive Definite Matrices}}.
\newblock Princeton University Press, 2009.

\bibitem[Blu81]{blum81}
M.~Blum.
\newblock {Coin Flipping by Telephone}.
\newblock {\em Advances in Cryptology: A Report on CRYPTO '81}, 1981.
\newblock \\
  \texttt{DOI:\,\href{http://dx.doi.org/10.1145/1008908.1008911}{10.1145/1008908.1008911}}.

\bibitem[BM05]{buhrman05}
H.~Buhrman and S.~Massar.
\newblock {Causality and Tsirelson's bounds}.
\newblock {\em Phys. Rev. A}, 72(5), 2005.
\newblock \\
  \texttt{DOI:\,\href{http://dx.doi.org/10.1103/PhysRevA.72.052103}{10.1103/PhysRevA.72.052103}}.

\bibitem[BN00]{boneh00}
D.~Boneh and M.~Naor.
\newblock {Timed Commitments}.
\newblock {\em Advances in Cryptology: Proc. CRYPTO '00, LNCS}, 1880(3), 2000.
\newblock \\
  \texttt{DOI:\,\href{http://dx.doi.org/10.1007/3-540-44598-6\_15}{10.1007/3-540-44598-6\_15}}.

\bibitem[Bor26]{born26}
M.~Born.
\newblock {On The Quantum Mechanics Of Collisions}.
\newblock {\em Zeitschrift f{\"{u}}r Physik}, 37(12), 1926.

\bibitem[BS98]{bennett98}
C.~H. Bennett and P.~W. Shor.
\newblock {Quantum Information Theory}.
\newblock {\em IEEE Trans. Inf. Theory}, 44(6), 1998.
\newblock \\
  \texttt{DOI:\,\href{http://dx.doi.org/10.1109/18.720553}{10.1109/18.720553}}.

\bibitem[BS15]{bavarian15}
M.~Bavarian and P.~W. Shor.
\newblock {Information Causality, Szemer{\'e}di-Trotter and Algebraic Variants
  of CHSH}.
\newblock {\em Proc. Conference on Innovations in Theoretical Computer
  Science}, 2015.
\newblock \\
  \texttt{DOI:\,\href{http://dx.doi.org/10.1145/2688073.2688112}{10.1145/2688073.2688112}}.

\bibitem[BT08]{broadbent08}
A.~Broadbent and A.~Tapp.
\newblock {Information-Theoretically Secure Voting Without an Honest Majority}.
\newblock {\em Proc. IAVoSS Workshop on Trustworthy Elections}, 2008.
\newblock \\ \texttt{arXiv:\,\href{http://arxiv.org/abs/0806.1931}{0806.1931}}.

\bibitem[Can01]{canetti01}
R.~Canetti.
\newblock {Universally Composable Security: A New Paradigm for Cryptographic
  Protocols}.
\newblock {\em Proc. 42nd IEEE FOCS}, 2001.
\newblock \\
  \texttt{DOI:\,\href{http://dx.doi.org/10.1109/SFCS.2001.959888}{10.1109/SFCS.2001.959888}}.

\bibitem[CCL15]{chakroborty15}
K.~Chakraborty, A.~Chailloux, and A.~Leverrier.
\newblock {Arbitrarily long relativistic bit commitment}.
\newblock 2015.
\newblock \\
  \texttt{arXiv:\,\href{http://arxiv.org/abs/arXiv:1507.00239}{arXiv:1507.00239}}.

\bibitem[CDP{\etalchar{+}}13]{chiribella13}
G.~Chiribella, G.~M. D'Ariano, P.~Perinotti, D.~Schlingemann, and R.~F. Werner.
\newblock {A short impossibility proof of quantum bit commitment}.
\newblock {\em Phys. Lett. A}, 377(15), 2013.
\newblock \\
  \texttt{DOI:\,\href{http://dx.doi.org/10.1016/j.physleta.2013.02.045}{10.1016/j.physleta.2013.02.045}}.

\bibitem[CFL83]{chandra83}
A.~K. Chandra, M.~L. Furst, and R.~J. Lipton.
\newblock {Multi-party protocols}.
\newblock {\em Proc. 15th ACM STOC}, 1983.
\newblock \\
  \texttt{DOI:\,\href{http://dx.doi.org/10.1145/800061.808737}{10.1145/800061.808737}}.

\bibitem[Che52]{chernoff52}
H.~Chernoff.
\newblock {A measure of asymptotic efficiency for tests of a hypothesis based
  on the sum of observations}.
\newblock {\em Ann. Math. Statist.}, 23, 1952.

\bibitem[CHSH69]{clauser69}
J.~F. Clauser, M.~A. Horne, A.~Shimony, and R.~A. Holt.
\newblock {Proposed experiment to test local hidden-variable theories}.
\newblock {\em Phys. Rev. Lett.}, 23(15), 1969.
\newblock \\
  \texttt{DOI:\,\href{http://dx.doi.org/10.1103/PhysRevLett.23.880}{10.1103/PhysRevLett.23.880}}.

\bibitem[CHTW04]{cleve04}
R.~Cleve, P.~H{\o}yer, B.~Toner, and J.~Watrous.
\newblock {Consequences and Limits of Nonlocal Strategies}.
\newblock {\em Proc. IEEE Comput. Comp. '04}, 2004.
\newblock \\
  \texttt{DOI:\,\href{http://dx.doi.org/10.1109/CCC.2004.1313847}{10.1109/CCC.2004.1313847}}.

\bibitem[CJPPG15]{cooney15}
T.~Cooney, M.~Junge, C.~Palazuelos, and D.~P{\'e}rez-Garc{\'{\i}}a.
\newblock {Rank-one quantum games}.
\newblock {\em Comput. Complex.}, 24, 2015.
\newblock \\
  \texttt{DOI:\,\href{http://dx.doi.org/10.1007/s00037-014-0096-x}{10.1007/s00037-014-0096-x}}.

\bibitem[CK06]{colbeck06a}
R.~Colbeck and A.~Kent.
\newblock {Variable-bias coin tossing}.
\newblock {\em Phys. Rev. A}, 73(3), 2006.
\newblock \\
  \texttt{DOI:\,\href{http://dx.doi.org/10.1103/PhysRevA.73.032320}{10.1103/PhysRevA.73.032320}}.

\bibitem[CK11]{chailloux11}
A.~Chailloux and I.~Kerenidis.
\newblock {Optimal bounds for quantum bit commitment}.
\newblock {\em Proc. 52nd IEEE FOCS}, 2011.
\newblock \\
  \texttt{DOI:\,\href{http://dx.doi.org/10.1109/FOCS.2011.42}{10.1109/FOCS.2011.42}}.

\bibitem[CK12]{croke12}
S.~Croke and A.~Kent.
\newblock {Security details for bit commitment by transmitting measurement
  outcomes}.
\newblock {\em Phys. Rev. A}, 86(5), 2012.
\newblock \\
  \texttt{DOI:\,\href{http://dx.doi.org/10.1103/PhysRevA.86.052309}{10.1103/PhysRevA.86.052309}}.

\bibitem[Col06]{colbeck06}
R.~Colbeck.
\newblock {\em {Quantum And Relativistic Protocols For Secure Multi-Party
  Computation}}.
\newblock PhD thesis, University of Cambridge, 2006.
\newblock \\ \texttt{arXiv:\,\href{http://arxiv.org/abs/0911.3814}{0911.3814}}.

\bibitem[Col07]{colbeck07a}
R.~Colbeck.
\newblock {Impossibility of secure two-party classical computation}.
\newblock {\em Phys. Rev. A}, 76(6), 2007.
\newblock \\
  \texttt{DOI:\,\href{http://dx.doi.org/10.1103/PhysRevA.76.062308}{10.1103/PhysRevA.76.062308}}.

\bibitem[Cr{\'e}88]{crepeau88}
C.~Cr{\'e}peau.
\newblock {Equivalence Between Two Flavours of Oblivious Transfers}.
\newblock {\em Advances in Cryptology: Proc. CRYPTO '87, LNCS}, 293, 1988.
\newblock \\
  \texttt{DOI:\,\href{http://dx.doi.org/10.1007/3-540-48184-2\_30}{10.1007/3-540-48184-2\_30}}.

\bibitem[Cr{\'e}96]{crepeau96}
C.~Cr{\'e}peau.
\newblock {What is going on with Quantum Bit Commitment?}
\newblock {\em Proc. Pragocrypt '96: 1st International Conference on the Theory
  and Applications of Cryptology}, 1996.
\newblock \\ Online: \url{http://www.cs.mcgill.ca/~crepeau/PS/Cre96a.ps}.

\bibitem[Cr{\'e}97]{crepeau97}
C.~Cr{\'e}peau.
\newblock {Efficient Cryptographic Protocols Based on Noisy Channels}.
\newblock {\em Advances in Cryptology: Proc. EUROCRYPT '97, LNCS}, 1233, 1997.
\newblock \\
  \texttt{DOI:\,\href{http://dx.doi.org/10.1007/3-540-69053-0\_21}{10.1007/3-540-69053-0\_21}}.

\bibitem[Cr{\'e}11]{crepeau11a}
C.~Cr{\'e}peau.
\newblock {Commitment}.
\newblock {\em Encyclopedia of Security and Cryptography}, 2011.
\newblock \\
  \texttt{DOI:\,\href{http://dx.doi.org/10.1007/978-1-4419-5906-5\_239}{10.1007/978-1-4419-5906-5\_239}}.

\bibitem[CSST11]{crepeau11}
C.~Cr{\'e}peau, L.~Salvail, J.-R. Simard, and A.~Tapp.
\newblock {Two provers in isolation}.
\newblock {\em Advances in Cryptology: Proc. ASIACRYPT '11, LNCS}, 7073, 2011.
\newblock \\
  \texttt{DOI:\,\href{http://dx.doi.org/10.1007/978-3-642-25385-0\_22}{10.1007/978-3-642-25385-0\_22}}.

\bibitem[Deu83]{deutsch83}
D.~Deutsch.
\newblock {Uncertainty in Quantum Measurements}.
\newblock {\em Phys. Rev. Lett.}, 50(9), 1983.
\newblock \\
  \texttt{DOI:\,\href{http://dx.doi.org/10.1103/PhysRevLett.50.631}{10.1103/PhysRevLett.50.631}}.

\bibitem[Deu85]{deutsch85}
D.~Deutsch.
\newblock {Quantum theory, the Church-Turing principle and the universal
  quantum computer}.
\newblock {\em Proc. R. Soc. Lond. A}, 400, 1985.
\newblock \\
  \texttt{DOI:\,\href{http://dx.doi.org/10.1098/rspa.1985.0070}{10.1098/rspa.1985.0070}}.

\bibitem[DFR{\etalchar{+}}07]{damgard07}
I.~B. Damg{\aa}rd, S.~Fehr, R.~Renner, L.~Salvail, and C.~Schaffner.
\newblock {A Tight High-Order Entropic Quantum Uncertainty Relation with
  Applications}.
\newblock {\em Advances in Cryptology: Proc. CRYPTO '07, LNCS}, 2007.
\newblock \\
  \texttt{DOI:\,\href{http://dx.doi.org/10.1007/978-3-540-74143-5\_20}{10.1007/978-3-540-74143-5\_20}}.

\bibitem[DFSS05]{damgard05}
I.~B. Damg{\aa}rd, S.~Fehr, L.~Salvail, and C.~Schaffner.
\newblock {Cryptography In the Bounded Quantum-Storage Model}.
\newblock {\em Proc. 46th IEEE FOCS}, 2005.
\newblock \\
  \texttt{DOI:\,\href{http://dx.doi.org/10.1109/SFCS.2005.30}{10.1109/SFCS.2005.30}}.

\bibitem[Dir39]{dirac39}
P.~A. Dirac.
\newblock {A new notation for quantum mechanics}.
\newblock {\em Mathematical Proceedings of the Cambridge Philosophical
  Society}, 35(03), 1939.
\newblock \\
  \texttt{DOI:\,\href{http://dx.doi.org/10.1017/S0305004100021162}{10.1017/S0305004100021162}}.

\bibitem[DKSW07]{dariano07}
G.~M. D'Ariano, D.~Kretschmann, D.~Schlingemann, and R.~F. Werner.
\newblock {Reexamination of quantum bit commitment: The possible and the
  impossible}.
\newblock {\em Phys. Rev. A}, 76(3), 2007.
\newblock \\
  \texttt{DOI:\,\href{http://dx.doi.org/10.1103/PhysRevA.76.032328}{10.1103/PhysRevA.76.032328}}.

\bibitem[DLTW08]{doherty08}
A.~C. Doherty, Y.-C. Liang, B.~Toner, and S.~Wehner.
\newblock {The Quantum Moment Problem and Bounds on Entangled Multi-prover
  Games}.
\newblock {\em Proc. IEEE Comput. Comp. '08}, 2008.
\newblock \\
  \texttt{DOI:\,\href{http://dx.doi.org/10.1109/CCC.2008.26}{10.1109/CCC.2008.26}}.

\bibitem[EBH09]{eaton09}
J.~W. Eaton, D.~Bateman, and S.~Hauberg.
\newblock {\em {GNU Octave version 3.0.1 manual: a high-level interactive
  language for numerical computations}}.
\newblock CreateSpace Independent Publishing Platform, 2009.
\newblock \\ Online: \url{http://www.gnu.org/software/octave/}.

\bibitem[Eke91]{ekert91}
A.~K. Ekert.
\newblock {Quantum Cryptography Based on Bell's Theorem}.
\newblock {\em Phys. Rev. Lett.}, 67(6), 1991.
\newblock \\
  \texttt{DOI:\,\href{http://dx.doi.org/10.1103/PhysRevLett.67.661}{10.1103/PhysRevLett.67.661}}.

\bibitem[EPR35]{einstein35}
A.~Einstein, B.~Podolsky, and N.~Rosen.
\newblock {Can Quantum-Mechanical Description of Physical Reality Be Considered
  Complete?}
\newblock {\em Phys. Rev.}, 47, 1935.

\bibitem[ER14]{ekert14}
A.~K. Ekert and R.~Renner.
\newblock {The ultimate physical limits of privacy}.
\newblock {\em Nature}, 507(7493), 2014.
\newblock \\
  \texttt{DOI:\,\href{http://dx.doi.org/10.1038/nature13132}{10.1038/nature13132}}.

\bibitem[Fey82]{feynman82}
R.~P. Feynman.
\newblock {Simulating Physics with computers}.
\newblock {\em International Journal of Theoretical Physics}, 21(6-7), 1982.
\newblock \\
  \texttt{DOI:\,\href{http://dx.doi.org/10.1007/BF02650179}{10.1007/BF02650179}}.

\bibitem[FF15a]{fehr15}
S.~Fehr and M.~Fillinger.
\newblock {Multi-Prover Commitments Against Non-Signaling Attacks}.
\newblock {\em Advances in Cryptology: Proc. CRYPTO '15, LNCS}, 2015.
\newblock \\
  \texttt{DOI:\,\href{http://dx.doi.org/10.1007/978-3-662-48000-7\_20}{10.1007/978-3-662-48000-7\_20}}.

\bibitem[FF15b]{fehr15a}
S.~Fehr and M.~Fillinger.
\newblock {On the Composition of Two-Prover Commitments, and Applications to
  Multi-Round Relativistic Commitments}.
\newblock 2015.
\newblock \\
  \texttt{arXiv:\,\href{http://arxiv.org/abs/1507.00240}{1507.00240}}.

\bibitem[Fri12]{fritz12}
T.~Fritz.
\newblock {Beyond Bell's theorem: correlation scenarios}.
\newblock {\em New J. Phys.}, 14, 2012.
\newblock \\
  \texttt{DOI:\,\href{http://dx.doi.org/10.1088/1367-2630/14/10/103001}{10.1088/1367-2630/14/10/103001}}.

\bibitem[Fri14]{fritz14}
T.~Fritz.
\newblock {Beyond Bell's Theorem II: Scenarios with arbitrary causal
  structure}.
\newblock 2014.
\newblock \\ \texttt{arXiv:\,\href{http://arxiv.org/abs/1404.4812}{1404.4812}}.

\bibitem[FRS94]{fortnow94}
L.~Fortnow, J.~Rompel, and M.~Sipser.
\newblock {On the power of multi-prover interactive protocols}.
\newblock {\em Theoretical Computer Science}, 134(2), 1994.
\newblock \\
  \texttt{DOI:\,\href{http://dx.doi.org/10.1016/0304-3975(94)90251-8}{10.1016/0304-3975(94)90251-8}}.

\bibitem[Gas04]{gasarch04}
W.~Gasarch.
\newblock {A survey on private information retrieval}.
\newblock {\em Bull. Eur. Assoc. Theor. Comput. Sci.}, 82, 2004.

\bibitem[GIKM00]{gertner00}
Y.~Gertner, Y.~Ishai, E.~Kushilevitz, and T.~Malkin.
\newblock {Protecting Data Privacy in Private Information Retrieval Schemes}.
\newblock {\em J. Comput. System Sci.}, 60(3), 2000.
\newblock \\
  \texttt{DOI:\,\href{http://dx.doi.org/10.1006/jcss.1999.1689}{10.1006/jcss.1999.1689}}.

\bibitem[GKR08]{goldwasser08}
S.~Goldwasser, Y.~T. Kalai, and G.~N. Rothblum.
\newblock {One-time Programs}.
\newblock {\em Advances in Cryptology: Proc. CRYPTO '08, LNCS}, 5157, 2008.
\newblock \\
  \texttt{DOI:\,\href{http://dx.doi.org/10.1007/978-3-540-85174-5\_3}{10.1007/978-3-540-85174-5\_3}}.

\bibitem[GMR85]{goldwasser85}
S.~Goldwasser, S.~Micali, and C.~Rackoff.
\newblock {The Knowledge Complexity of Interactive Proof Systems}.
\newblock {\em Proc. 17th ACM STOC}, 18(1), 1985.
\newblock \\
  \texttt{DOI:\,\href{http://dx.doi.org/10.1137/0218012}{10.1137/0218012}}.

\bibitem[GMW86]{goldreich86}
O.~Goldreich, S.~Micali, and A.~Wigderson.
\newblock {Proofs that Yield Nothing But their Validity and a Methodology of
  Cryptographic Protocol Design}.
\newblock {\em Proc. 27th ACM STOC}, 1986.
\newblock \\
  \texttt{DOI:\,\href{http://dx.doi.org/10.1109/SFCS.1986.47}{10.1109/SFCS.1986.47}}.

\bibitem[Gol08]{goldreich08}
O.~Goldreich.
\newblock {Probabilistic Proof Systems: A Primer}.
\newblock 2008.
\newblock \\ Online: \url{http://www.wisdom.weizmann.ac.il/~oded/PS/pps5.pdf}.

\bibitem[Gro96]{grover96}
L.~K. Grover.
\newblock {A fast quantum mechanical algorithm for database search}.
\newblock {\em Proc. 28th ACM STOC}, 1996.
\newblock \\
  \texttt{DOI:\,\href{http://dx.doi.org/10.1145/237814.237866}{10.1145/237814.237866}}.

\bibitem[GWC{\etalchar{+}}14]{gallego14}
R.~Gallego, L.~E. W{\"{u}}rflinger, R.~Chaves, A.~Ac{\'{\i}}n, and
  M.~Navascu{\'e}s.
\newblock {Nonlocality in sequential correlation scenarios}.
\newblock {\em New J. Phys.}, 16(3), 2014.
\newblock \\
  \texttt{DOI:\,\href{http://dx.doi.org/10.1088/1367-2630/16/3/033037}{10.1088/1367-2630/16/3/033037}}.

\bibitem[Hel76]{helstrom76}
C.~W. Helstrom.
\newblock {\em {Quantum detection and estimation theory}}.
\newblock Academic Press, New York, USA, 1976.

\bibitem[HK04]{kent04}
L.~Hardy and A.~Kent.
\newblock {Cheat Sensitive Quantum Bit Commitment}.
\newblock {\em Phys. Rev. Lett.}, 92(15), 2004.
\newblock \\
  \texttt{DOI:\,\href{http://dx.doi.org/10.1103/PhysRevLett.92.157901}{10.1103/PhysRevLett.92.157901}}.

\bibitem[HM12]{hayden12}
P.~Hayden and A.~May.
\newblock {Summoning Information in Spacetime, or Where and When Can a Qubit
  Be?}
\newblock 2012.
\newblock \\ \texttt{arXiv:\,\href{http://arxiv.org/abs/1210.0913}{1210.0913}}.

\bibitem[HM15]{heinosaari15}
T.~Heinosaari and T.~Miyadera.
\newblock {Universality of Sequential Quantum Measurements}.
\newblock {\em Phys. Rev. A}, 022110, 2015.
\newblock \\
  \texttt{DOI:\,\href{http://dx.doi.org/10.1103/PhysRevA.91.022110}{10.1103/PhysRevA.91.022110}}.

\bibitem[Hol73]{holevo73}
A.~S. Holevo.
\newblock {Statistical problems in quantum physics}.
\newblock {\em Proc. 2nd Japan-USSR Symposium on Probability Theory, LNCS},
  330, 1973.
\newblock \\
  \texttt{DOI:\,\href{http://dx.doi.org/10.1007/BFb0061483}{10.1007/BFb0061483}}.

\bibitem[KdW04]{kerenidis04}
I.~Kerenidis and R.~de~Wolf.
\newblock {Quantum symmetrically-private information retrieval}.
\newblock {\em Inform. Process. Lett.}, 90(3), 2004.
\newblock \\
  \texttt{DOI:\,\href{http://dx.doi.org/10.1016/j.ipl.2004.02.003}{10.1016/j.ipl.2004.02.003}}.

\bibitem[Ken99]{kent99}
A.~Kent.
\newblock {Unconditionally Secure Bit Commitment}.
\newblock {\em Phys. Rev. Lett.}, 83(7), 1999.
\newblock \\
  \texttt{DOI:\,\href{http://dx.doi.org/10.1103/PhysRevLett.83.1447}{10.1103/PhysRevLett.83.1447}}.

\bibitem[Ken05]{kent05}
A.~Kent.
\newblock {Secure Classical Bit Commitment Using Fixed Capacity Communication
  Channels}.
\newblock {\em J. Cryptology}, 18(4), 2005.
\newblock \\
  \texttt{DOI:\,\href{http://dx.doi.org/10.1007/s00145-005-0905-8}{10.1007/s00145-005-0905-8}}.

\bibitem[Ken11]{kent11}
A.~Kent.
\newblock {Unconditionally secure bit commitment with flying qudits}.
\newblock {\em New J. Phys.}, 13, 2011.
\newblock \\
  \texttt{DOI:\,\href{http://dx.doi.org/10.1088/1367-2630/13/11/113015}{10.1088/1367-2630/13/11/113015}}.

\bibitem[Ken12a]{kent12a}
A.~Kent.
\newblock {Quantum tasks in Minkowski space}.
\newblock {\em Class. Quantum Grav.}, 29(22), 2012.
\newblock \\
  \texttt{DOI:\,\href{http://dx.doi.org/10.1088/0264-9381/29/22/224013}{10.1088/0264-9381/29/22/224013}}.

\bibitem[Ken12b]{kent12}
A.~Kent.
\newblock {Unconditionally Secure Bit Commitment by Transmitting Measurement
  Outcomes}.
\newblock {\em Phys. Rev. Lett.}, 109(13), 2012.
\newblock \\
  \texttt{DOI:\,\href{http://dx.doi.org/10.1103/PhysRevLett.109.130501}{10.1103/PhysRevLett.109.130501}}.

\bibitem[Ken12c]{kent12b}
A.~Kent.
\newblock {Why classical certification is impossible in a quantum world}.
\newblock {\em Quant. Inf. Proc.}, 11(2), 2012.
\newblock \\
  \texttt{DOI:\,\href{http://dx.doi.org/10.1007/s11128-011-0262-x}{10.1007/s11128-011-0262-x}}.

\bibitem[Ken13]{kent13}
A.~Kent.
\newblock {A no-summoning theorem in relativistic quantum theory}.
\newblock {\em Quant. Inf. Proc.}, 12(2), 2013.
\newblock \\
  \texttt{DOI:\,\href{http://dx.doi.org/10.1007/s11128-012-0431-6}{10.1007/s11128-012-0431-6}}.

\bibitem[Kil88]{kilian88}
J.~Kilian.
\newblock {Founding Cryptography on Oblivious Transfer}.
\newblock {\em Proc. 20th ACM STOC}, 1988.
\newblock \\
  \texttt{DOI:\,\href{http://dx.doi.org/10.1145/62212.62215}{10.1145/62212.62215}}.

\bibitem[KM02]{kobayashi02}
H.~Kobayashi and K.~Matsumoto.
\newblock {Quantum Multi-prover Interactive Proof Systems with Limited Prior
  Entanglement}.
\newblock {\em Proc. ISAAC '02, LNCS}, 2518, 2002.
\newblock \\
  \texttt{DOI:\,\href{http://dx.doi.org/10.1007/3-540-36136-7\_11}{10.1007/3-540-36136-7\_11}}.

\bibitem[KMS11]{kent11a}
A.~Kent, W.~J. Munro, and T.~P. Spiller.
\newblock {Quantum tagging: Authenticating location via quantum information and
  relativistic signaling constraints}.
\newblock {\em Phys. Rev. A}, 84(1), 2011.
\newblock \\
  \texttt{DOI:\,\href{http://dx.doi.org/10.1103/PhysRevA.84.012326}{10.1103/PhysRevA.84.012326}}.

\bibitem[KTHW13]{kaniewski13}
J.~Kaniewski, M.~Tomamichel, E.~H{\"{a}}nggi, and S.~Wehner.
\newblock {Secure Bit Commitment From Relativistic Constraints}.
\newblock {\em IEEE Trans. Inf. Theory}, 59(7): 4687--4699, 2013.
\newblock \\
  \texttt{DOI:\,\href{http://dx.doi.org/10.1109/TIT.2013.2247463}{10.1109/TIT.2013.2247463}}.

\bibitem[KWW12]{konig12}
R.~K{\"{o}}nig, S.~Wehner, and J.~Wullschleger.
\newblock {Unconditional Security From Noisy Quantum Storage}.
\newblock {\em IEEE Trans. Inf. Theory}, 58(3), 2012.
\newblock \\
  \texttt{DOI:\,\href{http://dx.doi.org/10.1109/TIT.2011.2177772}{10.1109/TIT.2011.2177772}}.

\bibitem[LC97]{lo97}
H.-K. Lo and H.~Chau.
\newblock {Is Quantum Bit Commitment Really Impossible?}
\newblock {\em Phys. Rev. Lett.}, 78(17), 1997.
\newblock \\
  \texttt{DOI:\,\href{http://dx.doi.org/10.1103/PhysRevLett.78.3410}{10.1103/PhysRevLett.78.3410}}.

\bibitem[LCC{\etalchar{+}}14]{liu14}
Y.~Liu, Y.~Cao, M.~Curty, S.~K. Liao, J.~Wang, K.~Cui, Y.~H. Li, Z.~H. Lin,
  Q.~C. Sun, D.~D. Li, H.~F. Zhang, Y.~Zhao, T.~Y. Chen, C.~Z. Peng, Q.~Zhang,
  A.~Cabello, and J.~W. Pan.
\newblock {Experimental unconditionally secure bit commitment}.
\newblock {\em Phys. Rev. Lett.}, 112(1), 2014.
\newblock \\
  \texttt{DOI:\,\href{http://dx.doi.org/10.1103/PhysRevLett.112.010504}{10.1103/PhysRevLett.112.010504}}.

\bibitem[LKB{\etalchar{+}}13]{lunghi13}
T.~Lunghi, J.~Kaniewski, F.~Bussi{\`{e}}res, R.~Houlmann, M.~Tomamichel,
  A.~Kent, N.~Gisin, S.~Wehner, and H.~Zbinden.
\newblock {Experimental Bit Commitment Based on Quantum Communication and
  Special Relativity}.
\newblock {\em Phys. Rev. Lett.}, 111(18), 2013.
\newblock \\
  \texttt{DOI:\,\href{http://dx.doi.org/10.1103/PhysRevLett.111.180504}{10.1103/PhysRevLett.111.180504}}.

\bibitem[LKB{\etalchar{+}}15]{lunghi15}
T.~Lunghi, J.~Kaniewski, F.~Bussi{\`{e}}res, R.~Houlmann, M.~Tomamichel,
  S.~Wehner, and H.~Zbinden.
\newblock {Practical Relativistic Bit Commitment}.
\newblock {\em Phys. Rev. Lett.}, 115(3), 2015.
\newblock \\
  \texttt{DOI:\,\href{http://dx.doi.org/10.1103/PhysRevLett.115.030502}{10.1103/PhysRevLett.115.030502}}.

\bibitem[Lo97]{lo97a}
H.-K. Lo.
\newblock {Insecurity of Quantum Secure Computations}.
\newblock {\em Phys. Rev. A}, 56(2), 1997.
\newblock \\
  \texttt{DOI:\,\href{http://dx.doi.org/10.1103/PhysRevA.56.1154}{10.1103/PhysRevA.56.1154}}.

\bibitem[LWW{\etalchar{+}}10]{lydersen10}
L.~Lydersen, C.~Wiechers, C.~Wittmann, D.~Elser, J.~Skaar, and V.~Makarov.
\newblock {Hacking commercial quantum cryptography systems by tailored bright
  illumination}.
\newblock {\em Nat. Phot.}, 4(686), 2010.
\newblock \\
  \texttt{DOI:\,\href{http://dx.doi.org/10.1038/NPHOTON.2010.214}{10.1038/NPHOTON.2010.214}}.

\bibitem[Mal00]{malkin00}
T.~Malkin.
\newblock {\em {A Study of Secure Database Access and General Two-Party
  Computation}}.
\newblock PhD thesis, Massachusetts Institute of Technology, 2000.

\bibitem[Mau91]{maurer91}
U.~M. Maurer.
\newblock {A Provably-Secure Strongly-Randomized Cipher}.
\newblock {\em Advances in Cryptology: Proc. EUROCRYPT '90, LNCS}, 473, 1991.
\newblock \\
  \texttt{DOI:\,\href{http://dx.doi.org/10.1007/3-540-46877-3\_33}{10.1007/3-540-46877-3\_33}}.

\bibitem[May97]{mayers97}
D.~Mayers.
\newblock {Unconditionally secure quantum bit commitment is impossible}.
\newblock {\em Phys. Rev. Lett.}, 78(17), 1997.
\newblock \\
  \texttt{DOI:\,\href{http://dx.doi.org/10.1103/PhysRevLett.78.3414}{10.1103/PhysRevLett.78.3414}}.

\bibitem[MHH{\etalchar{+}}97]{muller97}
A.~Muller, T.~Herzog, B.~Huttner, W.~Tittel, H.~Zbinden, and N.~Gisin.
\newblock {``Plug and play'' systems for quantum cryptography}.
\newblock {\em App. Phys. Lett.}, 70(793), 1997.
\newblock \\
  \texttt{DOI:\,\href{http://dx.doi.org/10.1063/1.118224}{10.1063/1.118224}}.

\bibitem[MM07]{mullen07}
G.~L. Mullen and C.~Mummert.
\newblock {\em {Finite Fields and Applications}}.
\newblock Amer. Math. Soc., 2007.

\bibitem[Moc07]{mochon07}
C.~Mochon.
\newblock {Quantum weak coin flipping with arbitrarily small bias}.
\newblock 2007.
\newblock \\ \texttt{arXiv:\,\href{http://arxiv.org/abs/0711.4114}{0711.4114}}.

\bibitem[MU88]{maassen88}
H.~Maassen and J.~B.~M. Uffink.
\newblock {Generalized Entropic Uncertainty Relations}.
\newblock {\em Phys. Rev. Lett.}, 60(12), 1988.
\newblock \\
  \texttt{DOI:\,\href{http://dx.doi.org/10.1103/PhysRevLett.60.1103}{10.1103/PhysRevLett.60.1103}}.

\bibitem[NC00]{nielsen00}
M.~A. Nielsen and I.~L. Chuang.
\newblock {\em {Quantum Computation and Quantum Information}}.
\newblock Cambridge University Press, 2000.

\bibitem[NP00]{naor00}
M.~Naor and B.~Pinkas.
\newblock {Distributed Oblivious Transfer}.
\newblock {\em Advances in Cryptology: Proc. ASIACRYPT '00, LNCS}, 2000.
\newblock \\
  \texttt{DOI:\,\href{http://dx.doi.org/10.1007/3-540-44448-3\_16}{10.1007/3-540-44448-3\_16}}.

\bibitem[NPA07]{navascues07}
M.~Navascu{\'e}s, S.~Pironio, and A.~Ac{\'{\i}}n.
\newblock {Bounding the set of quantum correlations}.
\newblock {\em Phys. Rev. Lett.}, 98(1), 2007.
\newblock \\
  \texttt{DOI:\,\href{http://dx.doi.org/10.1103/PhysRevLett.98.010401}{10.1103/PhysRevLett.98.010401}}.

\bibitem[PR94]{popescu94}
S.~Popescu and D.~Rohrlich.
\newblock {Quantum Nonlocality as an Axiom}.
\newblock {\em Foundations of Physics}, 24(3), 1994.
\newblock \\
  \texttt{DOI:\,\href{http://dx.doi.org/10.1007/BF02058098}{10.1007/BF02058098}}.

\bibitem[RG15]{ribeiro15}
J.~Ribeiro and F.~Grosshans.
\newblock {A Tight Lower Bound for the BB84-states
  Quantum-Position-Verification Protocol}.
\newblock 2015.
\newblock \\
  \texttt{arXiv:\,\href{http://arxiv.org/abs/1504.07171}{1504.07171}}.

\bibitem[Riv99]{rivest99}
R.~L. Rivest.
\newblock {Unconditionally Secure Commitment and Oblivious Transfer Schemes
  Using Private Channels and a Trusted Initializer}.
\newblock 1999.
\newblock \\ Online: \url{http://people.csail.mit.edu/rivest/pubs/Riv99d.pdf}.

\bibitem[RKKM14]{radchenko14}
I.~Radchenko, K.~Kravtsov, S.~Kulik, and S.~N. Molotkov.
\newblock {Relativistic quantum cryptography}.
\newblock {\em Laser Phys. Lett.}, 11, 2014.
\newblock \\
  \texttt{DOI:\,\href{http://dx.doi.org/10.1088/1612-2011/11/6/065203}{10.1088/1612-2011/11/6/065203}}.

\bibitem[RUV13]{reichardt13}
B.~W. Reichardt, F.~Unger, and U.~Vazirani.
\newblock {Classical command of quantum systems}.
\newblock {\em Nature}, 496(7446), 2013.
\newblock \\
  \texttt{DOI:\,\href{http://dx.doi.org/10.1038/nature12035}{10.1038/nature12035}}.

\bibitem[RV13]{regev13}
O.~Regev and T.~Vidick.
\newblock {Quantum XOR games}.
\newblock {\em Proc. IEEE Comput. Comp. '13}, 2013.
\newblock \\
  \texttt{DOI:\,\href{http://dx.doi.org/10.1109/CCC.2013.23}{10.1109/CCC.2013.23}}.

\bibitem[Sal98]{salvail98}
L.~Salvail.
\newblock {Quantum Bit Commitment From a Physical Assumption}.
\newblock {\em Advances in Cryptology: Proc. CRYPTO '98, LNCS}, 1462, 1998.
\newblock \\
  \texttt{DOI:\,\href{http://dx.doi.org/10.1007/BFb0055740}{10.1007/BFb0055740}}.

\bibitem[Sca12]{scarani12}
V.~Scarani.
\newblock {The device-independent outlook on quantum physics}.
\newblock {\em Acta Phys. Slov.}, 62(4), 2012.
\newblock \\ \texttt{arXiv:\,\href{http://arxiv.org/abs/1303.3081}{1303.3081}}.

\bibitem[Sch95]{schumacher95}
B.~Schumacher.
\newblock {Quantum coding}.
\newblock {\em Phys. Rev. A}, 51(4), 1995.
\newblock \\
  \texttt{DOI:\,\href{http://dx.doi.org/10.1103/PhysRevA.51.2738}{10.1103/PhysRevA.51.2738}}.

\bibitem[Sch10]{schaffner10}
C.~Schaffner.
\newblock {Simple protocols for oblivious transfer and secure identification in
  the noisy-quantum-storage model}.
\newblock {\em Phys. Rev. A}, 82(3), 2010.
\newblock \\
  \texttt{DOI:\,\href{http://dx.doi.org/10.1103/PhysRevA.82.032308}{10.1103/PhysRevA.82.032308}}.

\bibitem[SCK14]{sikora14}
J.~Sikora, A.~Chailloux, and I.~Kerenidis.
\newblock {Strong connections between quantum encodings, nonlocality, and
  quantum cryptography}.
\newblock {\em Phys. Rev. A}, 89(2), 2014.
\newblock \\
  \texttt{DOI:\,\href{http://dx.doi.org/10.1103/PhysRevA.89.022334}{10.1103/PhysRevA.89.022334}}.

\bibitem[Sha48]{shannon48}
C.~E. Shannon.
\newblock {A Mathematical Theory of Communication}.
\newblock {\em Bell System Technical Journal}, 27(3-4), 1948.
\newblock \\
  \texttt{DOI:\,\href{http://dx.doi.org/10.1002/j.1538-7305.1948.tb01338.x}{10.1002/j.1538-7305.1948.tb01338.x}}.

\bibitem[Sha49]{shannon49}
C.~E. Shannon.
\newblock {Communication Theory of Secrecy Systems}.
\newblock {\em Bell System Technical Journal}, 28(4), 1949.
\newblock \\
  \texttt{DOI:\,\href{http://dx.doi.org/10.1002/j.1538-7305.1949.tb00928.x}{10.1002/j.1538-7305.1949.tb00928.x}}.

\bibitem[Sho94]{shor94}
P.~W. Shor.
\newblock {Algorithms for Quantum Computation: Discrete Logarithms and
  Factoring}.
\newblock {\em Proc. 35th ACM STOC}, 1994.
\newblock \\
  \texttt{DOI:\,\href{http://dx.doi.org/10.1109/SFCS.1994.365700}{10.1109/SFCS.1994.365700}}.

\bibitem[Sim07]{simard07}
J.-R. Simard.
\newblock {\em {Classical and Quantum Strategies for Bit Commitment Schemes in
  the Two-Prover Model}}.
\newblock Master's thesis, McGill University, 2007.
\newblock \\ Online:
  \url{http://crypto.cs.mcgill.ca/~crepeau/PDF/memoire-JR.pdf}.

\bibitem[SR01]{spekkens01}
R.~W. Spekkens and T.~Rudolph.
\newblock {Degrees of concealment and bindingness in quantum bit commitment
  protocols}.
\newblock {\em Phys. Rev. A}, 65, 2001.
\newblock \\
  \texttt{DOI:\,\href{http://dx.doi.org/10.1103/PhysRevA.65.012310}{10.1103/PhysRevA.65.012310}}.

\bibitem[SR02]{spekkens02}
R.~W. Spekkens and T.~Rudolph.
\newblock {Quantum Protocol for Cheat-Sensitive Weak Coin Flipping.}
\newblock {\em Phys. Rev. Lett.}, 89(22), 2002.
\newblock \\
  \texttt{DOI:\,\href{http://dx.doi.org/10.1103/PhysRevLett.89.227901}{10.1103/PhysRevLett.89.227901}}.

\bibitem[TFKW13]{tomamichel13}
M.~Tomamichel, S.~Fehr, J.~Kaniewski, and S.~Wehner.
\newblock {A monogamy-of-entanglement game with applications to
  device-independent quantum cryptography}.
\newblock {\em New J. Phys.}, 15(10), 2013.
\newblock \\
  \texttt{DOI:\,\href{http://dx.doi.org/10.1088/1367-2630/15/10/103002}{10.1088/1367-2630/15/10/103002}}.

\bibitem[Tom12]{tomamichel12}
M.~Tomamichel.
\newblock {\em {A Framework for Non-Asymptotic Quantum Information Theory}}.
\newblock PhD thesis, ETH Zurich, 2012.
\newblock \\ \texttt{arXiv:\,\href{http://arxiv.org/abs/1203.2142}{1203.2142}}.

\bibitem[TR11]{tomamichel11}
M.~Tomamichel and R.~Renner.
\newblock {Uncertainty Relation for Smooth Entropies}.
\newblock {\em Phys. Rev. Lett.}, 106(11), 2011.
\newblock \\
  \texttt{DOI:\,\href{http://dx.doi.org/10.1103/PhysRevLett.106.110506}{10.1103/PhysRevLett.106.110506}}.

\bibitem[Tsi80]{tsirelson80}
B.~S. Tsirelson.
\newblock {Quantum generalizations of Bell's inequality}.
\newblock {\em Lett. Math. Phys.}, 4(2), 1980.
\newblock \\
  \texttt{DOI:\,\href{http://dx.doi.org/10.1007/BF00417500}{10.1007/BF00417500}}.

\bibitem[Uhl76]{uhlmann76}
A.~Uhlmann.
\newblock {The ``transition probability'' in the state space of a *-algebra}.
\newblock {\em Rep. Math. Phys.}, 9(2), 1976.
\newblock \\
  \texttt{DOI:\,\href{http://dx.doi.org/10.1016/0034-4877(76)90060-4}{10.1016/0034-4877(76)90060-4}}.

\bibitem[Unr14]{unruh14}
D.~Unruh.
\newblock {Quantum Position Verification in the Random Oracle Model}.
\newblock {\em Advances in Cryptology: Proc. CRYPTO '14, LNCS}, 8617, 2014.
\newblock \\
  \texttt{DOI:\,\href{http://dx.doi.org/10.1007/978-3-662-44381-1\_1}{10.1007/978-3-662-44381-1\_1}}.

\bibitem[Vid13]{vidick13}
T.~Vidick.
\newblock {Three-player entangled XOR games are NP-hard to approximate}.
\newblock {\em Proc. 54th IEEE FOCS}, 2013.
\newblock \\
  \texttt{DOI:\,\href{http://dx.doi.org/10.1109/FOCS.2013.87}{10.1109/FOCS.2013.87}}.

\bibitem[Weh06]{wehner06}
S.~Wehner.
\newblock {Tsirelson bounds for generalized Clauser-Horne-Shimony-Holt
  inequalities}.
\newblock {\em Phys. Rev. A}, 73, 2006.
\newblock \\
  \texttt{DOI:\,\href{http://dx.doi.org/10.1103/PhysRevA.73.022110}{10.1103/PhysRevA.73.022110}}.

\bibitem[Wer89]{werner89}
R.~F. Werner.
\newblock {Quantum states with Einstein-Podolsky-Rosen correlations admitting a
  hidden-variable model}.
\newblock {\em Phys. Rev. A}, 40(8), 1989.
\newblock \\
  \texttt{DOI:\,\href{http://dx.doi.org/10.1103/PhysRevA.40.4277}{10.1103/PhysRevA.40.4277}}.

\bibitem[Wie83]{wiesner83}
S.~Wiesner.
\newblock {Conjugate coding}.
\newblock {\em ACM SIGACT News}, 15(1), 1983.
\newblock \\
  \texttt{DOI:\,\href{http://dx.doi.org/10.1145/1008908.1008920}{10.1145/1008908.1008920}}.

\bibitem[Wil13]{wilde13}
M.~M. Wilde.
\newblock {\em {Quantum Information Theory}}.
\newblock Cambridge University Press, 2013.

\bibitem[WNI03]{winter03}
A.~Winter, A.~C.~A. Nascimento, and H.~Imai.
\newblock {Commitment Capacity of Discrete Memoryless Channels}.
\newblock {\em Cryptography and Coding, LNCS}, 2898, 2003.
\newblock \\
  \texttt{DOI:\,\href{http://dx.doi.org/10.1007/978-3-540-40974-8\_4}{10.1007/978-3-540-40974-8\_4}}.

\bibitem[WST08]{wehner08a}
S.~Wehner, C.~Schaffner, and B.~Terhal.
\newblock {Cryptography from Noisy Storage}.
\newblock {\em Phys. Rev. Lett.}, 100(22), 2008.
\newblock \\
  \texttt{DOI:\,\href{http://dx.doi.org/10.1103/PhysRevLett.100.220502}{10.1103/PhysRevLett.100.220502}}.

\bibitem[WTHR11]{winkler11}
S.~Winkler, M.~Tomamichel, S.~Hengl, and R.~Renner.
\newblock {Impossibility of growing quantum bit commitments}.
\newblock {\em Phys. Rev. Lett.}, 107(9), 2011.
\newblock \\
  \texttt{DOI:\,\href{http://dx.doi.org/10.1103/PhysRevLett.107.090502}{10.1103/PhysRevLett.107.090502}}.

\bibitem[WWW11]{winkler11a}
S.~Winkler, J.~Wullschleger, and S.~Wolf.
\newblock {Bit Commitment From Nonsignaling Correlations}.
\newblock {\em IEEE Trans. Inf. Theory}, 57(3), 2011.
\newblock \\
  \texttt{DOI:\,\href{http://dx.doi.org/10.1109/TIT.2011.2104471}{10.1109/TIT.2011.2104471}}.

\bibitem[W{\.Z}82]{wootters82}
W.~K. Wootters and W.~H. {\.Z}urek.
\newblock {A single quantum cannot be cloned}.
\newblock {\em Nature}, 299(5886), 1982.
\newblock \\
  \texttt{DOI:\,\href{http://dx.doi.org/10.1038/299802a0}{10.1038/299802a0}}.

\bibitem[Yao82]{yao82}
A.~C.-C. Yao.
\newblock {Protocols for secure computations}.
\newblock {\em Proc. 23rd IEEE FOCS}, 1982.
\newblock \\
  \texttt{DOI:\,\href{http://dx.doi.org/10.1109/SFCS.1982.38}{10.1109/SFCS.1982.38}}.

\bibitem[Yao95]{yao95}
A.~C.-C. Yao.
\newblock {Security of Quantum Protocols Against Coherent Measurements}.
\newblock {\em Proc. 27th ACM STOC}, 1995.
\newblock \\
  \texttt{DOI:\,\href{http://dx.doi.org/10.1145/225058.225085}{10.1145/225058.225085}}.

\end{thebibliography}
\appendix
\chapter{Classical certification of relativistic bit commitment}
\label{app:classical-certification}
In 1992 Bennett et al.~proposed how to construct oblivious transfer by combining bit commitment with quantum communication \cite{bennett92}. This was later formalised and proven secure by Yao \cite{yao95}, who refers to it as ``the canonical construction''. At that point the quantum bit commitment protocol by Brassard et al.~\cite{brassard93} was considered secure so the canonical construction uses it as a black box. While quantum bit commitment was later proven impossible, the canonical construction remains interesting because classically we do not know whether bit commitment can be used to implement oblivious transfer.

Attempts to use relativistic commitment schemes in the canonical construction led to some interesting insight. It turns out that the construction implicitly assumes a certain property of the commitment scheme known as \emph{classical certification}, which in the quantum world should not be taken for granted. In fact, Kent showed that classical certification is generally impossible in case of quantum protocols \cite{kent12b}.

In this appendix we show that classical certification is not determined solely by the protocol but depends also on the exact power of the adversary. More specifically, we show that the (classical relativistic) \texttt{sBGKW} scheme is classically-certifiable against classical adversaries but not against quantum adversaries. We discuss why lack of classical certification completely breaks the canonical construction and present an explicit cheating strategy.
\section{Classical certification of the \texttt{sBGKW} scheme}
Classical certification should be thought of as a stronger variant of the binding property and to make this connection clear we call it the \emph{strongly-binding property}. Since the goal here is to flesh out the difference between the two notions and we already have a particular example in mind, we use definitions tailored to that concrete scenario. For completeness, let us restate the protocol first.
\sBGKWnc
\noindent Note that no communication is required between the commitment point and the opening point (so the two are operationally equivalent) and that only \alice{2} is involved in the open phase.
\begin{df}
Let $\sigma_{ A_{1} A_{2} B }$ be a state that \alice{1} and \alice{2} can enforce at the commitment point and let $(\Phi^{0}_{A_{2} \to P}, \Phi^{1}_{A_{2} \to P})$ be opening maps performed by \alice{2}.

A multiagent bit commitment protocol is \textbf{$\varepsilon$-binding} if for all states and all maps we have
\begin{equation*}
p_{0} + p_{1} \leq 1 + \varepsilon,
\end{equation*}
where
\begin{equation*}
p_{\cval} = \tr \big( M_{\textnormal{accept}} \big[ \Phi^{\cval}_{A_{2} \to P} (\sigma_{ A_{2} B }) \big] \big).
\end{equation*}

A multiagent bit commitment protocol is \textbf{$\varepsilon$-strongly-binding} if every state $\sigma_{ A_{1} A_{2} B }$ can be supplemented by a binary random variable $D$
\begin{equation*}
\sigma_{ D A_{1} A_{2} B } = \ketbraq{0}_{D} \otimes \sigma_{ A_{1} A_{2} B }^{0} + \ketbraq{1}_{D} \otimes \sigma_{ A_{1} A_{2} B }^{1}
\end{equation*}
such that the subnormalised states $\sigma_{ A_{1} A_{2} B }^{d}$ satisfy
\begin{equation}
\label{eq:unveil-epsilon}
\tr \big( M_{\textnormal{accept}} \big[ \Phi^{\cval}_{A_{2} \to P} ( \sigma_{ A_{2} B }^{\cval} ) \big] \big) \leq \varepsilon
\end{equation}
for $d \in \{0, 1\}$.
\end{df}
\noindent Intuitively, the random variable $D$ tells us which value \alice{2} cannot unveil. More precisely, inequality~\eqref{eq:unveil-epsilon} states that the probability of $D = d$ \emph{and} Bob accepting the unveiling of $d$ is at most $\varepsilon$. As the random variable $D$ could be given to an external observer, it captures the notion that the commitment has an objective value, which is beyond Alice's influence. It is a simple exercise to show that every protocol which is $\varepsilon$-strongly-binding is also $\varepsilon$-binding.

In Chapter \ref{chap:multiround} we showed that Protocol \hyperref[prot:sBGKW-nc]{4} is $\varepsilon$-binding with
\begin{center}
\begin{tabular}{l c}
$\varepsilon = 2^{ -n }$ & (against classical adversaries),\\
$\varepsilon = \sqrt{2} \cdot 2^{ - n/2 }$ & (against quantum adversaries).
\end{tabular}
\end{center}
While there is clearly a quantitative difference, qualitatively the situation is the same: in both cases the protocol is secure and the security guarantee decays exponentially in $n$.

The situation turns out to be quite different when we consider the stronger definition. We start by showing that Protocol \hyperref[prot:sBGKW-nc]{4} is $\varepsilon$-strongly-binding with $\varepsilon = 2^{-n/2}$ against classical adversaries. However, we then prove that against quantum adversaries, the protocol does not satisfy the strongly-binding definition for any $\varepsilon < \frac{1}{4}$ regardless of how large $n$ is.

To show that the protocol is strongly-binding in the classical case we explicitly construct the random variable $D$. It suffices to provide a construction for deterministic strategies of Alice (any non-deterministic strategy can be written as a convex combination of deterministic strategies). The deterministic strategy of \alice{1} is a function $f : \{0, 1\}^{n} \to \bs{n}$, while for \alice{2} we have $g : \{0, 1\} \to \bs{n}$, where the argument of $g$ is the value $\cval$ that she is trying to unveil. Then, the condition for successfully unveiling $\cval$ becomes
\begin{equation*}
f(b) \oplus g(\cval) = \cval \cdot b.
\end{equation*}
\begin{prop}
Protocol \hyperref[prot:sBGKW-nc]{4} is $\varepsilon$-strongly-binding against classical adversaries with $\varepsilon = 2^{-n/2}$.
\end{prop}
\begin{proof}
Under the assumption that \alice{1} and \alice{2} behave deterministically the state of the protocol at the commitment point is completely described by two variables: $b$ and $f(b)$, which can be used to define the random variable $D$. For $c \in \bs{n}$ let
\begin{equation*}
\cS(c) := \{ b \in \bs{n} : f(b) = c \}
\end{equation*}
and
\begin{gather*}
\cT_{0} := \{ c \in \bs{n} : \abs{ \cS(c) } \leq 2^{n/2} \},\\
\cT_{1} := \{ c \in \bs{n} : \abs{ \cS(c) } > 2^{n/2} \}.
\end{gather*}
Note that $\abs{ \cT_{1} } < 2^{n/2}$ since
\begin{equation*}
2^{n} = \sum_{c} \abs{ \cS(c) } \geq \sum_{c \in \cT_{1}} \abs{ \cS(c) } > 2^{n/2} \sum_{c \in \cT_{1}} 1 = 2^{n/2} \abs{ \cT_{1} }.
\end{equation*}
We define the random variable $D$ as a deterministic function of $b$
\begin{equation*}
D :=
\begin{cases}
0 &\nbox{if} f(b) \in \cT_{0},\\
1 &\nbox{if} f(b) \in \cT_{1}.
\end{cases}
\end{equation*}
Now we check that this definition satisfies the conditions. For $D = 0$ we have
\setlength{\jot}{8pt}
\begin{align*}
\Pr[ D = 0 \wedge \textnormal{unveil 0} ] &= \Pr[ D = 0 \wedge f(b) \oplus g(0) = 0^{n} ]\\
&= 2^{-n} \abs{ \{ b \in \bs{n} : f(b) \in \cT_{0} \wedge f(b) = g(0) \} }\\
&= 2^{-n} \abs{ \{ b \in \bs{n} : g(0) \in \cT_{0} \wedge f(b) = g(0) \} }\\
&= 2^{-n} \abs[\big]{ \cS \big( g(0) \big) } \cdot I[ g(0) \in \cT_{0} ] \leq 2^{-n} 2^{n/2} = 2^{-n/2},
\end{align*}
where $I[\cdot]$ denotes the indicator function\footnote{The indicator function is defined to satisfy $I[\textnormal{true statement}] = 1$ and $I[\textnormal{false statement}] = 0$.}. For $D = 1$ we have
\begin{align*}
\Pr[ D = 1 \wedge \textnormal{unveil 1} ] &= \Pr[ D = 1 \wedge f(b) \oplus g(1) = b ]\\
&= 2^{-n} \abs{ \{ b \in \bs{n} : f(b) \in \cT_{1} \wedge f(b) \oplus g(1) = b \} }\\
&= 2^{-n} \sum_{c \in \cT_{1}} \abs{ \{ b \in \bs{n} : f(b) = c \wedge f(b) \oplus g(1) = b \} }\\
&\leq 2^{-n} \sum_{c \in \cT_{1}} \abs{ \{ b \in \bs{n} : c \oplus g(1) = b \} }\\
&\leq 2^{-n} \sum_{c \in \cT_{1}} 1 = 2^{-n} \abs{ \cT_{1} } < 2^{-n} \cdot 2^{n/2} = 2^{-n/2}. \qedhere
\end{align*}
\end{proof}
To prove that this stronger notion of security is not possible in the quantum case, we propose an explicit attack and show that the resulting state cannot be supplemented by an additional random variable satisfying the criteria.
\begin{prop}
Protocol \hyperref[prot:sBGKW-nc]{4} does not satisfy the $\varepsilon$-strongly-binding property against quantum adversaries for any $\varepsilon < \frac{1}{4}$.
\end{prop}
\begin{proof}
Suppose that at the beginning of the protocol \alice{1} and \alice{2} share the maximally entangled state of $2n$ qubits
\begin{equation*}
\ket{\Psi_{2^{n}}}_{A_{1} A_{2}} = 2^{-n/2} \sum_{x} \ket{x}_{A_{1}} \ket{x}_{A_{2}}.
\end{equation*}
Let $C$ be an auxiliary control register held by \alice{1}, initially prepared in the state $\ket{+}$
. When \alice{1} receives $b$, she applies the following unitary $U^{b}_{A_{1} C}$
\begin{equation}
\label{eq:unitary-ub}
\begin{aligned}
U^{b}_{A_{1} C} \ket{x}_{A_{1}} \ket{0}_{C} &= \ket{x}_{A_{1}} \ket{0}_{C},\\
U^{b}_{A_{1} C} \ket{x}_{A_{1}} \ket{1}_{C} &= \ket{x \oplus b}_{A_{1}} \ket{1}_{C}.
\end{aligned}
\end{equation}
Then, the tripartite state $\ket{\psi}_{ A_{1} A_{2} C }$ becomes
\begin{align*}
\ket{\psi}_{ A_{1} A_{2} C } &= ( U^{b}_{A_{1} C} \otimes \mathbb{1}_{A_{2}} ) \ket{\Psi_{2^{n}}}_{ A_{1} A_{2} } \ket{+}_{C}\\
&= 2^{-(n + 1)/2} \Big[ \sum_{x} \ket{x}_{A_{1}} \ket{x}_{A_{2}} \ket{0}_{C} + \sum_{x} \ket{x \oplus b}_{A_{1}} \ket{x}_{A_{2}} \ket{1}_{C} \Big]\\
&= 2^{-(n + 1)/2} \Big[ \sum_{x} \ket{x}_{A_{1}} \ket{x}_{A_{2}} \ket{0}_{C} + \sum_{x} \ket{x}_{A_{1}} \ket{x \oplus b}_{A_{2}} \ket{1}_{C} \Big]\\
&= 2^{-(n + 1)/2} \sum_{x} \ket{x}_{A_{1}} \Big[ \ket{x}_{A_{2}} \ket{0}_{C} + \ket{x \oplus b}_{A_{2}} \ket{1}_{C} \Big].
\end{align*}
Now, \alice{1} measures $A_{1}$ in the computational basis to obtain a classical random variable $X_{1}$ and it is easy to verify that the state $\sigma_{X_{1} A_{2} C}$ is
\begin{equation*}
\sigma_{X_{1} A_{2} C} = 2^{-n} \sum_{x} \ketbraq{x}_{X_{1}} \otimes \ketbraq{\alpha_{x, b}}_{A_{2} C},
\end{equation*}
where
\begin{equation*}
\ket{\alpha_{x, b}}_{A_{2} C} = \frac{1}{\sqrt{2}} \big( \ket{x}_{A_{2}} \ket{0}_{C} + \ket{x \oplus b}_{A_{2}} \ket{1}_{C} \big).
\end{equation*}
Recall that $b$ is drawn uniformly at random by Bob and we should explicitly include it in the state. The state $\sigma_{X_{1} B A_{2} C}$ represents a complete description of the state of the protocol at the commitment point
\begin{equation*}
\sigma_{X_{1} B A_{2} C} = 2^{-2n} \sum_{x, b} \ketbraq{x}_{X_{1}} \otimes \ketbraq{b}_{B} \otimes \ketbraq{\alpha_{x, b}}_{A_{2} C}.
\end{equation*}
Our goal now is to show that regardless of how we define the auxiliary random variable $D$, it will not meet the desired criteria. The most general form of the state with the additional random variable is
\begin{equation*}
\sigma_{D X_{1} B A_{2} C} = \sum_{d, x, b} \ketbraq{d}_{D} \otimes \ketbraq{x}_{X_{1}} \otimes \ketbraq{b}_{B} \otimes \sigma^{dxb}_{A_{2} C},
\end{equation*}
where $\sigma^{dxb}_{A_{2} C}$ are subnormalised quantum states. However, since tracing out $D$ must give us back $\sigma_{X_{1} B A_{2} C}$ and the states on $A_{2}C$ (conditional on particular values of $X_{1}$ and $B$) are pure, we conclude that for all $d, x, b$
\begin{equation*}
\sigma^{dxb}_{A_{2} C} \propto \ketbraq{\alpha_{x, b}}_{A_{2} C}.
\end{equation*}
Therefore, without loss of generality we can write
\begin{equation*}
\sigma_{D X_{1} B A_{2} C} = \sum_{d, x, b}  p_{dxb} \ketbraq{d}_{D} \otimes \ketbraq{x}_{X_{1}} \otimes \ketbraq{b}_{B} \otimes \ketbraq{\alpha_{x, b}}_{A_{2} C},
\end{equation*}
where $p_{dxb} = \Pr[ D = d, X_{1} = x, B = b ]$ is a probability distribution over $D, X_{1}$ and $B$. Now, we need to evaluate the probability of \alice{2} unveiling the commitment successfully. Since \alice{1} does not play a role in the open phase, we trace out subsystem $C$. The unveiling strategy of \alice{2} is simply to measure her subsystem in the computational basis (regardless of the value she is trying to unveil). If we represent her measurement outcome by $X_{2}$ we obtain the following (fully classical) state
\begin{equation*}
\sigma_{D X_{1} B X_{2}} = \sum_{d, x, b}  p_{dxb} \ketbraq{d}_{D} \otimes \ketbraq{x}_{X_{1}} \otimes \ketbraq{b}_{B} \otimes \frac{1}{2} \big( \ketbraq{x}_{X_{2}} + \ketbraq{x \oplus b}_{X_{2}} \big).
\end{equation*}
This allows us to evaluate the
\begin{align*}
\Pr[ D = d \wedge \textnormal{unveil d} ] &= \Pr[ D = d \wedge X_{1} \oplus X_{2} = d \cdot B ]\\
&= \sum_{xb} \frac{p_{dxb}}{2} \big( 1 + I[b = 0] \big)\\
&\geq \frac{1}{2} \sum_{xb} p_{dxb} = \frac{1}{2} \Pr[D = d].
\end{align*}
Since $\Pr[D = 0] + \Pr[D = 1] = 1$, we must have $\Pr[ D = d \wedge \textnormal{unveil d} ] \geq \frac{1}{4}$ for at least one value of $d$. Hence, the security requirement cannot hold for any $\varepsilon < \frac{1}{4}$.
\end{proof}
It is clear that register $C$ determines the commitment value and if $C$ was a classical register the protocol would satisfy the strongly-binding definition. However, since $C$ is a quantum register kept in a coherent superposition, it does not have a well-defined value. We cannot think of the value of the commitment as a classical random variable and no meaningful definition of $D$ is possible. In the next section we show that this feature makes the scheme unsuitable for the canonical construction.
\section{Consequences for the canonical construction}
The canonical construction requires Bob to generate random BB84 states and send them to Alice, who measures every incoming state in a random basis (computational or Hadamard). After all the states have been measured Bob announces the basis he used for every state. If Alice has followed the protocol she has learnt (on average) half of the (logical) bits. Note that Bob does not know which bits she has learnt.

An obvious problem with this construction comes from the fact that Alice can store the quantum states and only measure them once the basis information is available. In this way she will learn the entire string, which renders the scheme completely insecure.

To defeat this cheating strategy Alice is required to prove that she has really measured all the systems by making a certain commitment. More specifically, for every received BB84 state she is required to commit two bits: the basis used and the outcome observed.

Later Bob asks Alice to open commitments corresponding to a random subset of the rounds. He expects that whenever she claims to have measured in the correct basis, she should have obtained the correct outcome. Classical intuition tells us that if Alice succeeds on a randomly chosen subset, then the other commitments (with high probability) must also correspond to honest measurements. In that case we conclude that she has followed the protocol, the BB84 states have been measured so the attack described above is no longer a threat.

The classical intuition implicitly assumes that the commitments contain specific values, which are well-defined regardless of whether the commitment is ultimately opened or not. This is exactly the notion of classical certification, which is generally not satisfied in the quantum setting. Here, we show that if the $\texttt{sBGKW}$ scheme is used in the construction, there exists a quantum cheating strategy for Alice with the following two properties.
\begin{itemize}
\item If Alice is challenged to open the commitment, she can produce statistics indistinguishable from the honest execution of the protocol.
\item If the commitment is not opened and \alice{1} and \alice{2} are allowed to recombine their systems, they can recover the original BB84 states.
\end{itemize}
Let $\ket{\phi} = \alpha_{0} \ket{0} + \alpha_{1} \ket{1}$ be the state that \alice{1} has received from Bob and suppose she stores it in register $C$. \alice{1} picks the measurement basis $\theta \in \{0, 1\}$ uniformly at random and makes an honest commitment to it. If $\theta = 0$ she leaves $\ket{\phi}$ unchanged, if $\theta = 1$ she applies a Hadamard transform to it. For simplicity in the following argument we assume $\theta = 0$ (the case of $\theta = 1$ is analogous). For the second commitment (to the outcome of the measurement) \alice{1} follows the dishonest procedure outlined in the previous section and it is easy to verify that at the commitment point the state is
\begin{equation*}
\sigma_{X_{1} A_{2} C} = 2^{-n} \sum_{x} \ketbraq{x}_{X_{1}} \otimes \ketbraq{\alpha_{x, b}'}_{A_{2} C},
\end{equation*}
where
\begin{equation}
\label{eq:state-alpha-prime}
\ket{\alpha_{x, b}'}_{A_{2} C} = \alpha_{0} \ket{x}_{A_{2}} \ket{0}_{C} + \alpha_{1} \ket{x \oplus b}_{A_{2}} \ket{1}_{C}.
\end{equation}
If \alice{2} is challenged to unveil this round, she honestly unveils the basis information ($\theta = 0$), while for the second commitment she simply measures her system in the computational basis and obtains the correct string to unveil $d$ with probability at least $\abs{\alpha_{d}}^{2}$. It is easy to see that identical statistics would be obtained if \alice{1} made a measurement in the computational basis at the beginning and honestly committed to the classical outcome. Picking $\theta \in \{0, 1\}$ uniformly at random leads to statistics which is indistinguishable from honestly measuring in a random basis.

On the other hand, if the commitment is not opened we can still recover the original state. Conditional on $X_{1} = x$ and $B = b$ the state on $A_{2} C$ is $\ket{\alpha_{x, b}'}_{A_{2} C}$ given by Eq.~\eqref{eq:state-alpha-prime}. Note that $C$ is with \alice{1} while $A_{2}$ is with \alice{2} but if we bring the systems together we can apply a unitary $U^{b}_{A_{2} C}$ (as in Eq.~\eqref{eq:unitary-ub} except that it acts on $A_{2}$ instead of $A_{1}$) to obtain
\begin{equation*}
U^{b}_{ A_{2} C } \ket{\alpha_{x, b}'}_{A_{2} C} = \alpha_{0} \ket{x}_{A_{2}} \ket{0}_{C} + \alpha_{1} \ket{x}_{A_{2}} \ket{1}_{C} = \ket{x}_{A_{2}} \otimes \ket{\phi}_{C}.
\end{equation*}
Therefore, we have recovered the original state.

It is interesting to note that in this procedure the state $\ket{\phi}$ received by \alice{1} becomes ``delocalised'' between \alice{1} and \alice{2} at the commitment point. In other words, \alice{1} cannot recover it by acting on her own subsystem alone. This is not surprising as \alice{1} through the procedure has in some sense allowed \alice{2} to remotely perform a measurement on $\ket{\phi}$ so in order to reconstruct it we must combine the two systems together.

Since in the relativistic setting all communication constraints are temporary we cannot prohibit \alice{1} and \alice{2} from recovering the original state at some later point and once the basis information is available, they will perform the right measurement to obtain the correct outcome. Therefore, no uncertainty can be guaranteed on the rounds that have not been opened, which renders the construction completely insecure.
\printindex
\end{document}